\documentclass[11pt]{article}
\usepackage{color}
\usepackage{amssymb}   
\usepackage{amsthm}    
\usepackage{amsmath}   
\usepackage{stmaryrd}  
\usepackage{titletoc}  
\usepackage{mathrsfs}  
\usepackage{graphicx}
\usepackage{enumitem}
\usepackage[T1]{fontenc}
\usepackage{empheq}
\usepackage{todonotes}
\usepackage[symbol*]{footmisc}

\DeclareMathOperator*{\argmax}{arg\,max}
\DeclareMathOperator*{\argmin}{arg\,min}

\newlength{\defbaselineskip}
\newcommand{\setlinespacing}[1]%
           {\setlength{\baselineskip}{#1 \defbaselineskip}}

\theoremstyle{plain}
\newtheorem{thm}{Theorem}[section]

\newtheorem{lem}[thm]{Lemma}
\newtheorem{prop}[thm]{Proposition}

\theoremstyle{definition}

\newtheorem{rmk}{Remark}[section]

\newcommand{\eps}{\varepsilon}

\DeclareMathOperator*{\esssup}{esssup}

\newcommand{\cU}{\mathcal{U}}

\newcommand{\bP}{\mathbb{P}}
\newcommand{\bR}{\mathbb{R}}
\newcommand{\bN}{\mathbb{N}}

\newcommand{\sF}{\mathscr{F}}

\newcommand{\sL}{\mathscr{L}}

\makeatletter\@addtoreset{equation}{section} \makeatother
 \allowdisplaybreaks
\begin{document}

\title{Stochastic Path-Dependent Volatility Models for Price-Storage Dynamics in Natural Gas Markets and Discrete-Time Swing Option Pricing\footnotemark[1] 
}

\author{Jinniao Qiu\footnotemark[2]  \and Antony Ware\footnotemark[2]  \and Yang Yang\footnotemark[2]
}
\footnotetext[1]{This work was partially supported by the National Science and Engineering Research Council of Canada (NSERC). Yang was partially supported by a graduate scholarship through the NSERC-CREATE Program on Machine Learning in Quantitative Finance and Business Analytics (Fin-ML CREATE). Yang gratefully acknowledges support from DFG CRC/TRR 388 ``Rough Analysis, Stochastic Dynamics and Related Fields", Project B02. The authors also acknowledge the support of the Banff International Research Station (BIRS) for the Workshop [23w2004] “Stochastic Modelling of Big Data in Finance, Insurance and Energy Markets”, May 19-21, 2023, where part of this work was done.}
\footnotetext[2]{Department of Mathematics \& Statistics, University of Calgary, 2500 University Drive NW, Calgary, AB T2N 1N4, Canada. \textit{E-mail}: \texttt{jinniao.qiu@ucalgary.ca} (J. Qiu),  \texttt{aware@ucalgary.ca} (T. Ware),  \texttt{yang.yang1@ucalgary.ca} (Y. Yang)
.}

\maketitle

\begin{abstract}
This paper is devoted to the price-storage dynamics in natural gas markets. A novel stochastic path-dependent volatility model is introduced with path-dependence in both price volatility and storage increments. Model calibrations are conducted for both the price and storage dynamics. Further, we discuss the pricing problem of discrete-time swing options using the dynamic programming principle, and a deep learning-based method is proposed for numerical approximations. A numerical algorithm is provided, followed by a convergence analysis result for the deep-learning approach.
\end{abstract}



\section{Introduction}


Let $(\Omega,\sF,\{\sF_t\}_{t\ge 0},\bP)$ be a complete filtered probability space. The filtration $\{\sF_t\}_{t\ge 0}$ is generated by a one dimensional Wiener process $W = \{W(t): t\in [0,\infty]\}$ together with all the $\bP$-null sets in $\sF$ and it satisfies the usual conditions. 
We consider a day-ahead market for natural gas with price $S(t)$ for each $t\ge 0$, and define $\overline S(t)=\ln S(t)$ with the following stochastic path-dependent model:
\begin{equation}\label{ori_pro}
\begin{cases}
&d\overline S(t) = (\mu(t)-\lambda \overline S(t))dt + \sigma(t)dW(t), \quad S(0) \text{ is given,}
\\
&\sigma(t) = f(t,X(t),h((\overline S)_t)+\overline S(t)-\overline S(0)).
\end{cases}
\end{equation}
Here $(\mu(t))_{t\geq 0}$ is a stochastic process associated with the \emph{mean-revertion} level of the log-price, and the process $(X(t))_{t\geq 0}$ is a real-valued stochastic process tracking the deviation of the storage level from its periodic curve. $f$ and $h$ are, in general, real-valued continuous functions. Moreover, $(\overline S)_t:= (\overline S(u))_{0\le u\le t}$ denotes the log-price path up to time $t$, while $\sigma:=(\sigma(t))_{t\geq 0}$ characterizes the path-dependent volatility process. 


Understanding the drivers of natural gas prices is of significant interest to many economic agents. The growing role of natural gas as an energy source is reflected by an increasing interest in trying to explain the dynamics of its prices. For many years, fuel switching between natural gas and residual fuel oil kept natural gas prices closely aligned with those for crude oil. As shown by Y\"ucel and Guo in \cite{vucel1994fuel}, crude oil prices were shaped by world oil market conditions, and U.S. natural gas prices adjusted to oil prices. However, in recent times, as shown by \color{black}\cite{brown2008drives}\color{black}, U.S. natural gas prices have been on an upward trend with crude oil prices but with considerably independent movements. Natural gas analysts generally rely on weather and inventories as main drivers of natural gas prices. A subsequent work \cite{brown2008drives} by Brown and Y\"ucel  demonstrated that crude oil and natural gas prices still maintain a powerful relationship, but this relationship is conditioned by weather, seasonality, natural gas storage, and shut in production in the Gulf of Mexico. Rubaszek and Uddin further investigated the role of storage for the dynamics of the U.S. natural gas market in \cite{rubaszek2020role}. They concluded that spot prices were more responsive to economic fundamentals when storage was low, and that the level of underground natural gas storage had a significant impact on the relationship between spot and future prices. More work on the price-storage dynamics may be found in \cite{gouel2017estimating,mu2007weather,thompson2009natural}, to name a few. 

In the field of volatility modeling, the financial industry has predominantly employed four main categories of models: constant volatility, for instance the renowned Black-Scholes model \color{black} \cite{capinski2012black}\color{black} ; local volatility model (LV) (Dupire \cite{dupire1994pricing}); stochastic volatility (SV) (Hull and White \cite{hull1987pricing}, Heston \cite{heston1993closed}, Bergomi \cite{bergomi2005smile}, among many others); path-dependent volatility (PDV) (Guyon \cite{guyon2014path}, Guyon and Lekeufack \cite{guyon2021volatility}). In comparison to the classic constant volatility models, the LV models are superior in their flexibility to fit any abrbitrage-free implied volatility surface, i.e. the `smile', however no more flexibility is left. Meanwhile the SV models have an even richer dynamics capable of capturing key risk factors like forward skew and volatility of volatility. The PDV models combine benefits from both LV models and SV models. Not only can they fit market smile and produce rich implied volatility dynamics, but also generate a much broader range of spot-volatility, volatility-volatility of volatility dynamics. Hence the PDV models have better performance in capturing historical volatility patterns and preventing large mis-pricings. Notably in \cite{guyon2021volatility}, Guyon and Lekeufack uncovered a remarkably simple path-dependency: a linear combination of a weighted sum of past daily returns and the square root of a weighted sum of past daily square returns, employing different time-shifted power-law weights (which could be approximated by superpositions of exponential kernels to produce Markovian models), capturing both short and long memory. The proposed 4-factor Markovian PDV model captures crucial stylized facts of volatility, produces very realistic price and potentially rough volatility paths, and effectively fits both the SPX and VIX smiles. \color{black}Further, it is worthwhile to mention that Zumbach introduced a specific feedback of price returns of volatility in \cite{zumbach2010volatility}, where he proposed that the past trends in price convey significant information on future volatility. This property is later referred to as the \emph{Zumbach effect}. A follow-up work by Dandapani, Jusselin and Rosenbaum in \cite{dandapani2021quadratic} introduced a microscopic path-dependent model encoding the Zumbach effect under a quadratic Hawkes based price process. They showed that after suitable rescaling, the long term limits of these processes are refined versions of rough volatility models. In contrast to our work, their path-dependent volatility takes a quadratic form in which a storage factor is absent. \color{black}

In the realm of volatility models that encompass the concept of roughness, Gatheral, Jaisson and Rosenbaum delved into this subject in their work \cite{gatheral2022volatility}. They examine the smoothness of the volatility process by employing high-frequency S\&P data. Their focus centered on a rough fractional stochastic volatility (RFSV) model, and their findings revealed that the log-volatility behaves akin to a fractional Brownian motion with Hurst exponent $H$ of order 0.1, thereby leading to the conclusion that volatility inherently displays rough characteristics. Building up on this, Bayer, Friz and Gatheral extended the investigation in a subsequent paper \cite{bayer2016pricing} where they demonstrated that the RFSV model introduced by Gatheral et al. in their early work exhibited remarkable consistence with actual financial time series data. Additionally, their exploration extended to a rough Bergomi model, revealing its superior fit to the SPX volatility model in comparison to conventional Markovian stochastic volatility models. Notably, this enhanced fitting was achieved with fewer parameters. For further insights into the world of rough volatility models and the corresponding challenges in option pricing, refer to \cite{bayer2022pricing,chevalier2022american,fukasawa2019volatility}.

Our work focuses on a novel stochastic path-dependent volatility model and associated discrete-time swing option pricing in the natural gas market. Inspired by the early contributions of Guyon, we introduce the path-dependence via a power-kernel-based price moving average, thus endowing the price dynamics with non-Markovian properties. By adjusting \color{black}the parameters\color{black}, our model has the potential to encompass both the classic and rough volatility scenarios. The introduced price moving average not only serves as an indicator of the corresponding storage dynamics, capturing the essence of `buy low sell high' investment strategy, but also it captures key information of the ``roughness" to the price trajectories. This is different from Guyon's work in \cite{guyon2023volatility}. In particular, we differentiate between storage and volatility dynamics due to their distinct behaviours in financial contexts. Incorporating path-dependency into both price and storage dynamics empowers our model to not only capture intricate price-storage dynamics in the natural gas market, but also to effectively elucidate the interplay between storage and price volatility which is particularly evident when the storage nears full capacity or approaches depletion; this, to the best of our knowledge, is new.  Further, we address the pricing problems for discrete-time swing options with dynamic programming approaches. For approximating the involved conditional expectations, we utilize a deep learning-based method with neutral networks. A numerical algorithm is introduced and a convergence analysis conducted.

Before proceeding to a more detailed natural gas price-storage model, let us first consider some characteristics of natural gas prices and storage, as shown in Figure 1.   From the plots, we observe that
\begin{enumerate} [label=(\roman*)]
\item Natural gas prices are in general quite volatile and potentially contain "roughness".
\item There is strong seasonality in market storage in the natural gas market.
\item There is little visible periodicity in the evolution of the natural gas spot price or its increments.
\item In contrast to the highly volatile price process, the evolution of storage levels for natural gas is relatively smooth. However, the increments  do show more volatility.
\end{enumerate}
  These observations motivate us to consider a model allowing for possible ``roughness" in both the price volatility and storage increment processes. Our model recognizes different behaviours of market storage and price volatility since, for instance, they exhibit different regularity properties. Further, as the influence of periodicity on price volatilities is not directly observed,  a decomposition of the storage data is performed to better distinguish between the (log-)price-storage dynamics and the periodic behaviour of gas market storage. \color{black}It is worthwhile to mention that since seasonality is evident in many future curves of the gas markets as shown in \cite[Fig.3]{borovkova2006seasonal} and \cite{moreno2019long,chen2022stochastic}, among many others, it is natural to study the price-storage-volatility dynamics in the sense of future prices, for instance, the forward price. In this work we study one-day ahead prices, in line with \eqref{ori_pro}, although the model could be adapted to accommodate futures prices.

\begin{figure}
\begin{center}
\includegraphics[width=0.49\linewidth]{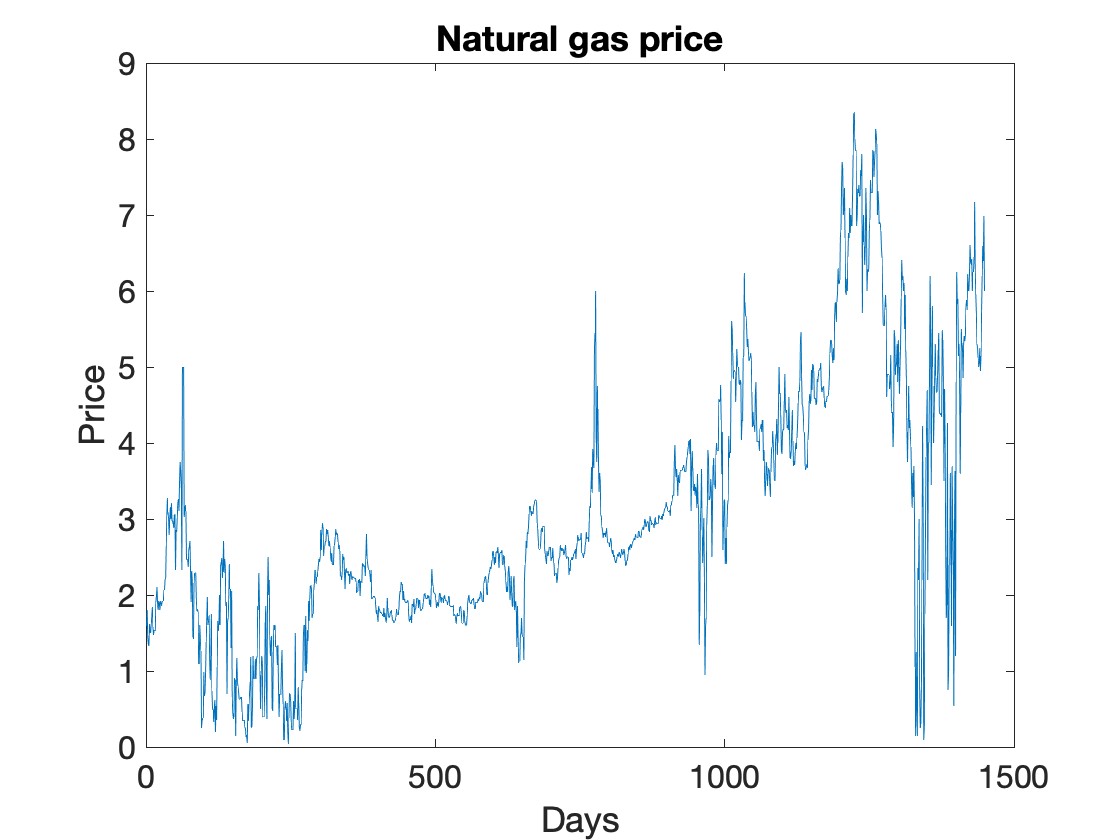}\label{price}
\includegraphics[width=0.49\linewidth]{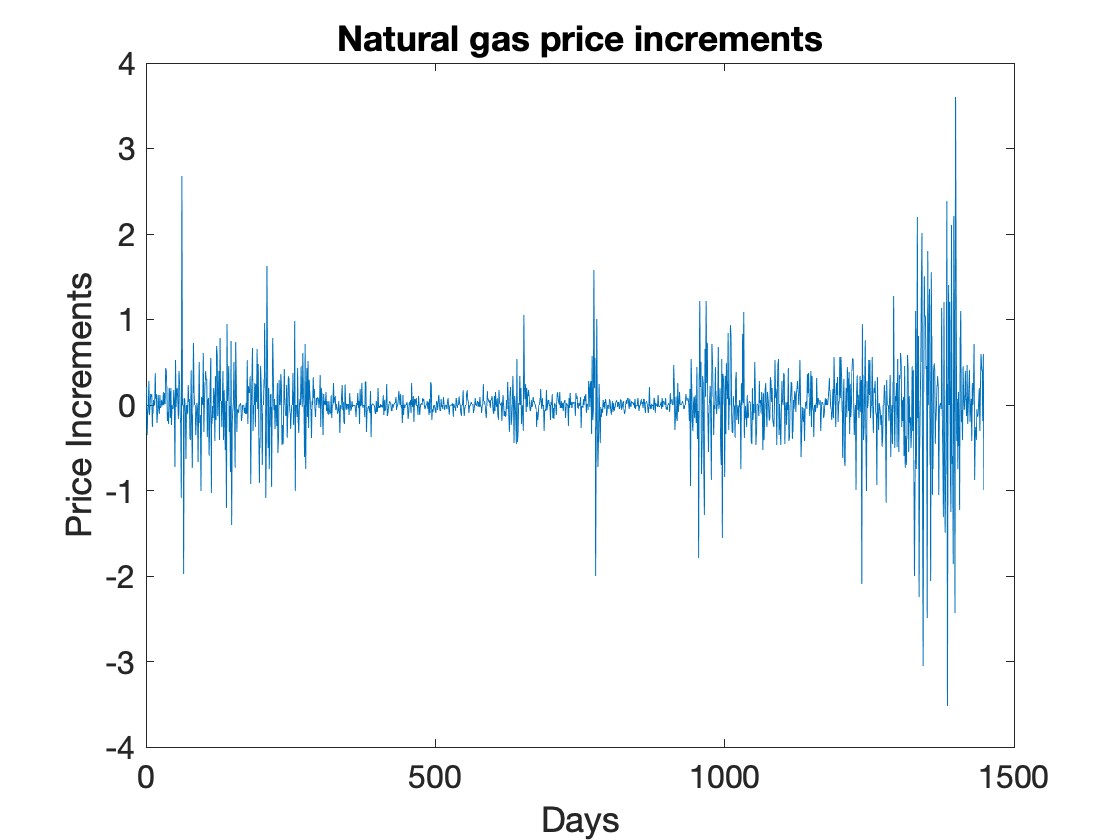}\label{price_inc}
\end{center}
\caption{NGX Phys FP [CA/GJ] TCPL - Empress - Daily and its Increments (Jan, 2019 - Dec, 2022)}
\label{fig:price_increment}
\end{figure}

\begin{figure}
\begin{center}
\includegraphics[width=0.49\linewidth]{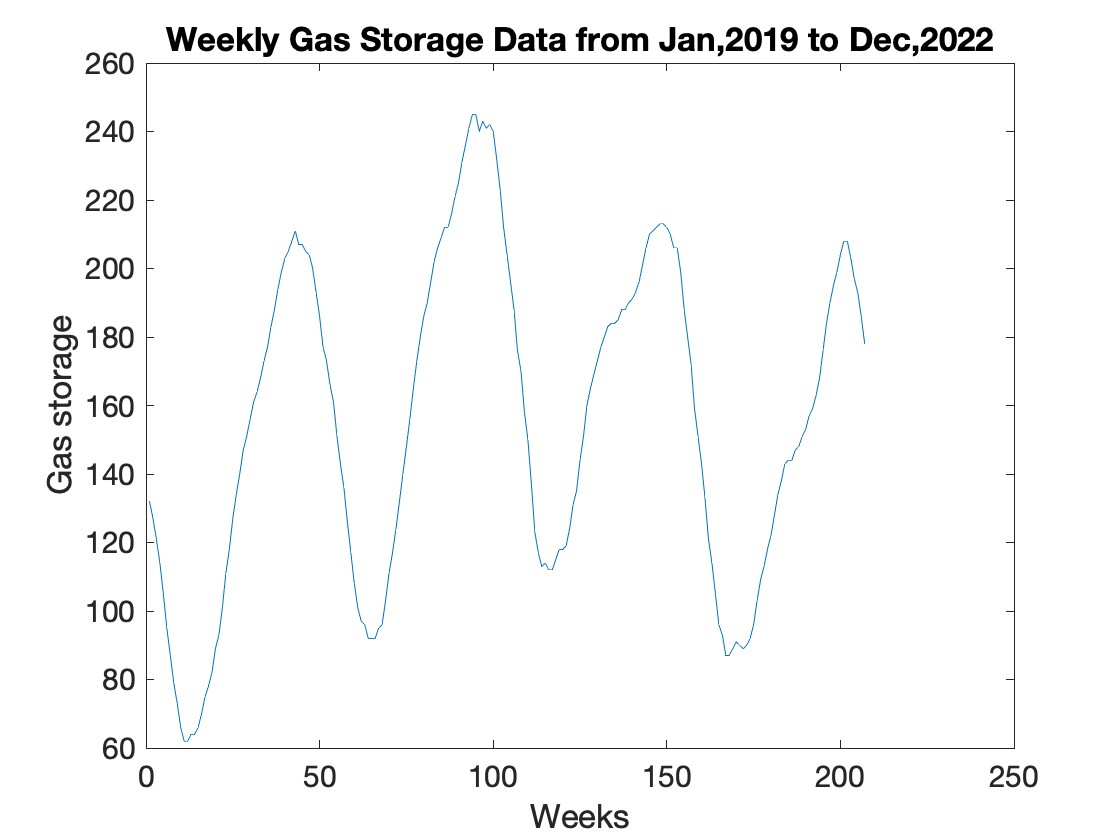}
\includegraphics[width=0.49\linewidth]{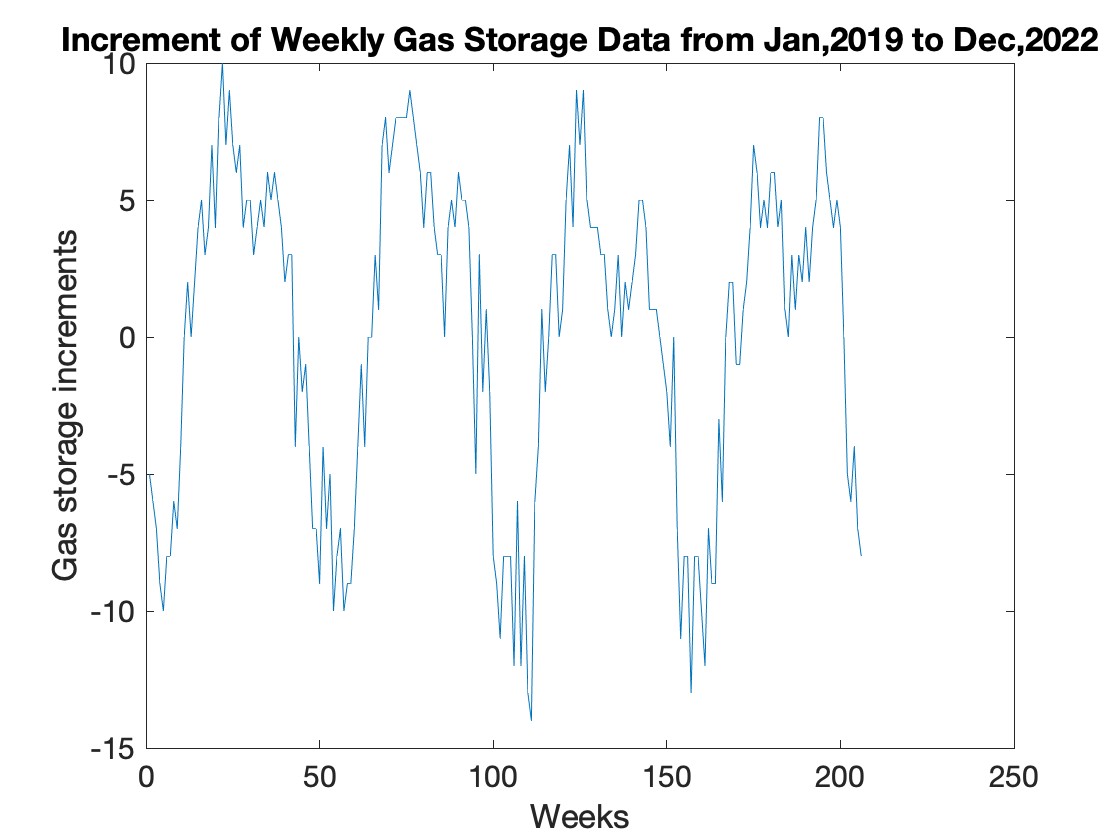}
\end{center}
\caption{Weekly Average Mountain Region Natural Gas Working Underground Storage and its Increments (Jan, 2019 - Dec, 2022)}
\label{fig:storage_increment}
\end{figure}

\color{black}
Our paper is organized as follows. In Section~2, we show some preliminary notations and introduce a specific kind of path-dependence. Then, incorporating a storage factor, we introduce a novel stochastic path-dependent volatility model, in which the path-dependence is embedded into both price volatility and storage dynamics. A discussion of the well-posedness of the model is included. In Section~3, we conduct a detailed calibration of both price and storage dynamics using a forward Euler-Maruyama approximation. We adopt log-likelihood function and mean square error to calibrate the log-price and storage dynamics, respectively. The optimal parameters are computed using a consensus-based optimization (CBO) method. Section~4 is devoted to discrete-time swing option contracts and associated pricing problems. A dynamic programming approach is adopted, with a neural network approximation scheme, and a numerical algorithm is presented along with a convergence analysis. \color{black}


\section{General Motivations}
\subsection{Source of path-dependence}
The model we work with in this paper is a special case of \eqref{ori_pro}, where we specify 
\[ \mu(t)=r-\frac{|\sigma(t)|^2}{2},\; t\geq 0,\]
so that the log-price process $\{\overline{S}(t)\}_{t\in [0,T]}$ satisfies 
\begin{align}
\overline S(t)=\overline S(0)+\int_0^t \left(r-\frac{|\sigma(u)|^2}{2} - \lambda\overline S(t)\right)du
+\int_0^t \sigma(u) dW(u), \label{eq-logPrice}
\end{align}
with $r,\lambda \ge 0$ as some constants. A possible mean-reverting behaviour is embedded into the log-price if $\lambda>0$. We note that we can also write
\begin{equation}
\label{eq-logPrice-exp}
    e^{\lambda t}\overline S(t)=\overline S(0)+\int_0^t e^{\lambda u}\left(r-\frac{|\sigma(u)|^2}{2}\right)du+\int_0^t e^{\lambda u}\sigma(u) dW(u).
\end{equation}

When $\lambda=0$, this coincides with the price dynamic satisfying a linear stochastic differential equation:
\begin{align}\label{linear_no_mr}
S(t) = S(0)+\int_0^trS(u)du + \int_0^t S(u)\sigma(u)dW(u).
\end{align}
To explore more about the source of path-dependence implied by $h(\overline{S}_t)$ as in \eqref{ori_pro}, let us, in the spirit of \cite{guyon2021volatility}, introduce a weighted moving average of $\{\overline{S}(t)\}_{t\in [0,T]}$
\begin{equation}
\overline{S}^{\alpha}_t := \int_0^t \! k^{\alpha}(t,u)\overline{S}(u)du \,\,\textbf{1}_{\{t>0\}} + \overline{S}(0)\textbf{1}_{\{t=0\}} ,\quad t\geq 0,\label{mov_aver}
\end{equation}
  where $\alpha\in (\frac{1}{2},1]$ is a constant parameter. For each $t>0$, the function $k^{\alpha}(t,\cdot): [0,t]\rightarrow [0,\infty]$ is an $\alpha$-parametrized weight kernel. 
Here, we choose
\begin{align}\label{kernel}
k^1(t,u)=2\delta_t(u),\quad
k^{\alpha}(t,u) = \frac{1-\alpha}{t^{1-\alpha}(t-u)^{\alpha}}, \text{ }\text{ }\alpha\in (\frac{1}{2},1), \quad t>0,
\end{align}
where for $t>0$, $\delta_t(\cdot)$ is the Dirac delta function with $\int_0^t 2\delta_t(u) g(u)\,du=g(t)$ for each continuous function $g$. Later (see \eqref{S-delta}), we will introduce an approximation to this kernel that allows the definition of $\overline{S}^{\alpha}_t$ to be extended to $\alpha>1$.

{\color{black}Note that, while the moving average $\overline{S}^\alpha(t)$ is only defined in \eqref{mov_aver} for $\alpha\leq 1$ when $k^\alpha$ is given by \eqref{kernel}, we will be interested in the \emph{difference} $\overline{S}_t^{\alpha} - \overline{S}(t)$, and this can a.s.~be well defined for a wider range of values of $\alpha$. We have the following lemma. 

}  
\begin{lem}\label{kernel1}
Suppose that the volatility process $\sigma$ in \eqref{eq-logPrice} has continuous trajectories and is square-integrable over $\Omega\times [0,T]$ for any $T>0$. Then for all $\alpha\in (\frac{1}{2},\frac{3}{2})$, the integral 
\begin{align}
h(\overline{S}_t):=(1-\alpha) \int_0^t \! \frac{\overline{S}(u) - \overline{S}(t)}{t^{1-\alpha}(t-u)^{\alpha}}du
\label{h}
\end{align}
is a.s. well defined for each $t\geq 0$. Moreover, for $\alpha\in (\frac12,1]$, $h(\overline{S}_t)$ can be identified with the relative log price level $\overline{S}^{\alpha}_t  -\overline S(t)$.
\end{lem}
 \begin{proof}
 When $\alpha\in (\frac{1}{2}, 1)$, the integral \eqref{h} is a.s. well defined for each $t\geq 0$, owing to the path-continuity of $\overline{S} $. If $\alpha=1$, we have $h\equiv 0$. Thus, our task narrows down to demonstrating the well-definedness of the integral \eqref{h}  when $\alpha\in(1, \frac{3}{2})$ as specified for the remainder of this proof.

 If the volatility process $\sigma$ has continuous trajectories and is square-integrable over $\Omega\times [0,T]$ for any $T>0$, one may check that the log-price process $\overline S$ admits $\beta$-H\"older continuous paths for any $\beta\in (0,\frac{1}{2})$ and thus the integral \eqref{h} is well defined. Indeed, recalling  the stochastic differential equation \eqref{eq-logPrice} for the log-price, and putting
 $$M(t)=\int_0^t e^{\lambda s}\sigma(s) dW(s),$$
we conclude from the continuity of $\sigma$ and from \eqref{eq-logPrice-exp} that almost surely
\begin{align}
|\overline S(t) - \overline S(u) | 
&=O(|t-u|) + |e^{-\lambda t}M(t)-e^{-\lambda t} M(u)| \nonumber\\
&\leq O(|t-u|) + e^{-\lambda t}|M(t)-M(u)| 
+ \sup_{s\in [0, t\vee u]} |M(s)| \cdot O(|t-u|). \label{rmk-kern-1}
\end{align}
The square-integrability of $\sigma$ implies that $(M(t))_{t\geq 0}$ is square-integrable continuous martingale, and by the representation of continuous martingales via time-changed Brownian motions (see \cite[Chapter V, Theorem (1.6)]{revuz1999continuous}), we have $M(t)=B(L(t))$ with $B(\cdot)$ being a standard Brownian motion and
$$L(t)=\int_0^t|e^{\lambda (s-t)}\sigma(s)|^2ds <\infty \quad \text{a.s.}$$
Therefore, for each $\beta\in (0,\frac{1}{2})$
\begin{align}
\sup_{0\leq u<t}\frac{|M(t)-M(u)|}{|t-u|^{\beta}}
&= \sup_{0\leq u<t} \frac{|B(L(t))-B(L(u))|}{\left(L(t)-L(u)\right)^{\beta}}\cdot \frac{\left(L(t)-L(u)\right)^{\beta}} {|t-u|^{\beta}}
\nonumber \\
&\leq \sup_{ 0\leq \tau< L(t)} \frac{|B(L(t))-B(\tau)|}{|L(t)-\tau|^{\beta}}\cdot \max_{0\leq u \leq t} |e^{\lambda u}\sigma(u)|^{2\beta}
\nonumber\\
&<\infty, \quad \text{a.s.,} \label{rmk-kern-2}
\end{align}
where 
the last inequality follows from the continuity of $\sigma$ and the local
$\beta$-H\"older continuity of Brownian motion paths for all $\beta\in(0,\frac{1}{2})$ (see \cite[Chapter I, Theorem (2.2)]{revuz1999continuous})).
We may take an arbitrary $\beta\in (\alpha-1,\frac{1}{2})$.
 Combining \eqref{rmk-kern-1} and \eqref{rmk-kern-2} yields that 
\begin{align}
\frac{|\overline S(t) - \overline S(u) |}{|t-u|^{\alpha}}
 =O(|t-u|^{\beta-\alpha}),  \label{holder-est-logS}
\end{align}
which implies that the integral of \eqref{h} is well defined since $\beta-\alpha \color{black}>\color{black} -1$.  
\end{proof}

{\color{black}With the help of Lemma \ref{kernel1}, we introduce a mechanism by which the history of $\overline{S}$ can have an influence on the volatility.}
It is natural to increase the amount of stored gas at any time $t$, if the storage is not already full yet, when the current price level of the underlying asset is comparatively lower, for instance $\overline{S}(t)\le \overline{S}_{threshold}$ where $\overline{S}_{threshold}$ is some price threshold, whereas, if the current price level gets too high, investors tend to free up the storage by selling the underlying assets at hand. Such threshold, for instance, can be captured by the above defined log moving average. Thus in view of Lemma \ref{kernel1},  we introduce the functional on path $h(\overline{S}_t)$ as defined in \eqref{h}. In fact, by \eqref{holder-est-logS}, integrating by parts gives that
 \begin{align*}
h (\overline{S}_t) 
&= (1-\alpha) \int_0^t \! \frac{\overline{S}(u) - \overline{S}(t)}{t^{1-\alpha}(t-u)^{\alpha}}du \nonumber
\\
&=   - \lim_{u\rightarrow t^-}(\overline{S}(u) - \overline{S}(t))t^{\alpha - 1} \cdot (t-u)^{1-\alpha} +  (\overline{S}(0) - \overline{S}(t))
+ \lim_{\tau \rightarrow t^-}\int_0^{\tau} \! t^{\alpha - 1} (t-u)^{1-\alpha}d\overline{S}(u) \nonumber
\\
&=  \overline{S}(0) - \overline{S}(t) + \int_0^t \!  \left(\frac{t-u}{t}\right)^{1-\alpha}d\overline{S}(u), 
\end{align*}
so that
\begin{equation}
h (\overline{S}_t)+\overline{S}(t)-\overline{S}(0)=\int_0^t \!  \left(\frac{t-u}{t}\right)^{1-\alpha}d\overline{S}(u). \label{h_delta_ibp}
\end{equation}
In Section~3 we will use this quantity to help define our rough volatility $\sigma(t)$. To motivate this approach, we consider the special case $\lambda=0$.
Inserting \eqref{eq-logPrice} into \eqref{h_delta_ibp}, we  obtain 
\begin{align}
h (\overline{S}_t) -\overline{S}(0) + \overline{S}(t) 
&=\int_0^t \!  \left(\frac{t-u}{t}\right)^{1-\alpha}\left[ \Big(r-\frac{\sigma(u)^2}{2}\color{black}\Big)du+ \sigma(u)dW(u)\right]
\nonumber\\
&
=\int_0^t \!  \left(\frac{t-u}{t}\right)^{1-\alpha}  \Big(r-\frac{\sigma(u)^2}{2}\Big)du+ \int_0^t \!  \left(\frac{t-u}{t}\right)^{1-\alpha}  \sigma(u)dW(u). 
\label{eq-rV1}
\end{align}
If we were developing a nonlinear model as in \eqref{ori_pro} by  setting  specifically
\begin{equation}
|\sigma(t)|^2-|\sigma(0)|^2= \left(h (\overline{S}_t) -\overline{S}(0) + \overline{S}(t) \right) \cdot t^{1-\alpha}, \label{rH_V}
\end{equation}
putting $Y(t)=|\sigma(t)|^2$, we could  have obtained  by \eqref{eq-rV1} that 
\begin{align}
Y(t)=Y(0) + \int_0^t \!  \frac{(t-u )^{1-\alpha} }{2} \Big(2r- Y(u)\Big)du+ \int_0^t \!   (t-u)^{1-\alpha}  \sqrt{Y(u)}dW(u),
\label{rough-Heston}
\end{align}
which yields a rough Heston model (see \cite{abi2019affine,chevalier2022american} for instance). Here, $\frac{3}{2} - \alpha$ measures the roughness analogous to the Hurst parameter.
\begin{prop}
When $\lambda=0$, the stochastic path-dependent differential system \eqref{eq-logPrice}, \eqref{rH_V} and \eqref{h} produces a rough Heston volatility model \eqref{rough-Heston} when $\alpha\in (1,\frac{3}{2})$, with Hurst parameter $H:=\frac{3}{2}-\alpha$.
\end{prop}

Thus, this approach allows us to establish a connection between the introduced path dependence and rough volatility models. As rough volatility emerges as a novel paradigm in finance, it prompts us to integrate path-dependence into the modeling of the volatility process $\sigma$.  However, in this study, we opt not to solely rely on rough volatility models, recognizing the necessity for a more nuanced approach to modeling natural gas prices. For instance, the volatility term $\sigma$ is also subject to market storage levels. Observations suggest that when storage is nearly full or depleted, (log-)price volatility (measured by $\sigma$) escalates. This is because the market's capacity to stabilize diminishes, either by requiring more storage for the underlying asset or by having limited storage space available. Therefore, we introduce both the path-dependence and storage-coupling by setting $\sigma(t) = f(X(t),h(\overline S_t)+\overline S(t) - \overline S(0))$ in \eqref{ori_pro} where the stochastic process $(X(t))_{t\geq 0}$ tracks the deviation of the storage level from its periodic curve as defined in next subsection.  Note that if $\alpha \in (\frac{1}{2}, 1]$, it holds obviously that $h(\overline S_t)-\overline S(0)+\overline S(t)=\overline{S}^{\alpha}_t -\overline S(0)$. 
In fact, empirical studies (see \cite{guyon2021volatility} for instance) indicate a positive correlation between  the average growth in (log-)prices (represented by $\overline{S}^{\alpha}_t -\overline S (0)$ herein) and  the volatility process $\sigma$. \color{black}This is different from the case in \cite{dandapani2021quadratic}, in which although the Zumbach effect remains explicit in the limiting rough Heston type volatility models, the main results are under a Hawkes based price process and the price volatility does not involve a storage factor.\color{black}

\subsection{Storage-volatility interaction}

Throughout this work, we apply the daily natural gas data \emph{NGX Phys FP [CA/GJ] AB-NIT-D} from January 5, 2019 to December 16, 2022, with a total of 1442 consecutive \emph{daily} one-day ahead prices and the associated \emph{weekly average} Mountain Region Natural Gas Working Underground Storage (Billion Cubic Feet) data from January 5, 2019 to December 9, 2022, with a total of 206 consecutive weekly storage data, to the calibration of the proposed model \eqref{model_new} and the associated discrete swing option pricing problem. 

Here we present Figure \ref{fig:storage_volatility} as a peek into the storage-volatility interactions in natural gas markets. The concept of ``price increment" is chosen to qualitatively describe the price volatility observed in the gas market data. This figure clearly shows when gas storage level is near depletion (reaching 0) or its full capacity (reaching 1), the gas price becomes more fluctuating, supporting our expectation that market storage should play the role of a \emph{calming factor} in natural gas markets,  at least to some extent. Meanwhile when the gas price is more volatile, gas storage level is not necessary at its peaks or troughs. This indicates that for one thing, the gas price, or its volatility, has a limited impact on the market storage level, and price volatility could also be driven by other factors, besides the gas storage level; for another, such "calming factor" might have a capacity.
\begin{figure}[!htbp]
\begin{center}
\includegraphics[width=0.49\linewidth]{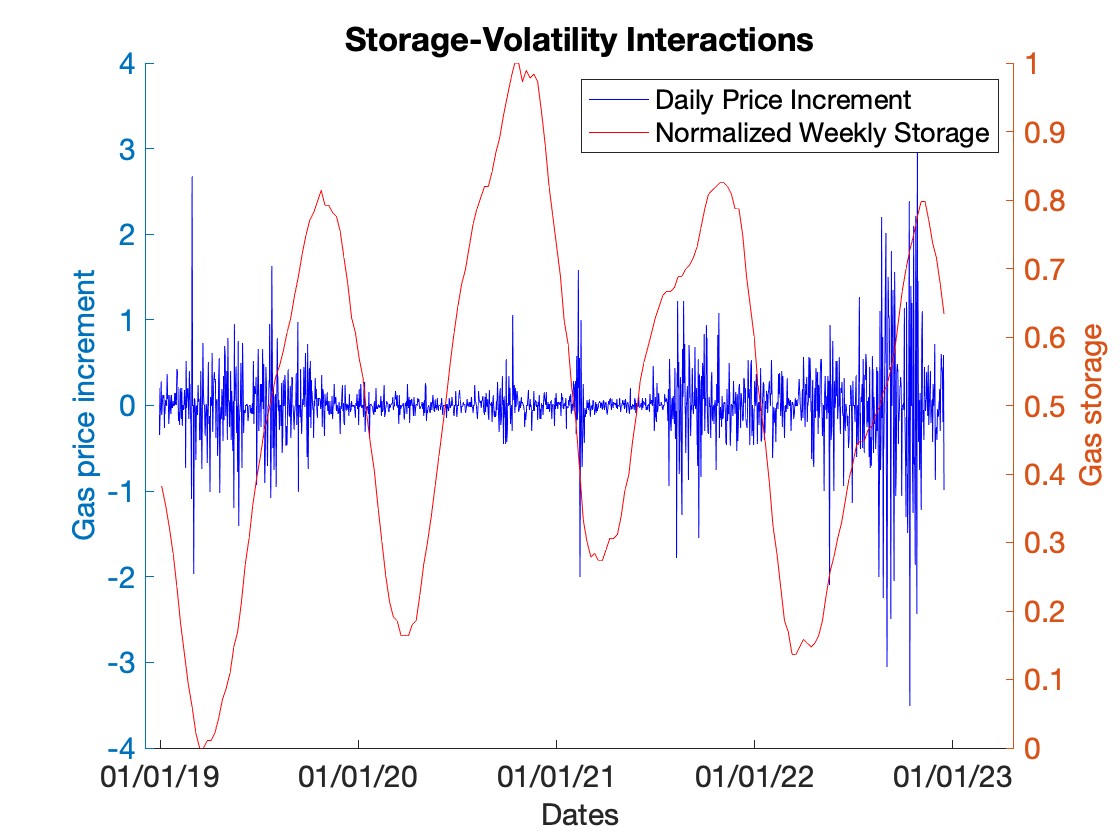}
\end{center}
\caption{Natural Gas Storage-Volatility Interactions (Jan, 2019 - Dec, 2022)}
\label{fig:storage_volatility}
\end{figure}

\subsection{Storage drivers}
In previous subsections, we have already established that the price volatility could be driven by past history of price and the market storage level in natural gas markets. \color{black}Such path-dependence could reasonably be extended to the gas storage modeling, as it inherently requires a proper assessment of gas prices at each point in time. The financial principle of \emph{“buy low, sell high”} naturally necessitates a comparison between current prices and their historical trends. \color{black} In this subsection, we take a deeper dive into the dynamics of the gas storage level and its drivers. Before we proceed further, there are two major issues we need to deal with:
\begin{enumerate}[label = (\Alph*)]
\item Deseasonalizing the storage data;
\item Mis-match of time step sizes between the price and storage data.
\end{enumerate}
\begin{rmk}
The strong periodic behaviour observed in historical gas storage data in Figure \ref{fig:storage_increment} is mainly driven by the seasonality of energy demand, for instance in heating, which heavily depends on other factors such as weather instead of the (log-)-price process itself. This renders the periodicity in gas storage highly predictable. Thus any influence of gas price levels can only be expected to be on the difference between storage levels and their usual seasonal levels. Because of this, we deseasonalize the storage data and incorporate this difference explicitly in our models. 
\end{rmk}
\begin{rmk}
The mis-match of time step sizes between the price and storage data is caused by our choice of gas price and storage datasets, which leads to inconsistency in the time evolution of log-price and storage processes in our proposed model \eqref{model_new}. Ideally this issue can be avoided provided that gas price and storage data with consistent time step sizes are available.
\end{rmk}

For the first issue, we adopt the following two-step decomposition scheme.

\noindent\textbf{Step 1}: Let $ST_{\text{weekly}}$ and $^1{ST_{\text{weekly}}}$ denote the original and normalized natural gas storage process, with $^1{ST_{\text{weekly}}}$ defined as
\begin{equation*}
^1{ST_{\text{weekly}}}(t) = \frac{ST_{\text{weekly}}(t)}{ST_{cap}},
\end{equation*}
for each $t\in [0,T]$ where the quantity $ST_{cap}>0$ is a known constant denoting regional gas storage capacity. For simplicity, we choose $ST_{cap} = \max_{t\in [0,T]}ST_{\text{weekly}}(t)$ throughout this work. Thus for any $t\in [0,T]$, $^1{ST_{\text{weekly}}}(t)\in (0,1]$. 

\noindent\textbf{Step 2}: We apply the Fourier transform method on the storage data to obtain the global periodic approximation $P_{\text{weekly}}(t)$ of the normalized storage data $^1{ST}_{\text{weekly}}(t)$. We take the period $L = \frac{365}{7}$, and the Fourier transform approximation result $\{P_{\text{weekly}}(t)\}_{t\in [0,T]}$ is shown below.
\begin{figure}[h]
\begin{center}
\includegraphics[width=0.45\linewidth]{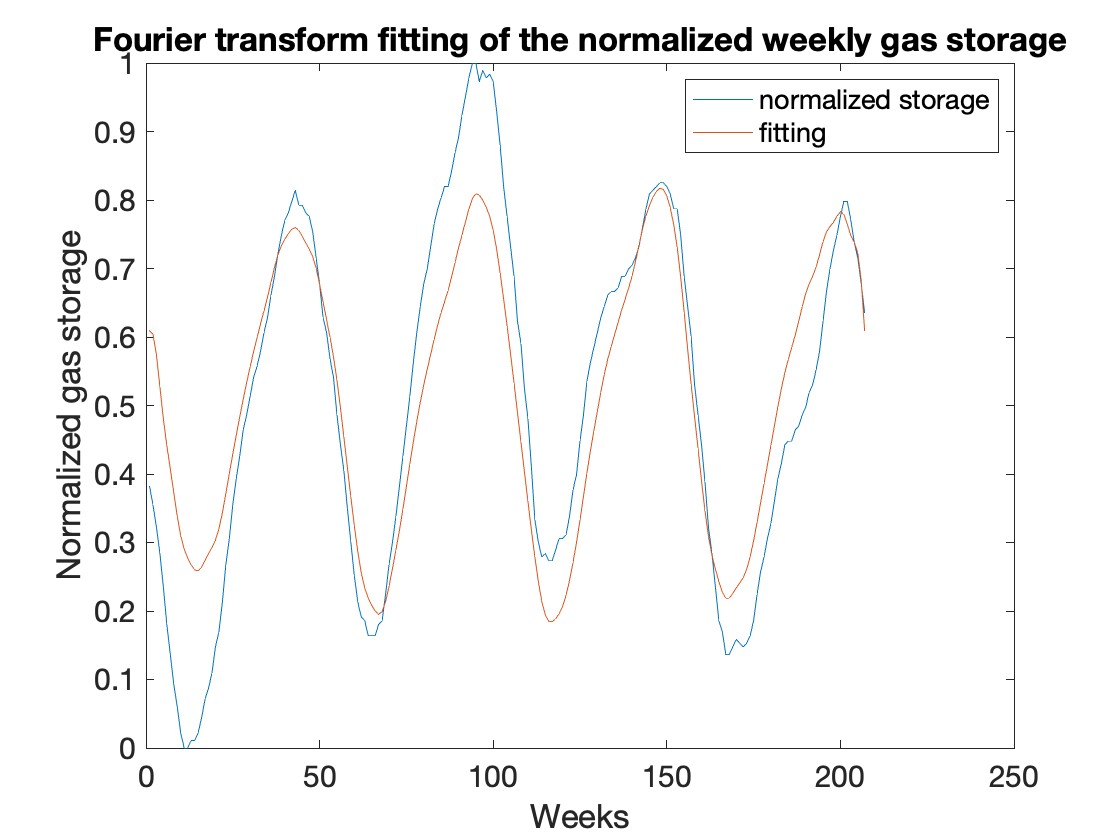}
\end{center}
\caption{Fourier transform of normalized storage data from Jan 4, 2019 to Dec 16, 2022.}
\label{fig:fourier}
\end{figure}
Note here the \emph{deseasonalized weekly average storage} is then defined by
\begin{equation}\label{realX_decomp}
X_{\text{weekly}}(t) := {^1{ST}}_{\text{weekly}}(t) - P_{\text{weekly}}(t).
\end{equation}

In order to solve the second issue, we apply step functions to convert weekly data
$P_{\text{weekly}},X_{\text{weekly}}$ to daily processes $P,X$, respectively, as introduced in \eqref{V_new} and \eqref{X_new}. \color{black}Alternatively, a linear interpolation approach could be applied to convert weekly data to daily ones. 

\section{A Novel Path-Dependent Volatility Model}
\subsection{Price-storage model}
\color{black}Consider the following log-price-volatility dynamics
\begin{align}
d\overline S(t) = &\left(r-\frac{1}{2}|\sigma(t)|^2-\lambda\overline S(t)\right)dt + \sigma(t)dW(t),\label{lin_S}
\\
\sigma(t) = &\Psi(t, (\overline S)_t, X(t) + P(t)),\label{price_vol}
\end{align}\color{black}
where $\Psi$ is, in general, an adapted, time-continuous, and non-negative function, and for each $t\in [0,T]$, $\Psi(t,\cdot,X(t)+P(t))$ is assumed to be nondecreasing and Lipschitz continuous with respect to the path. The drift coefficient in \eqref{lin_S} is possibly mean-reverting in general, with $\lambda\ge 0$; however, the diffusion coefficient is nonlinear due to the path-dependence introduced by the moving average. The well-posedness of the solution to \eqref{lin_S} may then be guaranteed. 

The following is the specific price-storage model that we adopt in this paper. Given a time-periodic function $(P(t))_{t\geq 0}$,  $\alpha\in (\frac{1}{2}, \frac{3}{2})$,  $\delta>0$, $\lambda,V_0,V_1,V_2\ge 0$, and $  r,\gamma_1,\gamma_2 \in \bR$, the couple $(\overline S,X)$ satisfies:
\begin{subequations}
\label{model_new}
\begin{align}
d\overline{S}(t) &= \left( r - \frac{1}{2} \sigma(t)^2 - \lambda\overline S(t) \right)dt +  \sigma(t)dW(t), \quad  \overline{S}(0) = \overline{S}_0 \text{ is given,}
 \label{S_new} \\ 
\sigma(t) &= V_0 + \frac{V_1}{(X(t)+P(t))(1- X(t) - P(t))+\delta} + V_2 \sqrt{|\overline{S}_t^{\alpha,\delta}-\overline{S}(0)| + \delta} ,\label{V_new}
\\
dX(t) &= \left[ \gamma_1 R^{+}(t) (1-X(t)- P(t)) - \gamma_2 R^{-}(t)(X(t) + P(t))\right]dt,\quad X(0) \text{ is given,} \label{X_new}
\end{align}
\end{subequations}
with
\begin{align}
\overline{S}^{\alpha,\delta}_t &=
(1-\alpha) \int_0^t \! \frac{\overline{S}(u) }{(t+\delta)^{1-\alpha}(t-u+\delta)^{\alpha}}du  \color{black}+\color{black} \overline S(t)1_{\{\alpha= 1\}}, \label{S-delta}
\\
R(t) &= (1-\alpha)\int_0^t \! \frac{\overline{S}(u)-\overline{S}(t)}{(t+\delta)^{1-\alpha}(t-u+\delta)^{\alpha}} du.\label{R_new}
\end{align}
Here, $X(t)+P(t)$ represents the daily time $t$ normalized market storage, with the ``weekly-daily" transformation accomplished by a step function. $R^+,R^-$ are non-negative and non-positive parts of $R$, respectively, defined in the standard way
\begin{align*}
R^{+}(t) &= {R}(t)\cdot 1_{\tilde{R}(t)\ge 0}(t),
\\
R^{-}(t) &= (-{R}(t))\cdot 1_{\tilde{R}(t)< 0}(t).
\end{align*}
To ensure well-posedness and computational stability without resorting to cumbersome arguments, we introduce a very small $\delta>0$. The model is to be calibrated for the parameters of the set $\Theta =\{\alpha,r,\lambda,V_0,V_1,V_2,\gamma_1,\gamma_2\}$.
\begin{rmk}
When $\lambda=0$, the stochastic differential equation (SDE) \eqref{S_new} is equivalent to the associated SDE with respect to the price process in \eqref{linear_no_mr}. Throughout this work, we stick to the log-price dynamics since it provides more convenience for us to assess the wellposedness of the solution to system \eqref{model_new}.
\end{rmk}
\begin{rmk}
The volatility process in \eqref{V_new} comprises three components. When $V_1=V_2=0$ and $V_0>0$, our model reduces to the conventional constant volatility scenario; when $V_0=V_1=0$ and $V_2>0$, it aligns with a rough Heston model as discussed in \eqref{eq-rV1}-\eqref{rough-Heston}; when $V_1>0$, the volatility process becomes dependent on storage levels.
Notably, the second term in \eqref{V_new} indicates that volatility tends to rise as the storage level approaches its boundaries, signifying moments when the market storage is near full capacity or depletion, thereby diminishing the market's stabilizing influence and consequently amplifying volatility in natural gas prices.
The inclusion of the third term in \eqref{V_new} is inspired by rough volatility models as discussed in \eqref{eq-rV1}-\eqref{rough-Heston}. To facilitate the analysis of well-posedness and subsequent computational approximations, we adopt a $\delta$-approximation technique to mitigate the lack of Lipschitz continuity in square root functions around 0. 
\end{rmk}

\begin{rmk}
\color{black} The normalized storage level $X+P\in (0,1]$. According to empirical studies, the deseasonalized storage $X$ is more likely to decrease or increase, respectively, when the price level is relatively high or low. Meanwhile this behaviour aligns with the tendency when the normalized storage level is reaching fullness or depletion, which is consistent with real market behaviour of natural gas storage. \color{black}This is consistent with the calibration results in Table \ref{st_14} and \ref{st_30}, where we have more parameters and most entries of $\vec{\gamma_1},\vec{\gamma_2}>0$\color{black}. Indeed, this is a sufficient condition for solutions of \eqref{X_new} to satisfy $X+P\in (0,1]$. Nevertheless, in the model we allow $\vec \gamma_1,\vec \gamma_2 \in\bR$ to cover all the possible range, since  the storage dynamics in real markets may also be affected by some complicated factors other than price and inventories, for instance weather, supply and demand (\cite{brown2008drives,hulshof2016market}). 
\end{rmk}

\subsection{Wellposedness of the path-dependent stochastic differential equation \eqref{model_new}}

This subsection is dedicated to establishing the existence and uniqueness of the strong solution to the path-dependent stochastic differential equation \eqref{model_new}. The methodology employed combines localization techniques and a priori estimates. 
Throughout the remainder of this paper, we will denote a constant by $C$, whose specific value may vary from line to line. Additionally, when we use $C_{a_1,a_2,\dots}$, it implies that the constant depends on the parameters $a_1,a_2,\dots$.

First, we present two priori estimates. 
 
\begin{lem}\label{lemma_wp}
  Let $\tau$ be an arbitrary stopping time. Suppose that the continuous process  $(\overline S,\sigma,X)$ satisfies the stochastic system \eqref{model_new} on $[0,\tau]$. It holds that for any $T>0$,
\begin{enumerate}[label=(\roman*)]
\item there exists a constant $C_{\alpha,\delta,\overline{S}(0),T,V_0,V_1,V_2}>0$ such that for each $t\in [0,T]$,
\begin{equation*}
\sup_{0\leq s\leq t}|\sigma(s\wedge \tau)|^2 \le C_{\alpha,\delta,\overline{S}(0),T,V_0,V_1,V_2} \left( 1+ \sup_{s\in [0,t]}|\overline{S}(s\land\tau)| \right), \text{ a.s.,}
\end{equation*}
\item there exists $0<C<\infty$ such that
\begin{equation*}
 \color{black}E\color{black} \left[ \color{black}\sup_{0\le t\le T}\color{black} \left| \overline{S}(t\land\tau) \right|^2 \right] \le \color{black}C,
\end{equation*}
where $C$ is independent of $\tau$.
\end{enumerate}
\end{lem}
For simplicity, the proof will be postponed to Appendix A. Subsequently, we are ready to present the well-posedness result. 
\begin{thm}\label{thm_wp}
For each $T>0$, the stochastic differential equation \eqref{model_new}  admits a unique strong solution $(\overline S, X)$ over the time interval $[0,T]$ with
\begin{align}
E\left[
\sup_{t\in[0,T]} \left(|\overline S(t)|^2+ |X(t)|^2\right)
\right] <\infty. \label{est-thm}
\end{align}
\end{thm}
For similar reasons, the proof of this theorem will also be postponed to Appendix A.

\section{Calibration}
\color{black}Our proposed storage dynamic $X$ is modeled using an ordinary differential equation (ODE). As illustrated in Figures \ref{fig:price_increment} and \ref{fig:storage_increment}, the weekly storage trajectory exhibits pronounced regularity compared to the highly volatile daily price dynamics. The \emph{``smooth"} evolution of weekly storage contrasts sharply with the erratic behavior of prices, which motivates the exclusion of diffusion processes (e.g., Brownian motion) from its characterization. This structural distinction allows for a two-step calibration procedure: first, we estimate the parameters of the price model \eqref{S_new}, and subsequently, we calibrate the unknown coefficients $\gamma_1,\gamma_2$ in the storage dynamic 
$X$. By decoupling these calibrations, we ensure robustness while maintaining mathematical tractability in our proposed framework. The two-step calibration is similar in \cite{chen2010implications}, in which the calibration of the volatility parameter pair is separated from that of other parameters.\color{black}
\color{black}
\subsection{Log-likelihood function}
For simplicity, let $T$ be divisible by $\Delta t=\frac{1}{365}$, i.e. $T=\frac{\mathcal{N}}{365}$ for some $\mathcal{N}\in\bN^+$. Further, let $\Pi:=\{t_0,\cdots,t_{\mathcal{N}}: 0=t_0<\cdots <t_{\mathcal{N}}=T\}$ be a partition of $[0,T]$ with time step $\Delta t$. Under \eqref{S_new} - \eqref{X_new} and \eqref{S-delta} - \eqref{R_new}, forward Euler scheme gives
\begin{subequations}
\label{euler_approx}
\begin{align}
\overline{S}(t_{i+1}) \approx &\overline{S}(t_i) + \left(r-\frac{\sigma(t_i)^2}{2} - \lambda \overline{S}(t_i)\right)\Delta t + \sigma(t_i)(W(t_{i+1})-W(t_{i}))
\\
\sigma(t_i) \approx &V_0 + \frac{V_1}{(X(t_i)+P(t_i))(1-X(t_i)-P(t_i)) + \delta} + V_2\sqrt{|\overline{S}_{t_i}^{\alpha,\delta} - \overline{S}(0)|+\delta}
\\
X(t_{i+1}) \approx &X(t_i) + \Delta t [\gamma_1R^+(t_i)(1 - X(t_i) - P(t_i)) - \gamma_2R^-(t_i)(X(t_i)+P(t_i))], \label{eulerX}
\end{align}
\end{subequations}
with
\begin{align*}
\overline{S}_{t_i}^{\alpha,\delta} \approx &(1-\alpha)\sum_{m=0}^{i} \frac{\overline{S}(t_m)\Delta t}{(t_i + \delta)^{1-\alpha}(t_i - t_m + \delta)^{\alpha}} + \overline{S}(t_i)1_{\{\alpha = 1\}}
\\
R(t_i) \approx &(1-\alpha)\sum_{m=0}^{i} \frac{(\overline{S}(t_m) - \overline{S}(t_i) )\Delta t}{(t_i + \delta)^{1-\alpha}(t_i - t_m + \delta)^{\alpha}}.
\end{align*}

\color{black}
 
 
Then let $\overrightarrow{\theta}:=(\alpha,r,\lambda,V_0,V_1,V_2)$ denote a parameter set that awaits optimization. For any $i\in \{1,\cdots,\mathcal{N}\}$, the approximated likelihood function under daily log-price framework takes the form
\begin{align}
\sL\{\overrightarrow{\theta}|\overline{S}_{t_i}\} \sim &\sL \{ \overrightarrow{\theta}|\overline{S}(t_0),\cdots, \overline{S}(t_i)\} \nonumber
\\
= &\text{ }p^{\overrightarrow{\theta}}(\overline{S}(t_0),\cdots,\overline{S}(t_i)) \nonumber
\\
= &\text{ }p^{\overrightarrow{\theta}}(\overline{S}(t_1),\cdots,\overline{S}(t_i)|\overline{S}(t_0))\cdot p^{\overrightarrow{\theta}}(\overline{S}(t_0)) \nonumber \nonumber
\\
= &\prod_{j=1}^i p^{\overrightarrow{\theta}}(\overline{S}(t_j)|\overline{S}(0),\cdots,\overline{S}(t_{j-1})) \nonumber
\\
= &\prod_{j=1}^i\frac{1}{ \sqrt{2\pi \Delta t} \left| \sigma(t_{j-1}) \right|} \exp\left\{ -\frac{( \overline{S}(t_j) - \overline{S}(t_{j-1}) - (r-\frac{1}{2}| \sigma(t_{j-1})|^2 - \lambda \overline{S}(t_{j-1}))\Delta t )^2}{2\Delta t| \sigma(t_{j-1})|^2} \right\}. 
\end{align}
Hence the corresponding log-likelihood function takes the form
\begin{align*}
\tilde{l}(\overrightarrow{\theta}) = -\frac{i}{2}\ln{(2\pi \Delta t)}  - \sum_{j=1}^i \ln{\left| \sigma(t_{j-1}) \right|} -\sum_{j=1}^i \frac{\left[ \overline{S}(t_j) - \overline{S}(t_{j-1}) - (r - \frac{1}{2}| \sigma(t_{j-1})|^2 - \lambda \overline{S}(t_{j-1}))\Delta t \right]^2}{2\Delta t | \sigma(t_{j-1})|^2}.
\end{align*}
For computational simplicity, we omit the terms that are not dependent on any of the parameters and obtain the following \emph{rescaled log-likelihood function} $l(\overrightarrow{\theta})$
\begin{equation}\label{rellf}
l(\overrightarrow{\theta}) =  - \sum_{j=1}^i \ln{\left| \sigma(t_{j-1}) \right|} -\sum_{j=1}^i \frac{\left[ \overline{S}(t_j) - \overline{S}(t_{j-1}) - (r - \frac{1}{2}| \sigma(t_{j-1})|^2 - \lambda \overline{S}(t_{j-1}))\Delta t \right]^2}{2\Delta t | \sigma(t_{j-1})|^2}.
\end{equation}

\color{black}
\subsection{Mean square error}

\noindent When calibrating the proposed model \eqref{model_new}, although it is natural for investors to pick their time of interest according to global or regional events that could have an impact on the natural gas market, the analysis of storage dynamics requires strategies that are more subtle in choosing proper time intervals. For general purposes, $\gamma_1$ and $\gamma_2$ could be time-dependent, as factors that impact on gas storages change through time in general. However, for practical reasons such as the delay in transportation and demand, it is reasonable to assume that these parameters remain constant, at least over a short time period. Consequently, we update $\gamma_1$ and $\gamma_2$ regularly within a time framework $\mathcal{T} \leq T$. Starting from the initial time $t_0$, the parameters $\gamma_1(t)$ and $\gamma_2(t)$ behave as step functions and are updated after each interval of $\mathcal{T}$ time units. In this case, calibrating $\gamma_1(t)$ and $\gamma_2(t)$ on the interval $[t_0,t_{\mathcal{N}}]$ requires obtaining the optimal parameter vectors $\vec{\gamma_1}$ and $\vec{\gamma_2} \in \mathbb{R}^k$, where $k = \frac{T}{\mathcal{T}} $ if $T$ is divisible by $\mathcal{T}$ or $k = \lfloor \frac{T}{\mathcal{T}} \rfloor +1$ otherwise. The choice of $\mathcal{T}$ can be flexible depending on specific market behaviour and investor demand. When $\mathcal{T} = T$, $\vec{\gamma_1},\vec{\gamma_2}\in\bR$ are instead constant numbers throughout the whole time interval.

Let $m = \frac{\mathcal{N}+1}{7} $ if $\mathcal{N}+1$ is divisible by $7$ or $m = \lfloor \frac{\mathcal{N}+1}{7} \rfloor+1$ otherwise. Under the time partition $\Pi$, following \eqref{eulerX} we use real data of $S$ and $P$ to simulate the daily fitting of storage factor $\hat{X}_{\text{daily}}(t_i)$ for $i = 1,\cdots,\mathcal{N}$ with $\hat{X}_{\text{daily}}(t_0) = X(t_0)$. The fitted weekly storage factor $\hat{X}_{\text{weekly}}$ is then obtained by computing weekly average of $\hat{X}_{\text{daily}}$ for every seven days, i.e., for each $i\in\{1,\cdots,m\}$,
\begin{equation*}
\hat{X}_{\text{weekly}}(i) = \frac{ \hat{X}_{\text{daily}}(t_{7(i-1)}) + \cdots + \hat{X}_{\text{daily}}(t_{(7(i-1)+6)\land\mathcal{N}})  }{7\land (\mathcal{N}-7(i-1)+1)}.
\end{equation*}

The optimized parameters $\vec{\gamma_1}$ and $\vec{\gamma_2}$ are obtained by minimizing the following mean square error
\begin{align}
\text{MSE} := \sum_{i=1}^{m} \left| X_{\text{weekly}}(i) - \hat{X}_{\text{weekly}}(i) \right|^2,
\end{align}
where $X_{\text{weekly}}(i)$ is the real weekly data at time step $i$ in \eqref{realX_decomp}, i.e., we have
\begin{equation}\label{min_mse}
\left( \vec{\gamma_1}, \vec{\gamma_2} \right) = \argmin_{\vec{\gamma_1}, \vec{\gamma_2}} \sum_{i=1}^{m} \left| X_{\text{weekly}}(i) - \hat{X}_{\text{weekly}}(i) \right|^2.
\end{equation}
Note that \eqref{min_mse} can be obtained piecewisely, i.e., on each time period $\gamma_1,\gamma_2$ are just two constants with 
\begin{align}
\vec{\gamma_1} &= \left( \gamma_{1,1}, \cdots, \gamma_{1,k} \right),
\\
\vec{\gamma_2} &= \left( \gamma_{2,1}, \cdots, \gamma_{2,k} \right),
\end{align}
where for each $i = 1,2$ and $j = 1, \cdots, k$, $\gamma_{i,j}\in\bR$.
\color{black}

\subsection{A two-step calibration algorithm}
\noindent\textbf{Step 1}. In order to find the optimized parameters maximizing \emph{the rescaled log-likelihood function} \eqref{rellf}, we apply a Consensus-Based Optimization (CBO) method with particle simulations for the parameters $\overrightarrow{\theta} = \theta_1 = \{\alpha,r,\lambda,V_0,V_1,V_2\}$. Note that in the formula \eqref{rellf} the volatility $\sigma$ is function of parameters $\overrightarrow{\theta}$. To avoid possible singularities when computing derivatives, we adopt the \textit{derivative-free} CBO method. In fact, it has been proven that CBO can guarantee global convergence under suitable assumptions, and it is a powerful and robust method for many high-dimensional non-convex optimization problems. A short introduction of CBO method and the associated references are included in the Appendix.

\noindent\textbf{Step 2}. \color{black}We apply the forward Euler approximation of $X$ introduced in \eqref{euler_approx} with the optimized parameter $\alpha$ that is obtained via Step 1 to numerically compute the storage dynamic. Then with CBO method on $\overrightarrow{\theta} = \theta_2=\{\vec{\gamma_1},\vec{\gamma_2}\}$, we are able to calibrate the model by minimizing the mean square error (MSE) \eqref{min_mse} between fitted storage factor $X$ via \eqref{eulerX} and its associated normalized real storage component. \color{black} 
\\
The calibrations are conducted piecewisely with respect to each time interval. The calibration results for log-price are shown in Table \ref{price_parameters}, with CBO parameters of \eqref{cbo} as in Table \ref{tab:CBO parameters}.

\begin{table}[!ht]
\centering
\begin{tabular}{|c|c|c|c|c|c|c|c|c|c|c|c|}
\hline
Step 1 & a    & b                     & $\sigma$ & M   & N    & Step 2 & a   & b                       & $\sigma$ & M   & N   \\ \hline
       & 1200 & 400 & 20       & 100 & 3000 &        & 1500 & 1500 & 30 & 500 & 4000 \\ \hline
\end{tabular}
\caption{CBO Parameters for price and storage model calibration}
\label{tab:CBO parameters}


\begin{center}
\begin{tabular}{|c|c|c|c|c|c|c|}
\hline
Time interval   & $\alpha$ & $r$    & $\lambda$ & $V_0$  & $V_1$  & $V_2$  \\ \hline
01/2019-10/2019 & 1.4561   & 5.2536 & 4.2638    & 2.1268 & 0.1361 & 4.0786 \\ \hline
11/2019-03/2020 & 0.8734   & 2.2244 & 4.8764    & 0.7193 & 0.0341 & 0.1893 \\ \hline
03/2020-12/2020 & 1.4555   & 5.0995 & 6.7183    & 0.4703 & 0.0331 & 0.5034 \\ \hline
01/2021-06/2021 & 1.0387   & 3.5386 & 1.4187    & 0.0501 & 0.0118 & 1.7770 \\ \hline
06/2021-02/2022 & 1.0440   & 1.3144 & 0.2145    & 0.4443 & 0.0758 & 3.4641 \\ \hline
02/2022-12/2022 & 1.4961   & 5.5812 & 0.9507    & 0.0803 & 0.0186 & 5.2523 \\ \hline
\end{tabular}
\end{center}
\caption{CBO Results for log-price model calibration} \label{price_parameters}

\begin{center}
\begin{tabular}{|c|c|c|}
\hline
Time  Interval          & $\gamma_1$ & $\gamma_2$ \\ \hline
01/2019-10/2019 & 0.1040     & -0.3616    \\ \hline
11/2019-03/2020 & -40.2865    & 59.8992     \\ \hline
03/2020-12/2020 & 14.2155     & 0.8597     \\ \hline
01/2021-06/2021 & 10.3480     & 56.8883     \\ \hline
06/2021-02/2022 & 10.1753    & 34.7471    \\ \hline
02/2022-12/2022 & 0.5443    & 0.1756    \\ \hline
\end{tabular}
\end{center}
\caption{CBO Results for storage calibration($\mathcal{T}=T$)} \label{st_T}
\end{table}

\begin{rmk}\color{black}
For each time interval of interest $[0,T]$, we provide storage calibration results for three cases: $\mathcal{T} = T$, $\mathcal{T} = 14$(biweekly) and $\mathcal{T} = 30$(monthly), corresponding to Table \ref{st_T}, Table \ref{st_14} and Table \ref{st_30} respectively. Please check Table \ref{st_14} and \ref{st_30} in Appendix B for formatting reasons. It is easy to see that most of the entries of $\vec{\gamma_1},\vec{\gamma_2}$ in Table \ref{st_14} and \ref{st_30} are in fact positive. The associated weekly $X$ fittings for each time interval are also included in Appendix B by Figure \ref{X_fit}, Figure \ref{X_fit_14} and Figure \ref{X_fit_30} respectively. The fitting performance depends not only on the choice of $\mathcal{T}$, but also on the choices of $T$ among many other factors. In reality, investors can choose their interval of interests according to market behaviours and better fitting performance.
\end{rmk}\color{black}

The optimal value $\alpha$ fits our observation of Figure~\ref{price_inc} very well. As we can see in the time interval from Jan, 2019 to Oct, 2019, we have a strong roughness in the volatility of the price process. During that period, there was extreme tension between US and Iran, with the latter happening to be a huge producer of natural gas in the global market. Due to the pressure posed by US, Iran terminated several items in the Iran Denuclearization Agreement to counter US sanctions on May 8th. The tension in middle-east continued to escalate throughout the year, making the energy market more chaotic and fluctuating, until December, near the end of 2019, when altogether six European countries including Belgium, joined in the trade settlement previously established by Iran, France, Britain, and Germany, in order to promote normalization of economical activities with Iran.

\color{black}Starting Nov, 2019, the initial onset of the COVID-19 pandemic had a significant global impact. Due to the numerous and irregular lock-downs and further economic shut-downs, the natural gas market stayed quiet. However, in 2020, the natural gas market in North America experienced significant volatility, driven by a combination of demand shocks caused by the on-going pandemic, saturated storage level due to over-supply, and the impact of extreme weather in late 2020. U.S. shale production remained high because of falling demand, leading to a surplus. Storage levels reached record highs by November (see Figure \ref{fig:storage_volatility}), pressuring prices downward. Meanwhile, the global temperature set a new record for hottest September. The expectation of a warm winter temporarily drove the gas market volatility up with increased uncertainty in demand but quickly faded out with the pass of winter. After winter, the gas market continued on in a \emph{``quiet"} state for the following year for pandemic issues, less volatile due to the past of \emph{initial shock} and continuous global economic depression.\color{black}

From June, 2021 to Feb, 2022, price volatility roughness continues to stay low due to the lasting effect of the pandemic. However tension starts to emerge on the border of Russia and Ukraine, which leads to the outbreak of war between the two countries at Feb, 24th, 2022. This explains the extreme roughness (with $\alpha$ close to $1.5$) in price volatility from Feb, 2022 to Dec, 2022, corresponding to our last time interval.

\begin{rmk}\color{black}
The time intervals in Table \ref{price_parameters} are carefully selected to align with significant real-life global events that influence the natural gas market. While there may be a theoretical distinction between the choice of time intervals for storage calibration and those for price analysis, this paper maintains consistency by adopting the same intervals for the calibration of $\vec{\gamma}_1$ and $\vec{\gamma}_2$. This decision is motivated by several factors. Firstly, achieving an optimized $\alpha$ is crucial for refining the storage model. Secondly, with only two parameters in the storage dynamics, employing smaller time intervals leads to simpler storage behavior, thereby enhancing calibration accuracy(see Figure \ref{X_fit}, \ref{X_fit_14}, \ref{X_fit_30}). In practice, investors can dynamically apply calibration over time, considering adjustments to the starting time in response to global or regional events impacting the natural gas market. The decision to switch time intervals may be affected by the investor's market experience and their awareness of prevailing global trends.
\end{rmk}\color{black}

\section{Discrete-Time Swing Option Pricing}
\color{black}It is natural to study the applications of our proposed stochastic path-dependent volatility model in the context of option pricing problems. Discrete-time swing options, in particular, are popular in natural gas markets due to the unique characteristics of natural gas as a commodity and the specific needs of market participants.

First, natural gas demand is highly seasonal and predictable. Discrete swing options allow buyers to adjust their gas consumption or delivery volumes within predefined limits, providing the flexibility needed to match supply with demand.

Second, natural gas prices are highly volatile, influenced by factors such as weather, storage levels, and geopolitical events. Swing options enable market participants to manage price risk by allowing them to purchase more gas when prices are low and reduce purchases when prices are high.

Third, natural gas infrastructure, including pipelines and storage facilities, has operational limits on flow rates and capacity. The limitations, unfortunately, lead to time delays in real life practice, rendering continuous options less realistic and favorable. Discrete-time swing options allow buyers to adjust their gas intake within these constraints, ensuring smooth operations.

In the following, we introduce a put-type discrete-time swing option and the associated pricing problem.\color{black}

\subsection{A put-type swing contract}

Let us consider the pricing problem of a discrete \emph{put-type} swing option. The corresponding \emph{call-type} swing option may be priced in a similar approach. Applying the updated stochastic differential equation system \eqref{model_new}, the pricing problem of the swing option above can be transformed into a path-dependent stochastic optimal control problem. 

Denote by $K$ the strike price. Let $L,\tilde{L}$ denote the global and local constraints specified by the swing put contract, i.e. the investor is allowed to exercise at most $L$ rights during the contract lifetime and each time with an upper threshold of $\tilde{L}$ rights. Denote $\mathbb{I}:=\{\tau_0,\tau_1,\cdots,\tau_M\}$ for some $M\in\mathbb{N}^+$ as the set of $M+1$ increasing possible choices of exercising time during the swing contract lifetime, where $\tau_i\in [0,T)$ for each $i\in \{0,1,2,\cdots,M\}$. In what follows, we set $\tau_{M+1}=T$. Without any loss of generality, for each $i$ let us make $\tau_i$ an integer multiple of $\frac{1}{365}$ and similarly for $T$ since throughout this work, we are dealing with daily price data. It is obvious that $\tilde{L}\le L\le (M+1)\tilde{L}$. Let $\{q(t)\}_{t\in \mathbb{I}}$ be the discrete control process expressing the exercise strategy of the swing option, where for each $t\in \mathbb{I}$, \color{black}$q(t)\in U:=\{0,1,\cdots,\tilde{L}\}$\color{black}. Moreover, we denote by $\cU$ the set of all the $U$-valued and $\{\sF_t^S\}_{t\in [0,T]}$-adapted processes, \color{black}where $\{\sF_t^S\}_{t\in [0,T]}$ is the filtration generated by $S$. $S$ could be the forward Euler approximation of the price process as introduced by \eqref{euler_approx}, but in the following we consider it as the continuous price process for generality purpose\color{black}. 

Naturally, $q(t)=0$ for any $t\in [0,T]\setminus\mathbb{I}$. Then for each time $t\in\mathbb{I}$, the remaining exercise right can be expressed by
\begin{equation}\label{Q}
Q(t) := L - \sum_{i\in\mathbb{I},i< t} q(i)\cdot 1_{\{Q(i)\geq q(i)\}}(i),
\end{equation}
where $Q(t)\ge 0$ for each $t\in [0,T]$.

When global constraints are not reached at maturity, some penalties may be applied. These might be proportional to the terminal price $S(T)$ and to the under-consumption of a positive $Q(T)$. The penalty function takes a general form of $G(S(T),Q(T))$.

Then for each $i\in\{ 0,\cdots,M\}$, the discrete dynamical reward functional is defined as
\begin{align}
&\mathcal{J}(\tau_i,Q(\tau_i);q(\tau_i)) \nonumber
\\
= &E_{\sF^S_{\tau_i}} \left[ \sum_{j\in [i,M]}e^{-\mu (\tau_j-\tau_i)}q(\tau_j)(K-S(\tau_j))^+\cdot 1_{\{Q(\tau_j)\geq q(\tau_j)\}}(\tau_j) + e^{-\mu(T-\tau_i)}G(S(T),Q(T)) \right], \label{reward}
\end{align}
where $\mu\ge 0$ is the risk free interest rate. And the corresponding value functional is then defined as
\begin{equation}\label{value}
\mathcal{V}(\tau_i,Q(\tau_i)) := \esssup_{q\in \cU} \mathcal{J}(\tau_i,Q(\tau_i);q).
\end{equation}

\begin{rmk}
A standard choice of the penalty function (see \cite{bardou2009optimal,barrera2006numerical,basei2014optimal,jaillet2004valuation,lari2001discrete}) is proportional to the terminal price. For instance, we may adopt the following one:
\begin{equation}\label{penalty}
{\color{black} G(S(T),Q(T)) = - A(K - S(T))^+Q(T),}
\end{equation}
where $A>0$ is a real constant.
\end{rmk}

We introduce the following dynamic programming principle for the pricing problem.
\begin{thm}\label{thm_dpp}
For any $i\in\{0,\cdots,M\}$, we have
\begin{equation}\label{DPP}
\mathcal{V}(\tau_i,Q(\tau_i)) = \esssup_{q\in U,q\le Q(\tau_i)} E_{\sF^S_{\tau_i}} \left\{  q(K-S(\tau_i))^+ + e^{-\mu(\tau_{i+1}-\tau_{i})}\mathcal{V}(\tau_{i+1},Q(\tau_{i+1})) \right\}.
\end{equation}
\end{thm}
The proof is postponed to Appendix A. 

The dynamic programming principle is the key component of the swing option pricing problem, although it redirects the problem to the computation of the conditional expectation, which is hard to compute numerically via conventional Monte-Carlo method \color{black}or tree-based methods. Classical Longstaff-Schwartz method in \cite{longstaff2001valuing} offers a robust linear unbiased estimator for such conditional expectations under a Markovian framework with different choices of state-dependent basis functions. The convergence of the approximation is in mean square error and in probability (see \cite{white1984nonlinear}). However, in this work we are dealing with a path-dependent and thus non-Markovian system \eqref{model_new} with the path-dependence embedded in log-price volatility. \color{black}For this we introduce a deep learning-based method.

\subsection{Deep-learning based method}
\subsubsection{Feed forward neural network}\label{NNSection}\label{section_NN}
Without any loss of generality, we take $0=\tau_0<\tau_1<\cdots<\tau_M<\tau_{M+1}=T$. And for each time interval $[\tau_i,\tau_{i+1}]$ with $i\in {0,1,\cdots,M}$, let us introduce a partition of $N_i+1$ points $\tau_i = t_0^i<t_1^i<\cdots<t_{N_i}^i=\tau_{i+1}$, with time step size, for instance, $dt = \frac{1}{365}$. Then $N_i=365(\tau_{i+1}-\tau_i)=365\Delta \tau_i$ and by \eqref{DPP} and forward Euler scheme, we have
\begin{align*}
\mathcal{V}(\tau_m,Q(\tau_m)) = \sup_{q} E_{\sF^S_{\tau_m}} \big\{ &q(\tau_m)(K-S(\tau_m))^+\cdot 1_{\{Q(\tau_m)\geq q(\tau_m)\}}(\tau_m) 
\\
& + e^{-\mu(\Delta \tau_m)}\mathcal{V}(\tau_{m+1},Q(\tau_{m+1})) \big\}.
\end{align*}
Note that unlike $\tau_i$ for $i\in \{0,\cdots,M\}$, time $\tau_{M+1}=T$ is not a choice of exercise rights.
\begin{rmk}
Some previous work (see \cite{barrera2006numerical}) adopt a penalty function similar to
\begin{equation*}
    G(S(T),Q(T)) = -A\cdot S(T)Q(T).
\end{equation*}
If the penalty function $G$ is proportional to the terminal price $S(T)$, then it is crucial for an optimal trading strategy to ensure $Q(T)=0$, as failure to do so would result in the investor incurring a penalty fee.
However, this may not hold true if the penalty function is proportional to the terminal intrinsic value rather than the terminal price itself, as proposed by \eqref{penalty}. In such cases, even if there are unexercised rights, a penalty of zero could still be possibly incurred if the terminal price exceeds the strike.
\end{rmk}
Starting with the interval $[\tau_{M},\tau_{M+1}]$, then by the terminal condition, we have
\begin{equation*}
\mathcal{V}(\tau_{M+1},Q(\tau_{M+1})) = G(S(T),Q(T)).
\end{equation*}
For each $i\in\{0,\cdots,M\}$ and $q\in\{0,\cdots,\tilde{L}\}$, let us denote
\begin{equation*}
\Phi(\tau_i,S(\tau_i),Q(\tau_i),q) =  q\left(K-S(\tau_i)\right)^+ \cdot 1_{\{Q(\tau_i)\ge q\}}(\tau_i),
\end{equation*}
and further
\begin{equation*}
1_i = 1_{\{ Q(\tau_i) \ge q \}}(\tau_i).
\end{equation*}
Clearly $\Phi(\tau_i,S(\tau_i),Q(\tau_i),q)$ is \color{black}$\sF_{\tau_i}^S$\color{black}-adapted for any $i\in\{0,\cdots,M\}$. Then for $Q(\tau_M) \in \{1,\cdots,L\}$,
\begin{equation*}
\mathcal{V}(\tau_{M},Q(\tau_{M})) = \sup_{q\in \{0,\cdots,\tilde{L}\}} E_{\sF^S_{\tau_M}} \big\{ \Phi(\tau_{M}, S(\tau_M), Q(\tau_M),q)+e^{-\mu\Delta \tau_M}\mathcal{V}(\tau_{M+1},Q(\tau_M)-q\cdot 1_M)\big\}\label{M}.
\end{equation*}
So recursively, for any $i\in\{0,\cdots,M\}$, we have
\begin{align}
\mathcal{V}(\tau_{i},Q(\tau_{i})) = &\sup_{q\in \{0,\cdots,\tilde{L}\}} E_{\sF^S_{\tau_i}} \left\{ \Phi(\tau_{i}, S(\tau_{i}), Q(\tau_{i}),q)+e^{-\mu\Delta \tau_{i}}\mathcal{V}(\tau_{i+1},Q(\tau_{i})-q\cdot 1_{i}) \right\} \nonumber
\\
= &\sup_{q\in \{0,\cdots,\tilde{L}\}} \left\{ \Phi(\tau_{i}, S(\tau_{i}), Q(\tau_{i}),q)+e^{-\mu\Delta \tau_{i}}E_{\sF^S_{\tau_i}}\mathcal{V}(\tau_{i+1},Q(\tau_{i})-q\cdot 1_i) \right\} \label{i}.
\end{align}

\subsubsection{Neural network approximation of random functions} \label{NNARF}

Note that in \eqref{i}, the randomness of $\mathcal{V}(\tau_{i+1},Q(\tau_{i})-q\cdot 1_i)$ is introduced by the $S_{\tau_{i+1}}$-dependence only, which, given $S_{\tau_i}$, follows a specific distribution that is complicated to assess due to the path-dependence. Even though we could apply Monte-Carlo method to obtain its numerical expectation and variance, it is hard to obtain the distribution itself. 
The problem then becomes how to compute the conditional expectation in \eqref{i}. Due to the path-dependence, it is also difficult to apply the conventional tree-based dynamic programming approach. In this work, we use the following neural network approach to approximate the above conditional expectation
\begin{align}
\mathcal{NN}^{\Theta^{i,l-q}}(S(0),S(dt),\cdots,S(\tau_{i})) \approx E_{\sF^S_{\tau_i}} \mathcal{V}(\tau_{i+1},S(\tau_{i+1}),l-q),
\label{approx-nn}
\end{align}
where $l\in\{1,\cdots,L\}$ denotes the remaining number of rights at time $\tau_i$ and $q\in\{1\cdots,\tilde{L}\land l\}$ denotes the number of rights exercised at time $\tau_i$.

For each time interval $[\tau_i,\tau_{i+1})$, the $l$ and $q$ mentioned above, we introduce a feedforward neural network with input dimension $d_0^{i}$ and output dimension $d_1^{i}$. Here, for instance, we take time step size $dt=\frac{1}{365}$, $d_1^i=1$, and thus, the dimension $d_0^i=365\tau_i +1$ is changing with respect to $i$ and may be high, which prompts us to use deep neural networks for approximations. The following result can be easily extended to the cases with more hidden layers. Here for simplicity, we suppose that it has three layers with $m_n^{i}$ neurons, $n=0,1,2$. Let us consider the case with only one hidden layer and $m_0^{i}=d_0^{i}$, $m_2^{i} = d_1^{i}$. This neural network can be considered as a function from $\bR^{d_0^{i}}$ to $\bR^{d_1^{i}}$ defined by the composition of functions as
\begin{equation}\label{NNi}
x\in\bR^{d_0^{i}} \to A_2^{i,l-q} \circ \varrho \circ A_1^{i,l-q}(x) \in\bR^{d_1^{i}}.
\end{equation}
Here $A_1^{i,l-q}:\bR^{d_0^i}\to\bR^{m_1^i}$, $A_2^{i,l-q}:\bR^{m_1^i}\to\bR^{d_1^i}$ are affine transformations and can be defined by
\begin{equation*}
A_n^{i,l-q}(x) := \mathcal{W}_n^{i,l-q} x+ \beta_n^{i,l-q},
\end{equation*}
with $n=1,2$, where the matrix $\mathcal{W}_n^{i,l-q}$ and the vector $\beta_n^{i,l-q}$ are called weight and bias, respectively, for the $n$th layer of the network applied for the time interval $[\tau_i,\tau_{i+1})$ with $l$ rights remained and $q$ rights exercised at time $\tau_i$. For the last layer, we apply the identity function as the activation function, and the activation function $\varrho$, which is identical regarding time intervals and exercising strategies, is applied component-wise on the output of $A_1^{i,l-q}$.

For each time interval $[\tau_i,\tau_{i+1})$, remaining number of rights $l$ and exercise strategy $q$ at time $\tau_i$, the parameters for the neural network may be denoted by $\Theta^{i,l-q}:=\left(\mathcal{W}_n^{i,l-q},\beta_n^{i,l-q}\right)_{n=1,2}$. Given $d_0^i$, $d_1^i$, and $m_1^i$, the total number of parameters in a network is \color{black}$M^i = m_1^i*(d_0^i + 1)+(m_1^i+1)*d_1^i$\color{black}. Hence $\Theta^{i,l-q}\in\bR^{M^i}$. For each $i\in \{0,\cdots,M\}$, by $\mathcal{NN}^{\varrho,l-q}_{d_0^i,d_1^i;m_1^i}(\bR^{M^i})$, we denote the set of all neural network functions defined in \eqref{NNi}.

The above-defined deep learning neural networks may be utilized to approximate large class of unknown functions. Fundamental approximation results in finite-dimensional cases may be found in \cite{hornik1989multilayer, hornik1990universal} by Hornik, Stichcombe, and White. 
However, in \eqref{approx-nn}, random variables that may be thought of as functions defined on infinite-dimensional underlying spaces are to be approximated via functions in finite-dimensional spaces. As inspired in \cite{bayer2022pricing}, we may use the following approximating result for random variables.
 
\begin{prop}\label{nMar_nn_approx}
For each $l\in\{0,\cdots,L\}$, $q\in\{0,\cdots,\tilde{L}\}$ and any $T_0\in(0,T]$, and $d_1 \in\bN^+$, the function space
\begin{align*}
\{ \Psi(S(t_0),\cdots,S(t_{K})): &\Psi(\cdot)\in\mathcal{NN}^{\varrho,l-q}_{K+1,d_1;m_1}(\bR^{m_1*(K+2)+(m_1+1)*d_1}),\,m_1,\,K\in\mathbb N^+,\,\ 
\\
& \left. 0\le t_0 < t_1 < \cdots <t_K \le T_0\right\}
\end{align*}
is dense in $L^2(\Omega,\sF_{T_0}^S,\mathbb{P};\bR^{d_1})$ where $\varrho$ is continuous and nonconstant.
\end{prop}
Obviously, for neural networks with $\geq 2$ hidden layers, the above approximation results are still holding. The proof is analogous to that of \cite[Proposition 4.2]{bayer2022pricing} and thus omitted.

\begin{rmk}
In our work, for each $t\in [0,T]$ the value function $\mathcal{V}(t,Q(t))$ is $\sF_{t}^S$-measurable, with the $\sigma$-algebra $\sF_t^S:=\sigma\{S_u:u\le t\}$. Hence we need the above approximation for measurable functions. 
The process $S$ in Proposition \ref{nMar_nn_approx} may be an arbitrary adapted continuous process.
\end{rmk}

\subsubsection{Algorithm}
Denote by $\mathcal{\hat{V}}$ the simulated price for the swing option achieved by numerical approximations of the conditional expectation \color{black}under forward Euler method\color{black}. Here is a four-step algorithm for the discrete swing option pricing.

\textbf{Step 1}. Simulate $\mathcal{D}$ price paths via \eqref{euler_approx}. Each simulated path is denoted by $\hat{S}_T^k$ where $k\in\{1,\cdots,\mathcal{D}\}$. The time step size we choose is $dt: = \frac{1}{365}$.
\begin{rmk}
Each path is piecewise-simulated since the coefficients we obtained in Table \ref{price_parameters} varies regarding different time intervals.
\end{rmk}

\textbf{Step 2}. For the time interval $[\tau_M,\tau_{M+1}]$, with $j\in\{1,\cdots,L\}$ we have
\begin{align}
\Theta^{M,j} = \argmin_{\Theta^{M,j}} \left\{\sum_{k=1}^{\mathcal{D}} \left| \mathcal{NN}^{\Theta^{M,j}}\left(\hat{S}^k(t^0_{0}),\cdots,\hat{S}^k(t^M_{N_M})\right)-G(\hat{S}^k(T),j) \right|^2\right\}.
\end{align}
Then for each $l\in\{1,\cdots,L\}$,
\begin{equation*}
\mathcal{\hat{V}}(\tau_M,\hat{S}^k(\tau_M),l) = \sup_{q\in\{0,\cdots,l\land\tilde{L}\}} \left\{ \Phi(\tau_M,\hat{S}^k(\tau_M),l,q)+e^{-\mu\Delta \tau_M}\mathcal{NN}^{\Theta^{M,l-q}}(\hat{S}^k(t_0^0),\cdots,\hat{S}^k(t_{N_M}^M)) \right\}.
\end{equation*}

\textbf{Step 3}. For the time interval $[\tau_i,\tau_{i+1}]$, with $i\in\{1,\cdots,M-1\}$ and $j\in\{1,\cdots,L\}$ we have
\begin{align}
\Theta^{i,j} = \argmin_{\Theta^{i,j}} \left\{  \sum_{k=1}^{\mathcal{D}}\left|\mathcal{NN}^{\Theta^{i,j}}\left(\hat{S}^k(t^0_{0}),\cdots,\hat{S}^k(t^i_{N_i})\right)-\mathcal{\hat V}(\tau_{i+1},\hat{S}^k(\tau_{i+1}),j) \right|^2 \right\}.
\end{align}
Then for each $i\in\{1,\cdots,M-1\}$ and $l\in\{1,\cdots,L\}$,
\begin{equation*}
\mathcal{\hat{V}}(\tau_i,\hat{S}^k(\tau_i),l) = \sup_{q\in\{0,\cdots,l\land\tilde{L}\}} \left\{ \Phi(\tau_i,\hat{S}^k(\tau_i),l,q)+e^{-\mu\Delta \tau_i}\mathcal{NN}^{\Theta^{i,l-q}}(\hat{S}^k(t_0^0),\cdots,\hat{S}^k(t_{N_i}^i)) \right\}.
\end{equation*}

\textbf{Step 4}. For the time interval $[\tau_0,\tau_{1}]$, we have
\begin{equation}
\mathcal{\hat{V}}(\tau_0,\hat{S}(\tau_0),L) = \sup_{q\in\{0,\cdots,\tilde{L}\}} \left\{ \Phi(\tau_0,\hat{S}(\tau_0),L,q)+\frac{e^{-\mu\Delta \tau_0}}{\mathcal{D}}\sum_{k=1}^{\mathcal{D}}\mathcal{\hat V}(\tau_1,\hat{S}^k(\tau_1),L-q) \right\}.
\end{equation}

\subsubsection{Convergence analysis}
This section is devoted to a convergence analysis for the deep-learning approach introduced above. The discussions are based on the model \eqref{model_new}.



Let us adopt the same notation $\hat{\mathcal{V}}$ as introduced above as the optimal pricing achieved by numerical approximations, either by a neural network approach, with $\hat{S}:=\{\hat{S} (t)\}_{t\in [0,T]}$ being the price path approximated by the forward Euler scheme or other numerical approximation schemes. Under the deep learning approximation approach, we have the following convergence result.
\begin{thm}\label{thm-convergence-analysis}
For each $p\ge 1$,
\begin{align}
& \max_{t\in \mathbb{I}\cup\{\tau_{M+1}\}}  E \left[ \max_{l\in\{0,1,\cdots,L\}} \left| \mathcal{V}(t,l)-\hat{\mathcal{V}}(t,l) \right|^p \right]  \nonumber
\\
\le &\tilde C_{p,A,M,L,\tilde{L}}\Bigg( \sum_{m=1}^{M+1} \left( E \left| S(\tau_m) - \hat{S}(\tau_m) \right|^p \right) \land (2K)^p \nonumber
\\ 
&\text{ }\text{ }\text{ }\text{ }\text{ }\text{ }\text{ }\text{ }\text{ }\text{ }\text{ }\text{ }+E \Bigg[ \sum_{m=0}^{M} \max_{l\in\{0,1,\cdots,L\} }  \left| E_{\sF^S_{\tau_{m}}}\hat{\mathcal{V}}(\tau_{m+1},l) - \mathcal{NN}^{\Theta^{m,l}}(\color{black}\hat{S}(t_0^0),\cdots,\hat{S}(t_{N_m}^m)\color{black}) \right|^p\Bigg] \Bigg). \label{eq-thm-convergence-analysis}
\end{align}
\end{thm}

\begin{proof}
Let us adopt the partition introduced in Section \ref{section_NN} for each $i\in\{0,1,\cdots,M+1\}$. First, by \eqref{penalty} we have
\begin{align*}
&\max_{t\in \mathbb{I}\cup\{\tau_{M+1}\}}   E\left[ \max_{l\in\{0,1,\cdots,L\} } \left| \mathcal{V}(t,l)-\hat{\mathcal{V}}(t,l) \right|^p  \right]
\\
=&\max_{i\in \{0,1,\cdots,M+1\}} E \left[ \max_{l\in\{0,1,\cdots,L\} } \left| \mathcal{V}(\tau_i,l)-\hat{\mathcal{V}}(\tau_i,l) \right|^p \right]
\\
=&\max\left\{ E \left[ \max_{l\in\{0,1,\cdots,L\} } \left| \mathcal{V}(T,l)-\hat{\mathcal{V}}(T,l) \right|^p \right] , \max_{i\in \{0,1,\cdots,M\}} E \left[ \max_{l\in\{0,1,\cdots,L\} } \left| \mathcal{V}(\tau_i,l)-\hat{\mathcal{V}}(\tau_i,l) \right|^p \right] \right\}
\\
\le &E \left[ \max_{l\in\{0,1,\cdots,L\} } \left| \mathcal{V}(T,l)-\hat{\mathcal{V}}(T,l) \right|^p \right] + \max_{i\in \{0,1,\cdots,M\}} E \left[ \max_{l\in\{0,1,\cdots,L\} } \left| \mathcal{V}(\tau_i,l)-\hat{\mathcal{V}}(\tau_i,l) \right|^p \right]
\\
= &E \left[ \max_{l\in\{0,1,\cdots,L\} } \left| G(S(T),l)-G(\hat S(T),l) \right|^p \right] + \max_{i\in \{0,1,\cdots,M\}} E \left[ \max_{l\in\{0,1,\cdots,L\} } \left| \mathcal{V}(\tau_i,l)-\hat{\mathcal{V}}(\tau_i,l) \right|^p \right]
\\
\le &\left[ (AL)^p \left( E\left| S(T) - \hat{S}(T) \right|^p \right)\land (2K)^p \right] + \max_{i\in \{0,1,\cdots,M\}} E \left[ \max_{l\in\{0,1,\cdots,L\} } \left| \mathcal{V}(\tau_i,l)-\hat{\mathcal{V}}(\tau_i,l) \right|^p \right].
\end{align*}
Also, we have
\begin{align*}
&\max_{i\in \{0,1,\cdots,M\}} E \left[ \max_{l\in\{0,1,\cdots,L\} } \left| \mathcal{V}(\tau_i,l)-\hat{\mathcal{V}}(\tau_i,l) \right|^p \right]
\\
&\leq  \max_{i\in\{0,1,\cdots,M\}} E\Bigg[ \max_{l\in\{0,1,\cdots,L\} }  \max_{q\in\{0,1,\cdots,l\land\tilde{L}\} }
\\
&\text{ }\text{ }\text{ }\text{ }\text{ }\text{ }\text{ }\text{ }\text{ }\text{ }\text{ }\text{ }\text{ }\text{ }\text{ }\text{ }\text{ }\text{ }\text{ }\Big|  \Phi(\tau_i,S(\tau_i),l,q) + E_{\sF^S_{\tau_i}}\mathcal{V}(\tau_{i+1},l-q) -  \Phi(\tau_i,\hat{S}(\tau_i),l,q) - E_{\sF^S_{\tau_i}}\hat{\mathcal{V}}(\tau_{i+1},l-q)
\\
&\text{ }\text{ }\text{ }\text{ }\text{ }\text{ }\text{ }\text{ }\text{ }\text{ }\text{ }\text{ }\text{ }\text{ }\text{ }\text{ }\text{ }\text{ }\text{ } + E_{\sF^S_{\tau_i}}\hat{\mathcal{V}}(\tau_{i+1},l-q) - \mathcal{NN}^{\Theta^{i,l-q}}(\color{black}\hat{S}(t_0^0),\cdots,\hat{S}(t_{N_i}^i)\color{black}) \Big|^p\Bigg]
\\
&\le  C_p\max_{i\in\{0,1,\cdots,M\}} E \Bigg[ \max_{l\in\{0,1,\cdots,L\} } \max_{q\in\{0,1,\cdots,l\land\tilde{L}\} } \Bigg(\Bigg|  E_{\sF^S_{\tau_i}}\left[\mathcal{V}(\tau_{i+1},l-q) -\hat{\mathcal{V}}(\tau_{i+1},l-q) \right] \Bigg|^p 
\\
&\text{ }\text{ }\text{ }\text{ }\text{ }\text{ }\text{ }\text{ }\text{ }\text{ }\text{ }\text{ }\text{ }\text{ }\text{ }\text{ }\text{ }\text{ }\text{ }\text{ }\text{ }\text{ }+\tilde{L}^p\left(\left |S(\tau_i)-\hat{S}(\tau_i)\right |^p \land (2K)^p\right)
\\
&\text{ }\text{ }\text{ }\text{ }\text{ }\text{ }\text{ }\text{ }\text{ }\text{ }\text{ }\text{ }\text{ }\text{ }\text{ }\text{ }\text{ }\text{ }\text{ }\text{ }\text{ }\text{ } + \left| E_{\sF^S_{\tau_i}}\hat{\mathcal{V}}(\tau_{i+1},l-q) - \mathcal{NN}^{\Theta^{i,l-q}}(\color{black}\hat{S}(t_0^0),\cdots,\hat{S}(t_{N_i}^i)\color{black}) \right|^p\Bigg) \Bigg]
\\
&\le C_p\max_{i\in\{0,1,\cdots,M\}} E \Bigg[ \max_{l\in\{0,1,\cdots,L\} }  \Bigg(\Bigg|  E_{\sF^S_{\tau_i}}\left[\mathcal{V}(\tau_{i+1},l) -\hat{\mathcal{V}}(\tau_{i+1},l) \right] \Bigg|^p +\tilde{L}^p\left(\left |S(\tau_i)-\hat{S}(\tau_i)\right |^p \land (2K)^p\right)
\\
&\text{ }\text{ }\text{ }\text{ }\text{ }\text{ }\text{ }\text{ }\text{ }\text{ }\text{ }\text{ }\text{ }\text{ }\text{ }\text{ }\text{ }\text{ }\text{ }\text{ }\text{ }\text{ }\text{ }\text{ }\text{ }\text{ }\text{ }\text{ }\text{ }\text{ }\text{ }\text{ }\text{ }\text{ }\text{ }\text{ }\text{ }\text{ }\text{ } + \left| E_{\sF^S_{\tau_i}}\hat{\mathcal{V}}(\tau_{i+1},l) - \mathcal{NN}^{\Theta^{i,l}}(\color{black}\hat{S}(t_0^0),\cdots,\hat{S}(t_{N_i}^i)\color{black}) \right|^p\Bigg) \Bigg]
\\
& \le C_p\max_{i\in\{0,1,\cdots,M\}} E \Bigg[ \max_{l\in\{0,1,\cdots,L\} } \Bigg( \left|  \mathcal{V} (\tau_{i+1},l) -  \hat{\mathcal{V}}(\tau_{i+1},l)  \right|^p +\tilde{L}^p\left( \left |S(\tau_i)-\hat{S}(\tau_i)\right |^p \land (2K)^p\right)
\\
&\text{ }\text{ }\text{ }\text{ }\text{ }\text{ }\text{ }\text{ }\text{ }\text{ }\text{ }\text{ }\text{ }\text{ }\text{ }\text{ }\text{ }\text{ }\text{ }\text{ }\text{ }\text{ }\text{ }\text{ }\text{ }\text{ }\text{ }\text{ }\text{ }\text{ }\text{ }\text{ }\text{ }\text{ }\text{ }\text{ }\text{ }\text{ }\text{ }+ \left| E_{\sF^S_{\tau_i}}\left[\hat{\mathcal{V}}(\tau_{i+1},l)\right] - \mathcal{NN}^{\Theta^{i,l}}(\color{black}\hat{S}(t_0^0),\cdots,\hat{S}(t_{N_i}^i)\color{black}) \right|^p \Bigg) \Bigg]
\\
&\dots \\
&\le C_{p,A,M,L,\tilde{L}}\Bigg(  \max_{i\in\{0,1,\cdots,M\}}  \sum_{m=i}^{M+1} \left( E \left| S(\tau_m) - \hat{S}(\tau_m) \right|^p \right) \land (2K)^p
\\ 
&\text{ }\text{ }\text{ }\text{ }\text{ }\text{ }\text{ }\text{ }\text{ }\text{ }\text{ }\text{ }\text{ }\text{ }\text{ }+E \Bigg[ \sum_{m=i}^{M} \max_{l\in\{0,1,\cdots,L\} } \left| E_{\sF^S_{\tau_{m}}}\hat{\mathcal{V}}(\tau_{m+1},l) - \mathcal{NN}^{\Theta^{m,l}}(\color{black}\hat{S}(t_0^0),\cdots,\hat{S}(t_{N_m}^m)\color{black}) \right|^p\Bigg] \Bigg).
\end{align*}
Combining the above calculations yields that
\begin{align*}
&\max_{t\in \mathbb{I}\cup\{\tau_{M+1}\}}   E\left[ \max_{l\in\{0,1,\cdots,L\} } \left| \mathcal{V}(t,l)-\hat{\mathcal{V}}(t,l) \right|^p  \right]
\\
\le &\left[ (AL)^p \left( E\left| S(T) - \hat{S}(T) \right|^p \right)\land (2K)^p \right]
\\
& + C_{p,A,M,L,\tilde{L}}\Bigg(  \max_{i\in\{0,1,\cdots,M\}}  \sum_{m=i}^{M+1} \left( E \left| S(\tau_m) - \hat{S}(\tau_m) \right|^p \right) \land (2K)^p
\\ 
&\text{ }\text{ }\text{ }\text{ }\text{ }\text{ }\text{ }\text{ }\text{ }\text{ }\text{ }\text{ }\text{ }\text{ }\text{ }\text{ }\text{ }\text{ }\text{ }\text{ }\text{ }+E \Bigg[ \sum_{m=i}^{M} \max_{l\in\{0,1,\cdots,L\} } \left| E_{\sF^S_{\tau_{m}}}\hat{\mathcal{V}}(\tau_{m+1},l) - \mathcal{NN}^{\Theta^{m,l}}(\color{black}\hat{S}(t_0^0),\cdots,\hat{S}(t_{N_m}^m)\color{black}) \right|^p\Bigg] \Bigg)
\\
\le &\tilde C_{p,A,M,L,\tilde{L}}\Bigg(  \max_{i\in\{0,1,\cdots,M\}}  \sum_{m=1}^{M+1} \left( E \left| S(\tau_m) - \hat{S}(\tau_m) \right|^p \right) \land (2K)^p
\\ 
&\text{ }\text{ }\text{ }\text{ }\text{ }\text{ }\text{ }\text{ }\text{ }\text{ }\text{ }\text{ }\text{ }\text{ }\text{ }\text{ }\text{ }\text{ }\text{ }\text{ }\text{ }+E \Bigg[ \sum_{m=0}^{M} \max_{l\in\{0,1,\cdots,L\} } \left| E_{\sF^S_{\tau_{m}}}\hat{\mathcal{V}}(\tau_{m+1},l) - \mathcal{NN}^{\Theta^{m,l}}(\color{black}\hat{S}(t_0^0),\cdots,\hat{S}(t_{N_m}^m)\color{black}) \right|^p\Bigg] \Bigg).
\end{align*}
\end{proof}

\begin{rmk}
In the above convergence estimate \eqref{eq-thm-convergence-analysis}, the simulated paths $\hat{S}$ may be obtained via different numerical schemes and the rate of convergence is attainable accordingly for the first term on the right hand side, whereas for the second term, the error estimate is an open question in the fields of neural network approximations.
\end{rmk}

\subsubsection{A synthetic numerical example}
Let us consider a hypothetical discrete swing contract starting from January 5, 2019 with a total lifetime 30 days. The initial condition is chosen as the price data at the starting day. We stick to the following parameters for the proposed model in \eqref{model_new} so that they are consistent with the calibration results in Table \ref{price_parameters} and Table \ref{st_30}.
\begin{table}[!htbp]
\begin{center}
\begin{tabular}{|c|c|c|c|c|c|c|c|}
\hline
$\alpha$ & $r$    & $\lambda$ & $V_0$  & $V_1$  & $V_2$  & $\gamma_1$ & $\gamma_2$ \\ \hline
1.4561   & 5.2536 & 4.2638    & 2.1268 & 0.1361 & 4.0786 & 0.1040     & -0.3616    \\ \hline
\end{tabular}
\end{center}
\caption{Parameters chosen for the proposed model \eqref{model_new}} \label{parameters_sw_exam}
\end{table}

Let us generate 12000 log-price trajectories under the proposed model with the above choice of parameters. The process $P$ is obtained by real periodic component of the storage data and the initial conditions of $S$ and $X$ are made consistent with real price and storage datasets. The neural network structure with one hidden layer is introduced in Section \ref{NNARF} and the activation function is chosen to be the sigmoid function. We apply the Levenberg-Marquardt optimization algorithm (see \cite{levenberg1944method,marquardt1963algorithm}) as the optimizer. 

Let us further specify the error function as
\begin{align}
\text{error}:=\sum_{i=1}^{M}\sum_{j=1}^{L}  \sum_{k=1}^{\mathcal{D}}\left|\mathcal{NN}^{\Theta^{i,j}}\left(\hat{S}^k(t^0_{0}),\cdots,\hat{S}^k(t^i_{N_i})\right)-\mathcal{\hat V}(\tau_{i+1},\hat{S}^k(\tau_{i+1}),j) \right|^2.
\end{align}

Then by choosing $L=3, \tilde L=2, A=5,  M=4, \mu=0, K=3, \delta = 10^{-2}$ with a chosen time step size $dt=\frac{1}{365}$, we have the following table of pricing results for the designed discrete swing contract at time $t=0$ with mean 6.7770 and variance 0.0293.
\begin{table}[!htbp]
\begin{center}
\begin{tabular}{|c|c|c|c|c|c|c|c|c|c|c|}
\hline
Run   & 1      & 2      & 3      & 4      & 5      & 6      & 7      & 8      & 9      & 10     \\ \hline
Price & 6.7115 & 6.6773 & 6.9937 & 6.7024 & 7.1251 & 6.7017 & 6.9733 & 6.9444 & 6.6411 & 6.8255 \\ \hline
Run   & 11     & 12     & 13     & 14     & 15     & 16     & 17     & 18     & 19     & 20     \\ \hline
Price & 6.7845 & 6.6647 & 6.7334 & 6.7841 & 6.6155 & 6.5445 & 6.6065 & 6.5553 & 6.8714 & 7.0845 \\ \hline
\end{tabular}
\end{center}
\caption{20 runs of swing option time 0 prices} \label{sw_price_20}
\end{table}

Performance of neural network approximations, as shown by the left plot in Figure \ref{NN_performance_errorhistogram}, is measured in terms of mean square error (MSE) and shown in log scale. As we can see it rapidly decreased as the network was trained. The performance is shown for each of the training, validation, and test sets, which eventually have similar mean square error results. The final network is the one that performed best on the validation set, highlighted by the green circle. This shows that our deep learning neural network scheme is doing a good job approximating the associated conditional expectations and hence generating reasonably accurate swing option prices.
\begin{figure}[!htbp]
\begin{center}
\includegraphics[width=0.45\linewidth]{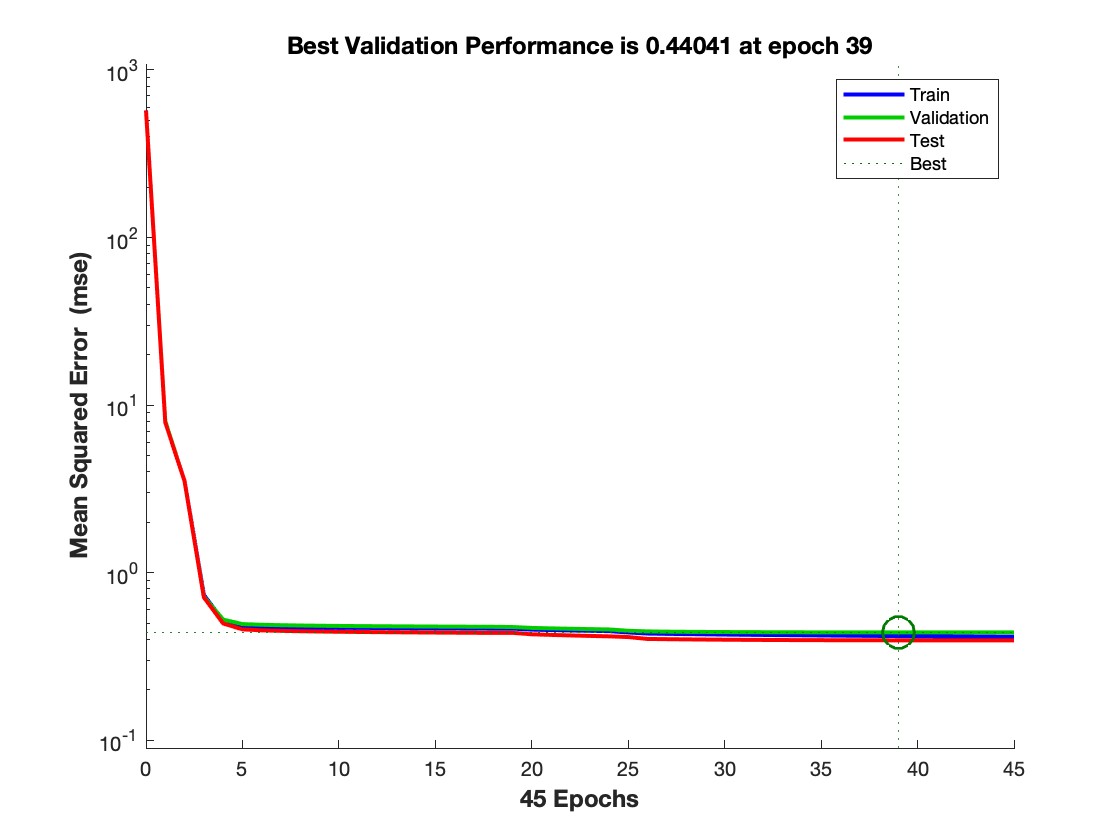}
\includegraphics[width=0.45\linewidth]{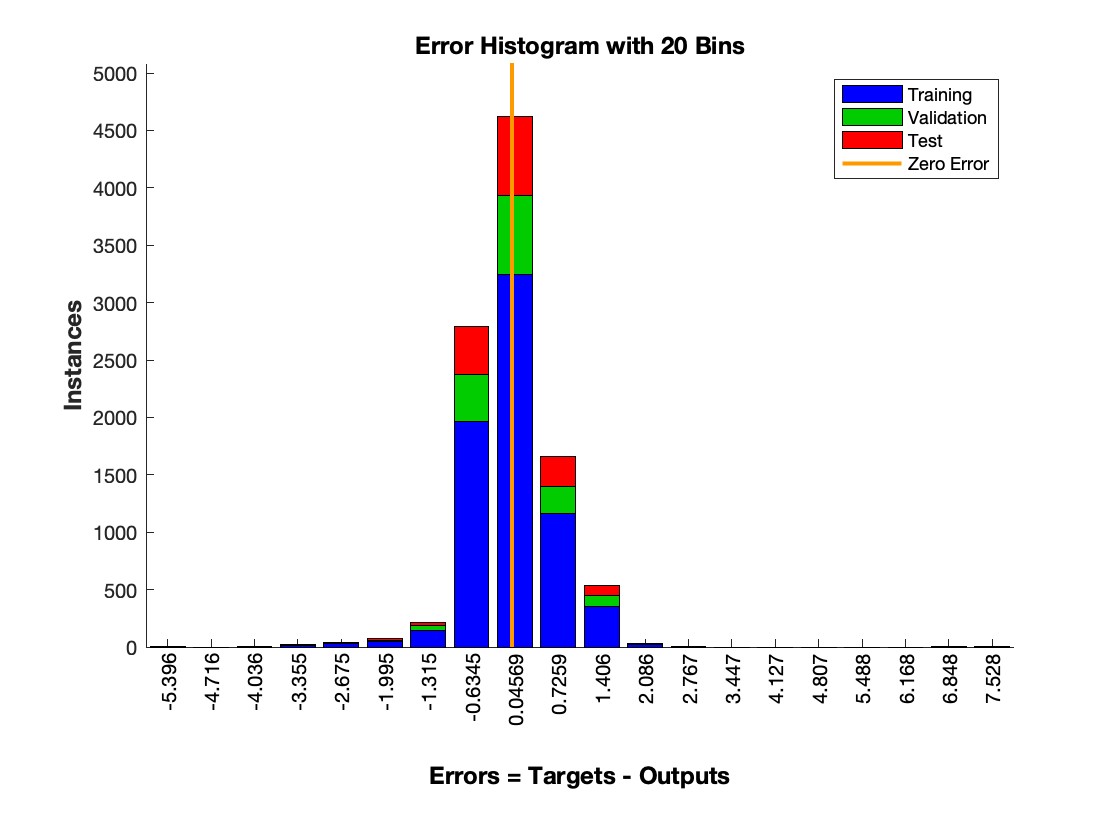}
\end{center}
\caption{Neural network performance \& error histogram with Sigmoid activation \& Levenberg-Marquardt optimizer}
\label{NN_performance_errorhistogram}
\end{figure}

A second measurement to the performance of the neural network is the error histogram as shown by the right plot of Figure \ref{NN_performance_errorhistogram}. Here ``Targets" means the conditional expectations of the value functions and ``Outputs" means the corresponding neural network approximations. This shows how the ``signed" errors are distributed. Typically most errors are near zero, with very few far away from that. Our error histogram demonstrates that point well and supports a good numerical result under the deep learning approximation scheme. As comparisons, we further attach the numerical results in Table \ref{result_comparison} with different activation functions and optimization schemes. The means and variances of the swing option prices at initial time are computed using results obtained in 20 runs of the MATLAB code, respectively. The performance and error histogram for one run in each case are included as well correspondingly. In general, the sigmoid activation function and Scaled Conjugate Gradient optimizer cases are performing slightly better considering lower errors. However, there is a slight ``gap" between the mean prices computed under the two chosen optimizers. It is still hard to precisely determine which case produces better computation result of the swing option price since we lack comparison from real-life option price data.
\begin{table}[!htbp]
\begin{center}
\begin{tabular}{|c|c|c|c|}
\hline
Activation Function & Optimizer                 & Mean   & Variance \\ \hline
Sigmoid             & Levenberg-Marquardt       & 6.7944 & 0.0267   \\ \hline
ReLU                & Levenberg-Marquardt       & 6.7770 & 0.0293   \\ \hline
Sigmoid             & Scaled Conjugate Gradient & 6.9941 & 0.0250   \\ \hline
ReLU                & Scaled Conjugate Gradient & 6.9724 & 0.0128   \\ \hline
\end{tabular}
\end{center}
\caption{Comparisons of numerical pricing results via different activations \& optimizers} \label{result_comparison}
\end{table}
\begin{figure}[!htbp]
\begin{center}
\includegraphics[width=0.45\linewidth]{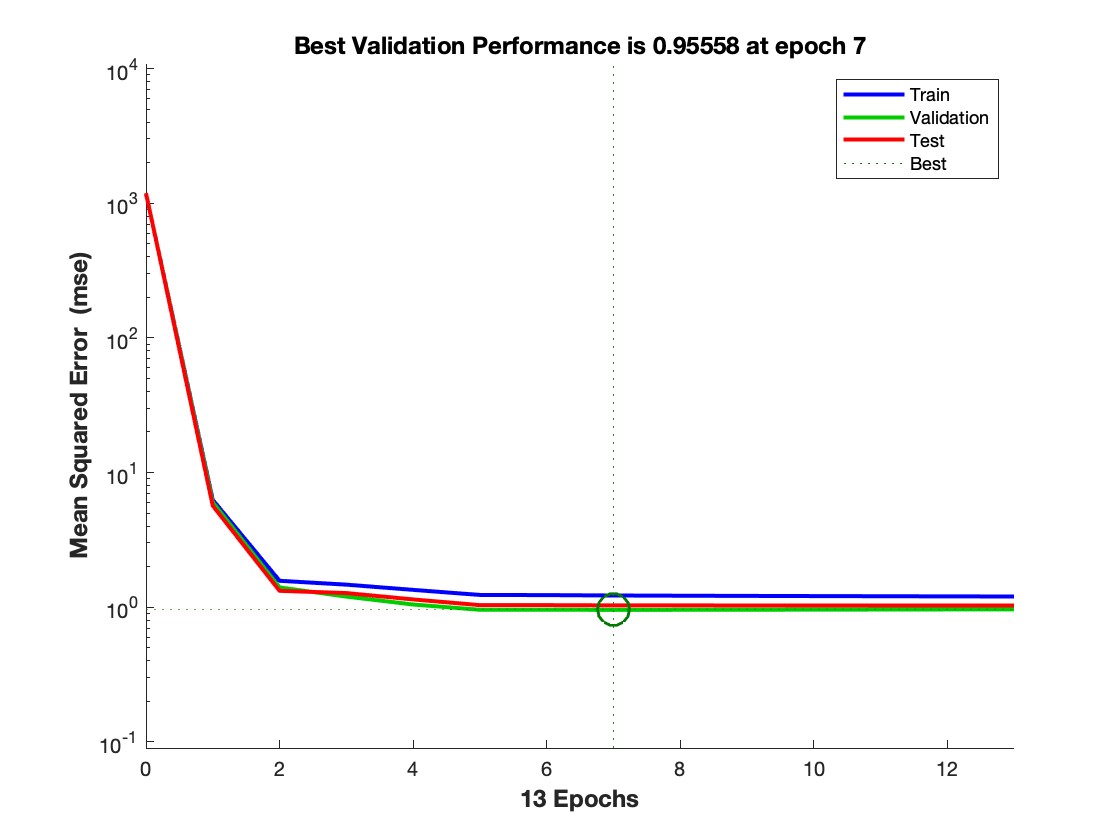}
\includegraphics[width=0.45\linewidth]{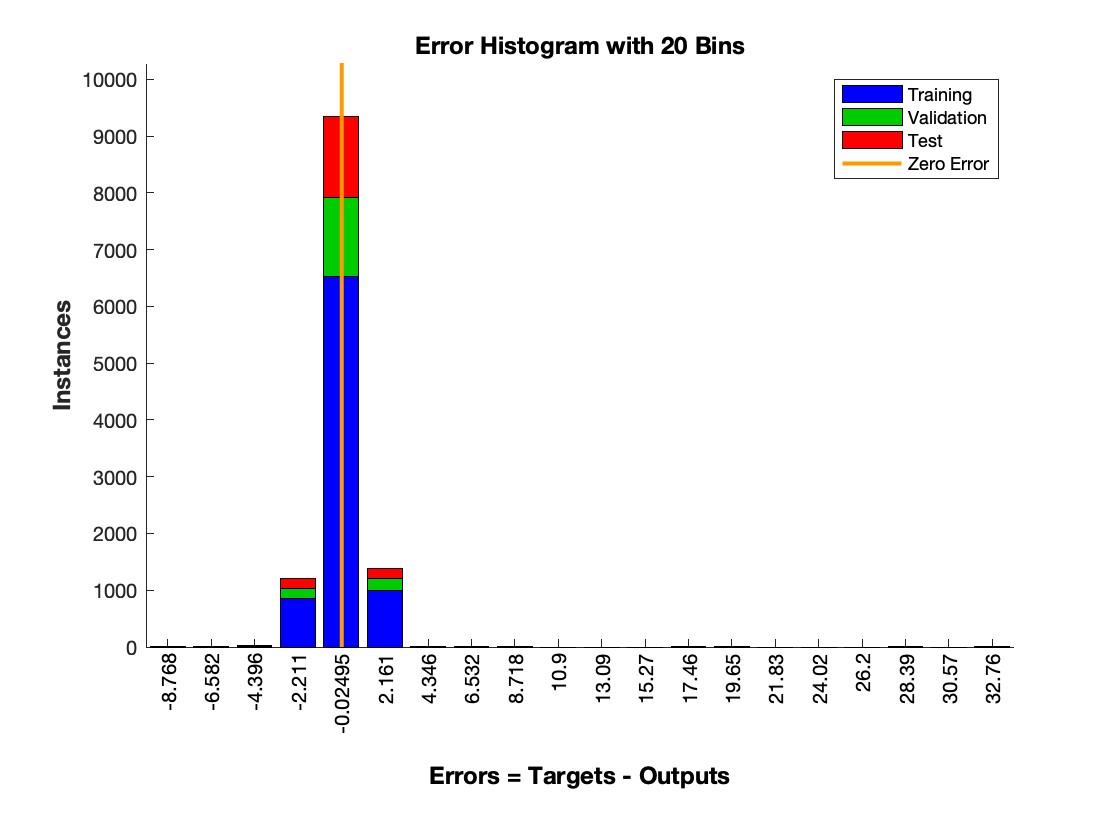}
\end{center}
\caption{Neural network performance \& error histogram with ReLU activation \& Levenberg-Marquardt optimizer}
\label{NN_performance_errorhistogram_ReLU_LM}
\end{figure}
\begin{figure}[!htbp]
\begin{center}
\includegraphics[width=0.45\linewidth]{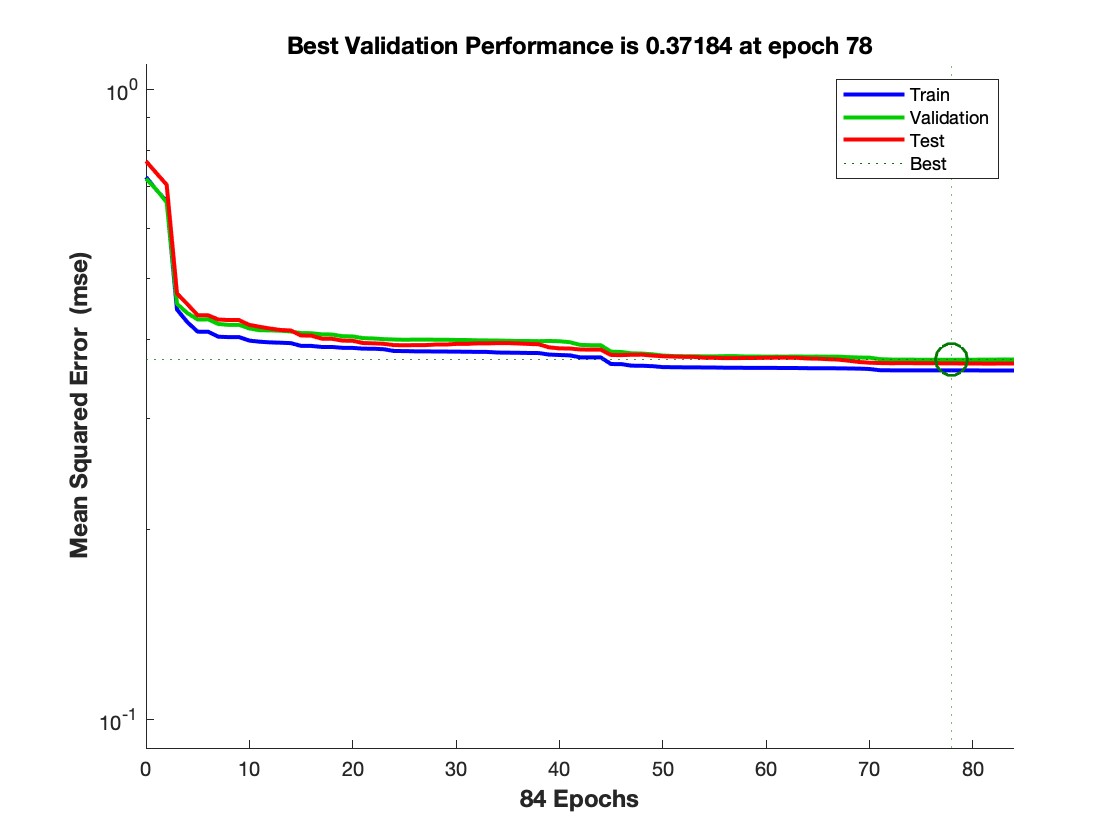}
\includegraphics[width=0.45\linewidth]{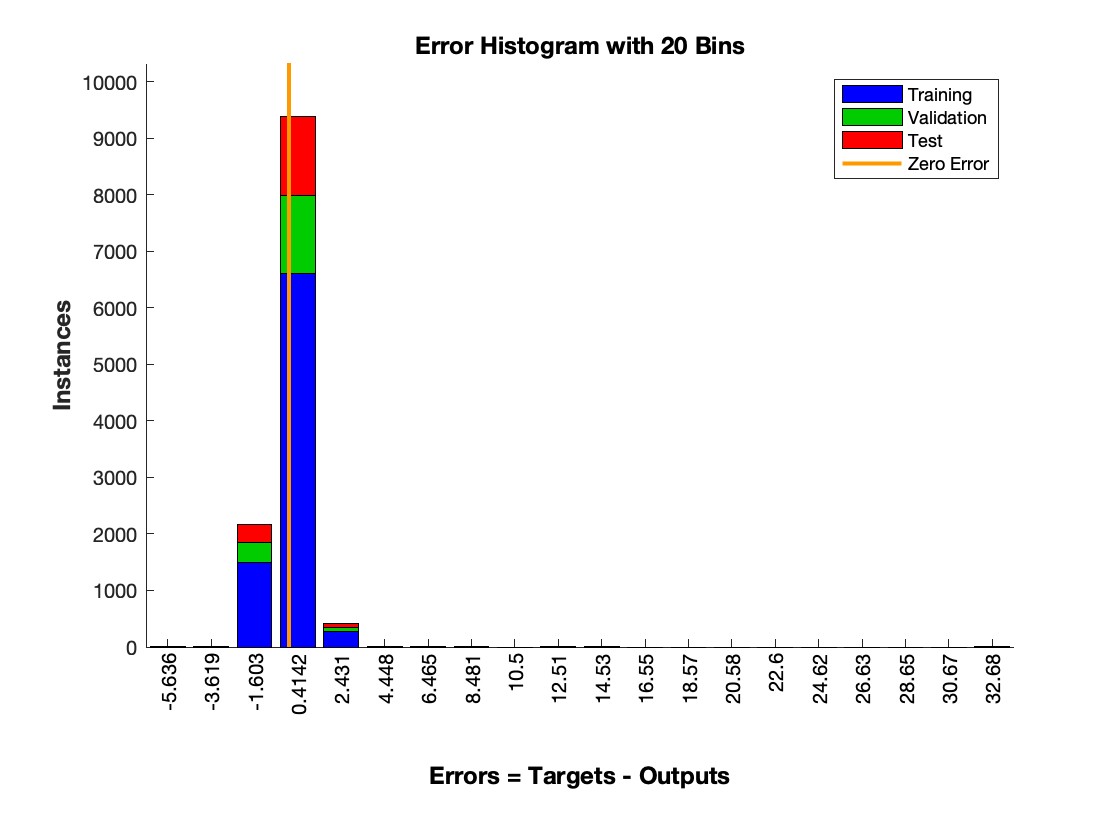}
\end{center}
\caption{Neural network performance \& error histogram with Sigmoid activation \& Scaled Conjugate Gradient optimizer}
\label{NN_performance_errorhistogram_ReLU_LM}
\end{figure}
\begin{figure}[!htbp]
\begin{center}
\includegraphics[width=0.45\linewidth]{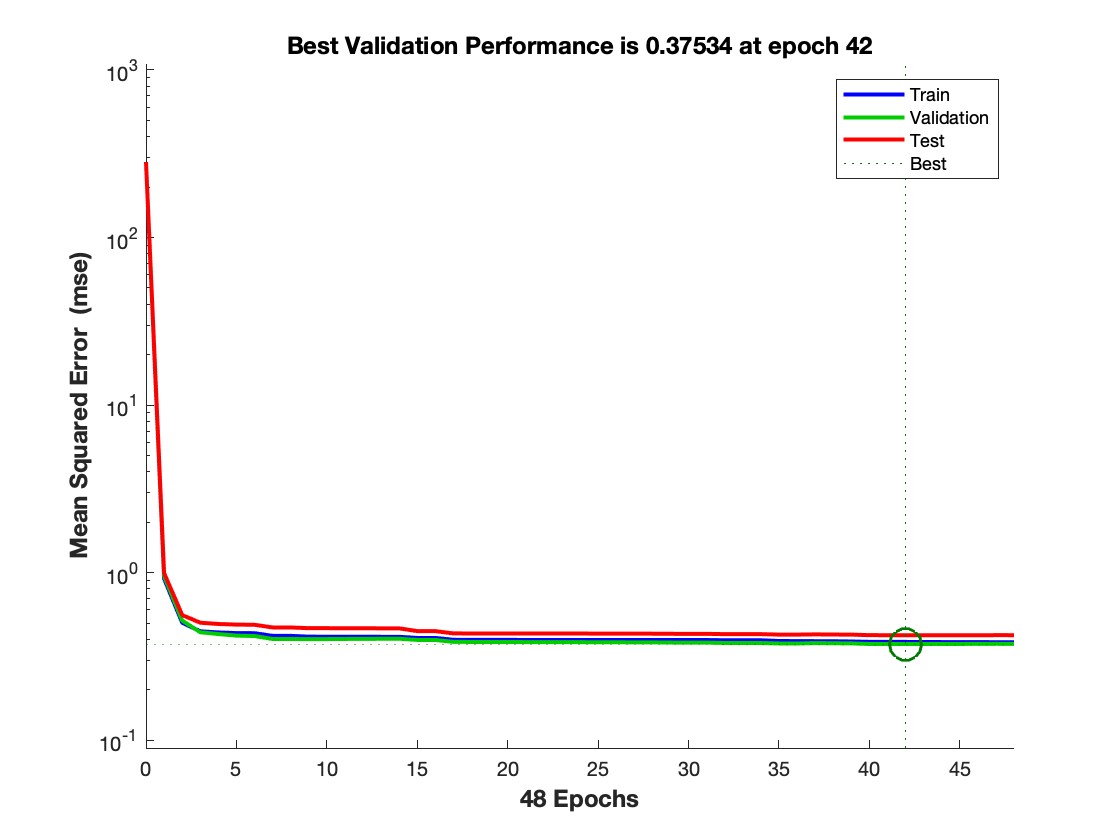}
\includegraphics[width=0.45\linewidth]{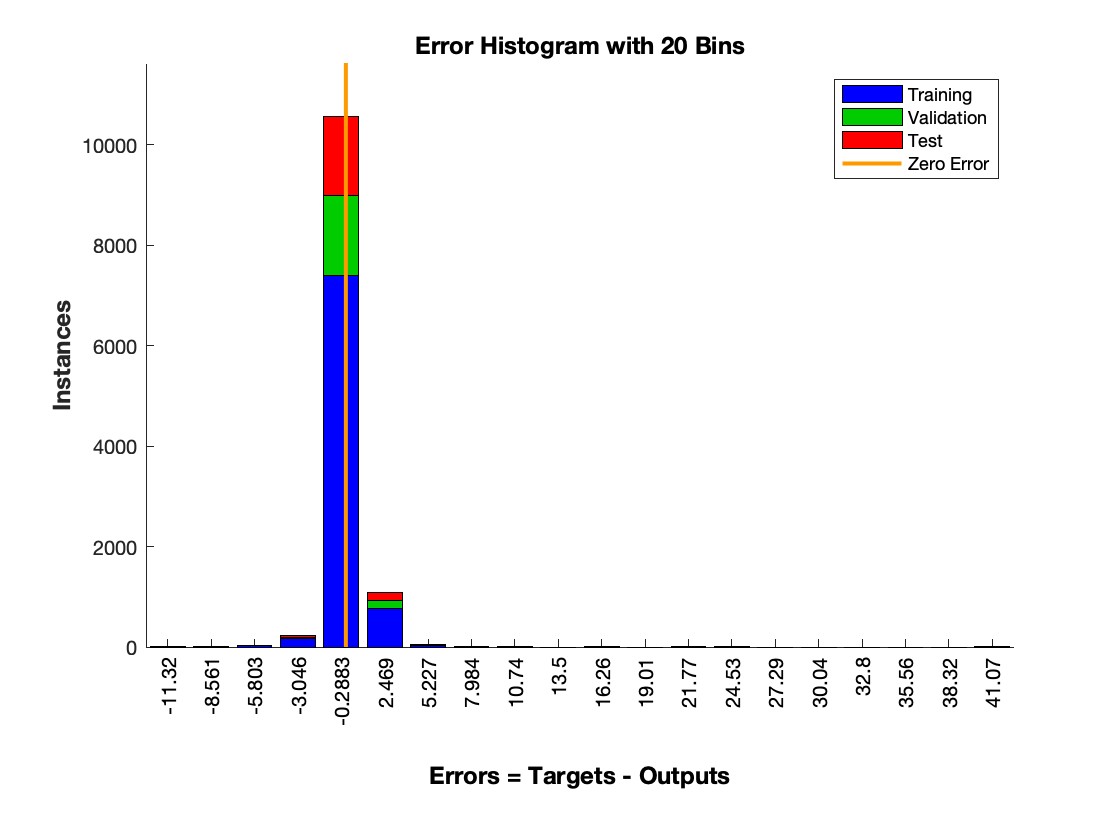}
\end{center}
\caption{Neural network performance \& error histogram with ReLU activation \& Scaled Conjugate Gradient optimizer}
\label{NN_performance_errorhistogram_ReLU_LM}
\end{figure}

\newpage
The performance of the deep-learning neural network approximations can be enhanced, in theory, by trying different optimizers and activation functions, or by increasing the complexity of the neural networks. However, this may lead to more training coefficients and thus may increase training time naturally. We will leave this discussion to future research work.


\pagebreak



 

  

\section{Appendix A}

\subsection{Consensus-Based Optimization}
For more detailed consensus-based optimization (CBO) scheme, please refer to \cite{cipriani2022zero,pinnau2017consensus}. Consider the following optimization problem:
 \begin{equation}
\argmax_{\overrightarrow{\theta}\in\bR^d}l(\overrightarrow{\theta}).
\end{equation}
Let $M\in\bN^+$ be the total number of particles in the system and $N\in\bN^+$ be the total number of CBO time steps. For each $k\in\{1,2,\cdots,M\}$, the $\bR^d$-valued stochastic process $\overrightarrow{\theta}^k_t$ denotes the position of the $k$-th particle at time $t$, with initial condition $\overrightarrow{\theta}_0^k$ following independent identical distributions for each $k$. Let $\{(W_t^k)_{t\ge 0}\}_{k=1}^M$ be $M$ independent $d$-dimensional Wiener processes. The consensus-based optimization dynamics may be given by
\begin{equation}\label{cbo}
d \overrightarrow{\theta}^k_t = a\cdot\left( \overrightarrow{\theta}^b_t(c^M) - \overrightarrow{\theta}^k_t  \right)dt + \sigma\cdot D\left( \overrightarrow{\theta}^b_t(c^M) - \overrightarrow{\theta}^k_t \right)d\overrightarrow{W}_t^k,
\end{equation}
where the weighted average is defined as
\begin{equation}\label{wa}
\overrightarrow{\theta}^b_t(c^M) := \frac{\int_{\bR^d}\! \overrightarrow{\theta}w_b^l(\overrightarrow{\theta})c^M(t,d\overrightarrow{\theta})}{\int_{\bR^d}\! w_b^l(\overrightarrow{\theta})c^M(t,d\overrightarrow{\theta})},
\end{equation}
with the empirical measure
\begin{equation}\label{measure}
c^M(t,d\overrightarrow{\theta}) := \frac{1}{M} \sum_{k=1}^M \delta_{\overrightarrow{\theta}_t^k}(d\overrightarrow{\theta}),
\end{equation}
and the weight function
\begin{equation}\label{weight}
w_b^l(\overrightarrow{\theta}) := \exp\{ b\cdot l(\overrightarrow{\theta}) \}.
\end{equation}
\begin{rmk}
Here we use the notation for the diagonal matrix
\begin{equation*}
D(\overrightarrow{\theta}_t) := \text{diag}\left\{ (\overrightarrow{\theta}_t)_1, \cdots, (\overrightarrow{\theta}_t)_d \right\}\in\bR^{d\times d},
\end{equation*}
where $(\overrightarrow{\theta}_t)_j$ is the $j$-th component of the $d$-dimensional vector $\overrightarrow{\theta}_t$. $a>0$ is the acceleration coefficient, and $\sigma>0$ is the diffusion coefficient. The choice of the weight function \eqref{weight} comes from the well-known Laplace principle \cite{dembo1998zeitouni,miller2006applied}, which states that for any probability measure $c$, there holds
\begin{equation}
\lim_{b\to\infty} \frac{1}{b} \log{\int_{\bR^d}\! e^{b\cdot l(x)}c(dx)}  = \esssup_{x\in \text{supp} (c)} l(x).
\end{equation}
Thus for $b$ large enough, when $M,N\to\infty$ one expects that 
\begin{equation}
\overrightarrow{\theta}_t^b(c^M) \approx \arg\max\{ l(\overrightarrow{\theta}^1_t),\cdots, l(\overrightarrow{\theta}^M_t) \}.
\end{equation}

\end{rmk}
 
Obviously, the CBO time step size $dt$ may be different than the price process time step  and a minimization problem can be computed similarly.
It has been proven that CBO can guarantee global convergence under suitable assumptions (see \cite{carrillo2018analytical}), and it is a powerful and robust method for many high-dimensional non-convex optimization problems in machine learning (see \cite{carrillo2021consensus,fornasier2021consensus}).

\subsection{Proof of Lemma \ref{lemma_wp}}
\begin{proof}
By \eqref{V_new}, we have for any $t\in [0,\tau\land T]$,
\begin{align*}
|\sigma(t)|^2 &
\le 3C_{{\delta},V_0,V_1}^2 + 3V_2^2 |\overline{S}_{t}^{\alpha,\delta} - \overline{S}(0)| + 3\delta V_2^2
\\
&= C_{\delta ,V_0,V_1,V_2} + 3V_2^2 |\overline{S}_{t}^{\alpha,\delta} - \overline{S}(0)|
\\
&\le C_{\delta ,V_0,V_1,V_2} + C_{\alpha,\delta,\overline{S}(0),T,V_2} \left( 1+\sup_{s\in [0,t]}|\overline{S}(s\land\tau)| \right)
\\
&\le C_{\alpha,\delta ,\overline{S}(0),T,V_0,V_1,V_2} \left( 1+ \sup_{s\in [0,t]}|\overline{S}(s\land\tau)| \right),
\end{align*}
which justifies (i).

For (ii), by \eqref{S_new} and the fact $\lambda\ge 0$, we have for any $t\in [0,\tau\land T]$,
\begin{align*}
&|\overline{S}(t)|^2 
\nonumber\\
&\le 3|\overline{S}(0)|^2 e^{-2\lambda t} + 3\left| \int_0^t \! e^{\lambda(u-t)}\Big( r-\frac{1}{2}|\sigma(u)|^2 \Big) du\right|^2  + 3\left( \int_0^t \! e^{\lambda(u-t)}\sigma(u) dW(u) \right)^2
\\
&\le 3|\overline{S}(0)|^2 + 3t\int_0^t \! \left( |r|+|\sigma(u)|^2 \right)^2 du + 3\left( \int_0^t \! e^{\lambda(u-t)}\sigma(u) dW(u) \right)^2
\\
&\le 3|\overline{S}(0)|^2 + 6r^2T^2 + 6\int_0^t \! |\sigma(u)|^4 du + 3\left( \int_0^t \! e^{\lambda(u-t)}\sigma(u) dW(u) \right)^2
\\
&\le C_{\alpha,\delta ,r,\overline{S}(0),T,V_0,V_1,V_2} \left(1+  \int_0^t \! \sup_{s\in [0,u]} |\overline{S}(s\land\tau)|^2 du\right) + 3\left( \int_0^t \! e^{\lambda(u-t)}\sigma(u) dW(u) \right)^2.
\end{align*}
By maximizing with respect to $t$, followed by taking expectations on both sides and applying the Burkholder-Davis-Gundy inequality, we obtain
\begin{align*}
&E\left[\sup_{s\in [0,t\land\tau]}|\overline{S}(s)|^2\right] \nonumber\\
&= E\left[\sup_{s\in [0,t]}|\overline{S}(s\land\tau)|^2\right]
\\
&\le  C+ C\cdot E\left[\int_0^t \! \sup_{s\in [0,u]} |\overline{S}(s\land\tau)|^2 du\right]
 + 3 E\left[\sup_{s\in [0,t\land\tau]}\left( \int_0^s \! e^{\lambda(u-s)}\sigma(u) dW(u) \right)^2\right]
\\
&\le C+ C\cdot E\left[\int_0^t \! \sup_{s\in [0,u]} |\overline{S}(s\land\tau)|^2 du\right]
	+ C \cdot E\left[ \int_0^t \! |\sigma(u\land\tau)|^2 du \right]
\\
&\le C_{\alpha,\delta, r,\overline{S}(0),T,V_0,V_1,V_2} + C\cdot E\left[\int_0^t \! \sup_{s\in [0,u]} |\overline{S}(s\land\tau)|^2 du\right]
 \\
 &\quad
+ C \cdot E\left[ \int_0^t \! C  \left( 1+ \sup_{s\in [0,u]}|\overline{S}(s\land\tau)| \right) du \right]
\\
&\le C + C\cdot E\left[\int_0^t \! \sup_{s\in [0,u]} |\overline{S}(s\land\tau)|^2 du\right],
\end{align*}
where the constant $C$ is independent of $\tau$. By Gr\"onwall's inequality, we have
\begin{align*}
E \sup_{t\in [0,\tau\land T]} | \overline{S}(t) |^2 \le &C\cdot e^{ {C}T},
\end{align*}
which yields the desired estimate.
\end{proof}

\subsection{Proof of Theorem \ref{thm_wp}}
\begin{proof}
\noindent \textbf{Step 1.}Take an arbitrary real number $B> |\overline S(0)|$. We first prove that there exists a  pair of processes $(\overline S, X)$ such that $(\overline S, X)(\cdot \wedge \tau_B)$ uniquely satisfies both \eqref{est-thm} and the stochastic differential equation \eqref{model_new}  on $[0,\tau_B]$ where we define the stopping time
 $$\tau_B := \inf\left\{ t>0: \left|\overline{S}(t)\right| \ge B \right\}.$$
    In view of the standard theories of stochastic differential equations (see \cite{oksendal2003stochastic,karatzas2012brownian,qiu2020stochastic,da2014stochastic,mohammed1984stochastic,yong1999stochastic} for instance), we only need to justify the associated uniform Lipschitz continuity of the coefficients.

Let $\hat S_1 $ and $\hat S_2$ be two pairs of continuous processes that are square-integrable on $\Omega\times[0,T]$ for any $T>0$. For $i = 1,2$, denote $\overline{S}_{i,t}^{\alpha,\delta}$ as the moving average of $\hat{S}_i(t)$ under $\delta$-approximation as defined in \eqref{S-delta}.  It follows that
\begin{align}
&\left| \overline{S}_{1,t}^{\alpha,\delta} - \overline{S}_{2,t}^{\alpha,\delta} \right| 
\nonumber \\
&\leq
  |1-\alpha|\cdot\left| \int_0^t\! \frac{\hat{S}_1(u)-\hat{S}_2(u)}{(t+\delta)^{1-\alpha}(t-u+\delta)^{\alpha}} du \right|
  + |\hat{S}_1(t)-\hat S_2(t)| \cdot 1_{\{\alpha=1\}}
\nonumber\\
&
\le   \max_{s\in[0,t]}| \hat{S}_{1}(s) - \hat{S}_{2}(s) | \cdot \left | \int_0^t\! \frac{1-\alpha}{(t+\delta)^{1-\alpha}(t-u+\delta)^{\alpha}} du \right |
+ |\hat{S}_1(t)-\hat S_2(t)|
\nonumber\\
&
= \max_{s\in[0,t]}| \hat{S}_{1}(s) - \hat{S}_{2}(s) |  \cdot \left| 1 - \left( \frac{\delta}{t+\delta} \right)^{1-\alpha} \right|
+ |\hat{S}_1(t)-\hat S_2(t)|
\nonumber\\
&
\le  \max_{s\in[0,t]}| \hat{S}_{1}(s) - \hat{S}_{2}(s) | 
  \cdot \max\left\{ \left(\frac{T+\delta}{\delta}\right)^{\alpha -1} +1, \, 2 \right\}
\nonumber\\
&
= C_{\alpha,\delta,T} \cdot \max_{s\in[0,t]}| \hat{S}_{1}(s) - \hat{S}_{2}(s) | . \label{est1-prf-thm}
\end{align}
 
 Let us introduce a truncated version of $(\tilde R,\overline S^{\alpha,\delta})$: for $i=1,2$,
 \begin{align}
 \tilde{R}_{B,i}(t) :&= (1-\alpha)\int_0^t \! \frac{(-B)\vee\hat{S}_i(u) \wedge{B}-(-B)\vee\hat{S}_i(t)\wedge B}{(t+\delta)^{1-\alpha}(t-u+\delta)^{\alpha}} du,\label{tildeRR_new}
 \\
  \overline{S}^{\alpha,\delta}_{B,i,t} 
  :&= (1-\alpha)\int_0^t \! \frac{(-B)\vee\hat{S}_i(u) \wedge{B}}{(t+\delta)^{1-\alpha}(t-u+\delta)^{\alpha}} du
  + (-B)\vee\hat{S}_i(t)\wedge B 1_{\alpha=1}
  ,\label{SSS_new}
 \end{align}
 and accordingly, we define $R^+_B$ and $R^-_B$ by replacing $R$ with $\tilde R_B$ as in \eqref{R_new}. All these truncated versions satisfy Lipschitz continuity analogous to \eqref{est1-prf-thm}.
 
 By solving the ordinary differential equation \eqref{X_new} associated to $R^+_B$ and $R^-_B$, we obtain the corresponding deseasonalized storage $X_1$ and $X_2$. Substituting $(\hat S_1, X_1,  \overline{S}^{\alpha,\delta}_{B,1})$ and $(\hat S_2, X_2,  \overline{S}^{\alpha,\delta}_{B,2})$ into \eqref{V_new} results in the volatilities $\sigma^1$ and $\sigma^2$, respectively.
 
By It\^o's formula and the Cauchy-Schwarz inequality, for each $t\in [0,T]$, we have
\begin{align}
&d(X_1(t) - X_2(t))^2 \nonumber
\\
&= 2(X_1(t)-X_2(t)) [ \gamma_1 R_{B,1}^{+}(t)(1-X_1(t)-P(t)) - \gamma_2 R_{B,1}^{-}(t)(X_1(t)+P(t)) \nonumber
\\
&\text{ }\text{ }\text{ }\text{ }\text{ }\text{ }\text{ }\text{ }\text{ }\text{ }\text{ }\text{ }\text{ }\text{ }\text{ }\text{ }\text{ }\text{ }\text{ }\text{ }\text{ }\text{ }\text{ }-\gamma_1 R_{B,2}^{+}(t)(1-X_2(t)-P(t)) + \gamma_2 R_{B,2}^{-}(t)(X_2(t)+P(t))] dt \nonumber
\\
&= 2(X_1(t)-X_2(t)) [ \gamma_1 (R_{B,1}^{+}(t)-R_{B,2}^{+}(t))(1-X_1(t)-P(t)) - \gamma_1 R_{B,2}^{+}(t)(X_1(t)-X_2(t))  \nonumber
\\
&\text{ }\text{ }\text{ }\text{ }\text{ }\text{ }\text{ }\text{ }\text{ }\text{ }\text{ }\text{ }\text{ }\text{ }\text{ }\text{ }\text{ }\text{ }\text{ }\text{ }\text{ }\text{ }\text{ }-\gamma_2 (R_{B,1}^{-}(t)-R_{B,2}^{-}(t))(X_1(t)+P(t)) - \gamma_2 R_{B,2}^{-}(t)(X_1(t)-X_2(t))] dt \nonumber
\\
&\le 2|X_1(t)-X_2(t)| [ (|\gamma_1| + |\gamma_2|) |R_{B,1}^{+}(t)-R_{B,2}^{+}(t)| + C_{\alpha,\delta,B,T} (|\gamma_1|+|\gamma_2|) |X_1(t)-X_2(t)| ]dt \nonumber
\\
&\le  [C_{\alpha,\delta,\gamma_1,\gamma_2,B,T}| X_1(t) - X_2(t) |^2 + | R_{B,1}^{+}(t) - R_{B,2}^{+}(t) |^2]dt \nonumber
\\
&\le  \left[C_{\alpha,\delta,\gamma_1,\gamma_2,B,T}| X_1(t) - X_2(t) |^2 + | \tilde{R}_{B,1}(t) - \tilde{R}_{B,2}(t) |^2\right]\, dt \nonumber
\\
&\le  \left[C_{\alpha,\delta,\gamma_1,\gamma_2,B,T}| X_1(t) - X_2(t) |^2 + C_{\alpha}\left ( | \overline{S}_{B,1,t}^{\alpha,\delta}-\overline{S}_{B,2,t}^{\alpha,\delta} | + |\hat{S}_1(t) - \hat{S}_2(t)| \right)^2\right]dt \nonumber
\\
&\le  \left[C_{\alpha,\delta,\gamma_1,\gamma_2,B,T}| X_1(t) - X_2(t) |^2 +  \overline C_{\alpha,\delta,T} \max_{s\in[0,t]} | \hat{S}_{1}(s) - \hat{S}_{2}(s) |^2\right]\, dt, \nonumber
\end{align}
which by Gr\"onwall's inequality implies that
\begin{align}
\max_{u\in [0,t]} | X_1(u) - X_2(u) |^2 
&\le T\overline{C}_{\alpha,\delta,T} 
\max_{s\in[0,t]} | \hat{S}_{1}(s) - \hat{S}_{2}(s) |^2
\cdot e^{TC_{\alpha,\delta,\gamma_1,\gamma_2,B,T}} \nonumber
\\
&= \overline{C}_{\alpha,\delta,\gamma_1,\gamma_2,B,T} 
\max_{s\in[0,t]} | \hat{S}_{1}(s) - \hat{S}_{2}(s) |^2. \label{X_Lipschitz}
\end{align} 
 Meanwhile, for each $t\in [0,T]$, it holds that
\begin{align*}
&| \sigma^1(t) - \sigma^2(t) |
\\
\le& \frac{|V_1| \cdot \left| (X_2(t)+P(t))(1- (X_2(t)+P(t))) - (X_1(t)+P(t))(1- (X_1(t)+P(t))) \right|}{\left[ (X_1(t)+P(t))(1- (X_1(t)+P(t))) +  {\delta} \right] \left[ (X_2(t)+P(t))(1- (X_2(t)+P(t))) +  {\delta} \right]}
\\
&\quad  + |V_2 C_{\delta}| \left| \overline S_{B,1,t}^{\alpha,\delta} - \overline S_{B,2,t}^{\alpha,\delta} \right|
\\
\le & 
	C_{{\delta},V_1} \left| (X_1(t)+X_2(t))-1 \right| \cdot |X_1(t) - X_2(t)| + |V_2 C_{\delta}| \left| \overline S_{B,1,t}^{\alpha,\delta} - \overline S_{B,2,t}^{\alpha,\delta} \right|
\\
\le& {C}_{ {\delta},V_1} |X_1(t) - X_2(t)| + C_{\delta,V_2} \left| \overline S_{B,1,t}^{\alpha,\delta} - \overline S_{B,2,t}^{\alpha,\delta} \right|
\\
\le& C_{\alpha,\delta ,\gamma_1,\gamma_2,B,T,V_1,V_2} \max_{s\in[0,t]} | \hat{S}_{1}(s) - \hat{S}_{2}(s) |.
\end{align*} 
 Moreover, we have
\begin{align*}
\left| r - \frac{1}{2}|\sigma^1(t)|^2 - \lambda\hat S_1(t) - r + \frac{1}{2}|\sigma^2(t)|^2 + \lambda\hat S_2(t) \right|
&\le \frac{1}{2} | \sigma^1(t) + \sigma^2(t) |\cdot | \sigma^1(t) - \sigma^2(t) | + \lambda |\hat S_1(t) - \hat S_2(t)|
\\
&\le  {C}_{\alpha,\delta ,\lambda,\gamma_1,\gamma_2,B,T,V_0,V_1,V_2} \max_{s\in[0,t]} | \hat{S}_{1}(s) - \hat{S}_{2}(s) |.
\end{align*} 
By the standard theories of stochastic differential equations with Lipschitz coefficients, we obtain the existence and uniqueness of the strong solution denoted by $(\overline S, X)$ to \eqref{S_new} - \eqref{X_new} associated to truncated coefficients $(\tilde R_B,\overline S^{\alpha,\delta}_B, R^+_B, R^-_B)$. Obviously, for all such $B> |\overline S(0)|$, the stopped joint process $(\overline S, X)(\cdot \wedge \tau_B)$ uniquely satisfies both \eqref{est-thm} and the stochastic differential equation \eqref{model_new} associated to the original coefficients $(\tilde R,\overline S^{\alpha,\delta}, R^+, R^-)$  on $[0,\tau_B]$. 

\noindent\textbf{Step 2.} It is evident that $\tau_B$ is a.s. increasing in $B$. We prove that for any $T>0$ there holds  $ \bP(\lim_{B\rightarrow \infty} \tau_B>T)=1$ through which letting $B\rightarrow \infty$ one may extend the unique solution from $[0,\tau_B]$ to any finite interval $[0,T]$ and obtain the estimate \eqref{est-thm} by the monotone convergence theorem.  
To the contrary, suppose that for some $T>0$ there exist $\eps\in(0,1)$ and an increasing sequence $\{B_i\}_{i\in\bN^+}$ with $B_1>\overline{S}(0)$ and $\lim_{i\to\infty} B_i = \infty$ such that
\begin{equation*}
\mathbb{P}(\lim_{i\to\infty}\tau_{B_i} \le T) > \eps.
\end{equation*}
Then, for each $N>0$, we must have
\begin{equation*}
\mathbb{P}\left(\lim_{i\to\infty} \sup_{0\le t\le \tau_{B_i}\land T} |\overline{S}(t)| > N\right) > \eps.
\end{equation*}
By Chebyshev's inequality, monotone convergence theorem and Lemma \ref{lemma_wp}, we have
\begin{align*}
\eps < &\text{ }\mathbb{P}\left(\lim_{i\to\infty} \sup_{0\le t\le \tau_{B_i}\land T} |\overline{S}(t)| > N\right)
\\
\le &\text{ }\frac{E\left[\lim_{i\to\infty} \sup_{0\le t\le \tau_{B_i}\land T} |\overline{S}(t)|^2\right]}{N^2}
\\
= &\text{ }\frac{\lim_{i\to\infty}E\left[ \sup_{0\le t\le \tau_{B_i}\land T} |\overline{S}(t)|^2\right]}{N^2}
\\
= &\text{ }\frac{\lim_{i\to\infty}E\left[ \sup_{0\le t\le T} |\overline{S}(t\land\tau_{B_i})|^2\right]}{N^2}
\\
\le &\text{ }\frac{C}{N^2}
\\
\to &\text{ }0, \text{ as } N\to\infty,
\end{align*}
where the constant $C>0$, by Lemma \ref{lemma_wp}, does not depend on $N$ or the choice of sequence $\{B_i\}_{i\in\bN^+}$. Letting $N$ tend to infinity incurs a contradiction. The proof is complete.
\end{proof}
\begin{rmk}
It is worthwhile to point out that although the Lipschitz constant of $X$ introduced by \eqref{X_Lipschitz} relies on the choice of $B$, the mean behaviour of the log-price path is bounded by some constant that is independent of the choice of $B$, as indicated by Lemma~\ref{lemma_wp} $(ii)$. This fact allows us to extend the well-posedness of the solution to \eqref{S_new} - \eqref{V_new} from a localized region $[0,\tau_B]$, with a fixed threshold $B$, to a global case $[0,T]$ for each $T>0$.
\end{rmk}

\subsection{Proof of Theorem \ref{thm_dpp}}
\begin{proof}
Denote the right hand side by $\overline{\mathcal{V}}(\tau_i,Q(\tau_i))$. 
First, when $i = M,M-1$, the relation \eqref{DPP} holds true. Indeed for any $q\in U$ with $q\le Q(\tau_M)$, we have
\begin{equation}\label{easy}
\mathcal{V}(\tau_M,Q(\tau_{M})) 
\ge E_{\sF^S_{\tau_{M}}} \Big[ q(K-S(\tau_{M}))^+ + e^{-\mu(T-\tau_{M})}G(S(T),Q(T)) \Big].
\end{equation}
Since set $U$ has finite elements, the maximum may be reached at some $U$-valued $\sF_{\tau_M}$-measurable $q^*$ with $q^*\le Q(\tau_M)$. Taking essential supremum w.r.t. $q$ on both hand sides of \eqref{easy} leads to
\begin{align}
\mathcal{V}(\tau_M,Q(\tau_{M})) 
&=
E_{\sF^S_{\tau_{M}}} \Big[ q^*(K-S(\tau_{M}))^+  + e^{-\mu(T-\tau_{M})}G(S(T),Q(T)) \Big]
\\
&= \esssup_{q\in U,q\le Q(\tau_M)}E_{\sF^S_{\tau_{M}}} \Big[ q(K-S(\tau_{M}))^+ + e^{-\mu(T-\tau_{M})}G(S(T),Q(T)) \Big]\label{equality}
\\
&= \overline{\mathcal{V}}(\tau_{M},Q(\tau_{M})),
\end{align}
where the equality in \eqref{equality} is achieved by measurable selection theorem. Thus for each $q\in U$ with $q\le Q(\tau_{M-1})$, there exists $q^*_M$ that is $U$-valued $\sF_{\tau_M}$-measurable such that $q^*_M\le Q(\tau_M)$ and
\begin{align*}
\mathcal{V}(\tau_M-1,Q(\tau_{M-1})) 
\ge &E_{\sF^S_{\tau_{M-1}}} \Big[ q(K-S(\tau_{M-1}))^+ +  e^{-\mu(\tau_{M}-\tau_{M-1})}q^*_M(K-S(\tau_{M}))^+
\\
&\text{ }\text{ }\text{ }\text{ }\text{ }\text{ }\text{ }\text{ }\text{ }\text{ }\text{ }\text{ }+ e^{-\mu(\tau_{M}-\tau_{M-1})}e^{-\mu(T-\tau_{M})}G(S(T),Q(\tau_{M})+ q^*_M) \Big]
\\
= &E_{\sF^S_{\tau_{M-1}}} \Big[ q(K-S(\tau_{M-1}))^+ + e^{-\mu(\tau_{M}-\tau_{M-1})}\mathcal{V}(\tau_{M},Q(\tau_{M})) \Big].
\end{align*}
Again since $U$ has finite elements, there exist some $U$-valued $\sF_{\tau_{M-1}}$-measurable $q^*$ such that $q^*\le Q(\tau_{M-1})$ and by taking essential supremums, we have
\begin{align}
\mathcal{V}(\tau_M-1,Q(\tau_{M-1})) 
= &E_{\sF^S_{\tau_{M-1}}} \Big[ q^*(K-S(\tau_{M-1}))^+ + e^{-\mu(\tau_{M}-\tau_{M-1})}\mathcal{V}(\tau_{M},Q(\tau_{M})) \Big]
\\
= &\esssup_{q\in U,q\le Q(\tau_{M-1})}E_{\sF^S_{\tau_{M-1}}} \Big[ q(K-S(\tau_{M-1}))^+ + e^{-\mu(\tau_{M}-\tau_{M-1})}\mathcal{V}(\tau_{M},Q(\tau_{M})) \Big] \label{measurable_M-1}
\\
= &\overline{\mathcal{V}}(\tau_{M-1},Q(\tau_{M-1})),
\end{align}
where the equality \eqref{measurable_M-1} is again achieved under measurable selection theorem. Then recursively, for any $i\in\{0,\cdots,M\}$, any $q\in U$ with $q\le Q(\tau_i)$ and each $j\in\{i+1,\cdots,M\}$, there exist $U$-valued $\sF_{\tau_j}$-measurable $q^*_j$ such that $q^*_j\le Q(\tau_j)$ and
\begin{align*}
\mathcal{V}(\tau_i,Q(\tau_i)) 
\ge &E_{\sF^S_{\tau_i}} \left[ q(K-S(\tau_i))^+  + \sum_{j=i+1}^{M} e^{-\mu(\tau_{j}-\tau_i)}q^*_{j}(K-S(\tau_{j}))^+ \right.
\\
&\text{ }\text{ }\text{ }\text{ }\text{ }\text{ }\text{ }\text{ }\left.+ e^{-\mu(\tau_{M}-\tau_i)}e^{-\mu(T-\tau_{M})}G\left(S(T),Q(\tau_{i+1})+\sum_{j=i+1}^M q^*_j\right) \right]
\\
= &E_{\sF^S_{\tau_i}} \left[ q(K-S(\tau_i))^+ + e^{-\mu(\tau_{i+1}-\tau_i)}\mathcal{V}(\tau_{i+1},Q(\tau_{i+1})) \right].
\end{align*}
Since $U$ has finite elements, there exist some $U$-valued $\sF_{\tau_{i}}$-measurable $q^*$ such that $q^*\le Q(\tau_{i})$ and by taking essential supremums, we have
\begin{align}
\mathcal{V}(\tau_i,Q(\tau_i)) 
= &E_{\sF^S_{\tau_i}} \left[ q^*(K-S(\tau_i))^+ + e^{-\mu(\tau_{i+1}-\tau_i)}\mathcal{V}(\tau_{i+1},Q(\tau_{i+1})) \right]
\\
= &\esssup_{q\in U,q\le Q(\tau_i)}E_{\sF^S_{\tau_i}} \left[ q(K-S(\tau_i))^+ + e^{-\mu(\tau_{i+1}-\tau_i)}\mathcal{V}(\tau_{i+1},Q(\tau_{i+1})) \right]\label{measurable_i}
\\
= &\overline{\mathcal{V}}(\tau_i,Q(\tau_i)),
\end{align}
where measurable selection theorem ensures the equality \eqref{measurable_i}. This finishes the proof. 
\end{proof}

\pagebreak
\section{Appendix B}

\begin{table}[!htbp]
\begin{center}
\fontsize{6}{8}\selectfont
\begin{tabular}{|cc|cc|cc|cc|cc|cc|}
\hline
\multicolumn{2}{|c|}{01/2019 - 10/2019}               & \multicolumn{2}{c|}{11/2019 - 03/2020}               & \multicolumn{2}{c|}{03/2020 - 12/2020}               & \multicolumn{2}{c|}{01/2021 - 06/2021}               & \multicolumn{2}{c|}{06/2021 - 02/2022}               & \multicolumn{2}{c|}{02/2022 - 12/2022}               \\ \hline
\multicolumn{1}{|c|}{$\vec\gamma_1$} & $\vec\gamma_2$ & \multicolumn{1}{c|}{$\vec\gamma_1$} & $\vec\gamma_2$ & \multicolumn{1}{c|}{$\vec\gamma_1$} & $\vec\gamma_2$ & \multicolumn{1}{c|}{$\vec\gamma_1$} & $\vec\gamma_2$ & \multicolumn{1}{c|}{$\vec\gamma_1$} & $\vec\gamma_2$ & \multicolumn{1}{c|}{$\vec\gamma_1$} & $\vec\gamma_2$ \\ \hline
\multicolumn{1}{|c|}{-4.2495}        & 31.7320        & \multicolumn{1}{c|}{20.7150}        & 21.5680        & \multicolumn{1}{c|}{15.6280}        & 20.8810        & \multicolumn{1}{c|}{20.9070}        & 20.6320        & \multicolumn{1}{c|}{19.9690}        & 19.6080        & \multicolumn{1}{c|}{16.3130}        & 20.9270        \\ \hline
\multicolumn{1}{|c|}{-10.7360}       & 20.9280        & \multicolumn{1}{c|}{18.8930}        & 21.2050        & \multicolumn{1}{c|}{18.8600}        & -77.8260       & \multicolumn{1}{c|}{18.4710}        & 19.7990        & \multicolumn{1}{c|}{20.2110}        & 15.8310        & \multicolumn{1}{c|}{-6.2300}        & 20.6460        \\ \hline
\multicolumn{1}{|c|}{-2.9971}        & 21.3150        & \multicolumn{1}{c|}{15.7520}        & 268.8400       & \multicolumn{1}{c|}{15.2010}        & 15.3130        & \multicolumn{1}{c|}{-111.6000}      & 20.4740        & \multicolumn{1}{c|}{27.7220}        & 130.2300       & \multicolumn{1}{c|}{1.9062}         & 26.6100        \\ \hline
\multicolumn{1}{|c|}{-0.1601}        & 25.5400        & \multicolumn{1}{c|}{12.6940}        & 20.9760        & \multicolumn{1}{c|}{10.1770}        & 23.6480        & \multicolumn{1}{c|}{20.6760}        & 19.4480        & \multicolumn{1}{c|}{23.4750}        & 19.2540        & \multicolumn{1}{c|}{3.9236}         & 23.4890        \\ \hline
\multicolumn{1}{|c|}{8.8117}         & 16.6190        & \multicolumn{1}{c|}{-263.3300}      & 20.2650        & \multicolumn{1}{c|}{30.4140}        & 26.5670        & \multicolumn{1}{c|}{20.4800}        & 15.1830        & \multicolumn{1}{c|}{23.7360}        & 19.3920        & \multicolumn{1}{c|}{17.1000}        & 23.2990        \\ \hline
\multicolumn{1}{|c|}{20.8840}        & -0.5193        & \multicolumn{1}{c|}{19.5870}        & 26.6720        & \multicolumn{1}{c|}{19.7120}        & 25.2260        & \multicolumn{1}{c|}{19.1430}        & 12.2350        & \multicolumn{1}{c|}{21.7210}        & 23.5780        & \multicolumn{1}{c|}{11.8900}        & 24.0980        \\ \hline
\multicolumn{1}{|c|}{17.3950}        & 1.5014         & \multicolumn{1}{c|}{-0.7404}        & 20.1410        & \multicolumn{1}{c|}{25.2020}        & 18.9310        & \multicolumn{1}{c|}{19.3260}        & 14.3120        & \multicolumn{1}{c|}{21.2520}        & 21.1270        & \multicolumn{1}{c|}{16.0990}        & 15.9910        \\ \hline
\multicolumn{1}{|c|}{11.9020}        & 4.6034         & \multicolumn{1}{c|}{-4.1755}        & 21.5140        & \multicolumn{1}{c|}{26.5590}        & -4.5096        & \multicolumn{1}{c|}{23.2480}        & 19.9810        & \multicolumn{1}{c|}{21.3090}        & 20.4390        & \multicolumn{1}{c|}{20.7090}        & -1.1195        \\ \hline
\multicolumn{1}{|c|}{1.5376}         & 4.0445         & \multicolumn{1}{c|}{-20.2900}       & 20.7530        & \multicolumn{1}{c|}{18.1340}        & 1.5650         & \multicolumn{1}{c|}{27.1530}        & 18.5710        & \multicolumn{1}{c|}{21.7060}        & 18.3190        & \multicolumn{1}{c|}{22.6760}        & 4.2732         \\ \hline
\multicolumn{1}{|c|}{4.8134}         & 19.5610        & \multicolumn{1}{c|}{}               &                & \multicolumn{1}{c|}{24.9480}        & 15.2760        & \multicolumn{1}{c|}{26.4370}        & 20.2510        & \multicolumn{1}{c|}{21.8460}        & 9.4550         & \multicolumn{1}{c|}{20.7610}        & 15.5140        \\ \hline
\multicolumn{1}{|c|}{23.4300}        & 2.1935         & \multicolumn{1}{c|}{}               &                & \multicolumn{1}{c|}{16.8490}        & 28.9340        & \multicolumn{1}{c|}{22.7680}        & 20.3120        & \multicolumn{1}{c|}{21.7480}        & -88.1770       & \multicolumn{1}{c|}{19.9600}        & -1.4674        \\ \hline
\multicolumn{1}{|c|}{20.3570}        & -1.5164        & \multicolumn{1}{c|}{}               &                & \multicolumn{1}{c|}{7.1234}         & 19.8380        & \multicolumn{1}{c|}{23.5280}        & 19.8620        & \multicolumn{1}{c|}{22.6720}        & 18.2720        & \multicolumn{1}{c|}{25.5680}        & 1.7704         \\ \hline
\multicolumn{1}{|c|}{10.3870}        & 8.6681         & \multicolumn{1}{c|}{}               &                & \multicolumn{1}{c|}{0.0276}         & 20.1730        & \multicolumn{1}{c|}{21.0110}        & 21.2860        & \multicolumn{1}{c|}{21.3800}        & 20.8700        & \multicolumn{1}{c|}{2.0777}         & 4.0722         \\ \hline
\multicolumn{1}{|c|}{1.0714}         & 3.5023         & \multicolumn{1}{c|}{}               &                & \multicolumn{1}{c|}{13.9190}        & 22.1140        & \multicolumn{1}{c|}{20.2640}        & 19.2370        & \multicolumn{1}{c|}{21.2880}        & 20.0720        & \multicolumn{1}{c|}{-0.2361}        & 22.1470        \\ \hline
\multicolumn{1}{|c|}{6.8071}         & 6.7794         & \multicolumn{1}{c|}{}               &                & \multicolumn{1}{c|}{20.1250}        & 3.2929         & \multicolumn{1}{c|}{}               &                & \multicolumn{1}{c|}{17.3520}        & 22.2240        & \multicolumn{1}{c|}{21.3850}        & 1.5424         \\ \hline
\multicolumn{1}{|c|}{15.9020}        & 10.1810        & \multicolumn{1}{c|}{}               &                & \multicolumn{1}{c|}{18.6590}        & -1.1758        & \multicolumn{1}{c|}{}               &                & \multicolumn{1}{c|}{-198.7900}      & 21.8480        & \multicolumn{1}{c|}{21.0160}        & 1.3097         \\ \hline
\multicolumn{1}{|c|}{25.3100}        & -0.3850        & \multicolumn{1}{c|}{}               &                & \multicolumn{1}{c|}{29.1280}        & 2.6777         & \multicolumn{1}{c|}{}               &                & \multicolumn{1}{c|}{15.9400}        & 19.9870        & \multicolumn{1}{c|}{5.1443}         & 23.5100        \\ \hline
\multicolumn{1}{|c|}{21.0790}        & 1.4863         & \multicolumn{1}{c|}{}               &                & \multicolumn{1}{c|}{27.5180}        & 11.7750        & \multicolumn{1}{c|}{}               &                & \multicolumn{1}{c|}{12.6180}        & 22.5630        & \multicolumn{1}{c|}{-2.6051}        & 19.8450        \\ \hline
\multicolumn{1}{|c|}{16.4650}        & 3.6534         & \multicolumn{1}{c|}{}               &                & \multicolumn{1}{c|}{24.6790}        & 2.1986         & \multicolumn{1}{c|}{}               &                & \multicolumn{1}{c|}{}               &                & \multicolumn{1}{c|}{-1.6737}        & 22.3880        \\ \hline
\multicolumn{1}{|c|}{-1.1339}        & 21.7600        & \multicolumn{1}{c|}{}               &                & \multicolumn{1}{c|}{25.5440}        & 7.0787         & \multicolumn{1}{c|}{}               &                & \multicolumn{1}{c|}{}               &                & \multicolumn{1}{c|}{12.8060}        & 22.5110        \\ \hline
\multicolumn{1}{|c|}{3.1252}         & 18.9210        & \multicolumn{1}{c|}{}               &                & \multicolumn{1}{c|}{3.3441}         & 3.4713         & \multicolumn{1}{c|}{}               &                & \multicolumn{1}{c|}{}               &                & \multicolumn{1}{c|}{}               &                \\ \hline
\multicolumn{1}{|c|}{}               &                & \multicolumn{1}{c|}{}               &                & \multicolumn{1}{c|}{21.2200}        & 19.3200        & \multicolumn{1}{c|}{}               &                & \multicolumn{1}{c|}{}               &                & \multicolumn{1}{c|}{}               &                \\ \hline
\multicolumn{1}{|c|}{}               &                & \multicolumn{1}{c|}{}               &                & \multicolumn{1}{c|}{19.9840}        & 20.1230        & \multicolumn{1}{c|}{}               &                & \multicolumn{1}{c|}{}               &                & \multicolumn{1}{c|}{}               &                \\ \hline
\end{tabular}
\end{center}
\caption{CBO results for storage calibration($\mathcal{T}=14$)} \label{st_14}
\end{table}

\begin{table}[!htbp]
\begin{center}
\fontsize{6}{8}\selectfont
\begin{tabular}{|cc|cc|cc|cc|cc|cc|}
\hline
\multicolumn{2}{|c|}{01/2019 - 10/2019}               & \multicolumn{2}{c|}{11/2019 - 03/2020}               & \multicolumn{2}{c|}{03/2020 - 12/2020}               & \multicolumn{2}{c|}{01/2021 - 06/2021}               & \multicolumn{2}{c|}{06/2021 - 02/2022}               & \multicolumn{2}{c|}{02/2022 - 12/2022}               \\ \hline
\multicolumn{1}{|c|}{$\vec\gamma_1$} & $\vec\gamma_2$ & \multicolumn{1}{c|}{$\vec\gamma_1$} & $\vec\gamma_2$ & \multicolumn{1}{c|}{$\vec\gamma_1$} & $\vec\gamma_2$ & \multicolumn{1}{c|}{$\vec\gamma_1$} & $\vec\gamma_2$ & \multicolumn{1}{c|}{$\vec\gamma_1$} & $\vec\gamma_2$ & \multicolumn{1}{c|}{$\vec\gamma_1$} & $\vec\gamma_2$ \\ \hline
\multicolumn{1}{|c|}{-8.6233}        & 35.2140        & \multicolumn{1}{c|}{11.5500}        & 32.3600        & \multicolumn{1}{c|}{23.4110}        & -16.0250       & \multicolumn{1}{c|}{18.8140}        & 20.9910        & \multicolumn{1}{c|}{21.1570}        & 18.6690        & \multicolumn{1}{c|}{-3.1088}        & 19.6630        \\ \hline
\multicolumn{1}{|c|}{-0.0311}        & 19.4360        & \multicolumn{1}{c|}{1.1846}         & 28.0060        & \multicolumn{1}{c|}{13.1190}        & 25.6320        & \multicolumn{1}{c|}{-126.8500}      & 20.3600        & \multicolumn{1}{c|}{25.2120}        & 123.8500       & \multicolumn{1}{c|}{4.1150}         & 26.2770        \\ \hline
\multicolumn{1}{|c|}{22.9270}        & -2.2590        & \multicolumn{1}{c|}{-279.7700}      & 21.5660        & \multicolumn{1}{c|}{29.3820}        & 23.5530        & \multicolumn{1}{c|}{19.7340}        & 13.1320        & \multicolumn{1}{c|}{21.4340}        & 19.2540        & \multicolumn{1}{c|}{12.0770}        & 22.1940        \\ \hline
\multicolumn{1}{|c|}{8.5972}         & 4.4128         & \multicolumn{1}{c|}{3.2502}         & 20.2980        & \multicolumn{1}{c|}{22.7660}        & -2.5110        & \multicolumn{1}{c|}{22.9030}        & 16.9040        & \multicolumn{1}{c|}{20.4500}        & 20.9090        & \multicolumn{1}{c|}{20.3730}        & 2.1452         \\ \hline
\multicolumn{1}{|c|}{-0.2189}        & -0.5735        & \multicolumn{1}{c|}{18.1130}        & 21.9360        & \multicolumn{1}{c|}{21.1830}        & 6.1780         & \multicolumn{1}{c|}{48.2680}        & 21.1660        & \multicolumn{1}{c|}{20.7430}        & -39.4780       & \multicolumn{1}{c|}{16.1730}        & 2.7387         \\ \hline
\multicolumn{1}{|c|}{2.2780}         & 0.1127         & \multicolumn{1}{c|}{}               &                & \multicolumn{1}{c|}{1.2317}         & 20.7470        & \multicolumn{1}{c|}{22.3340}        & 19.8730        & \multicolumn{1}{c|}{23.2830}        & 17.4280        & \multicolumn{1}{c|}{0.0867}         & 0.0685         \\ \hline
\multicolumn{1}{|c|}{2.5506}         & 3.8125         & \multicolumn{1}{c|}{}               &                & \multicolumn{1}{c|}{13.4640}        & 4.0841         & \multicolumn{1}{c|}{18.9210}        & 20.0050        & \multicolumn{1}{c|}{14.3740}        & 21.5680        & \multicolumn{1}{c|}{2.3317}         & -0.9659        \\ \hline
\multicolumn{1}{|c|}{2.3516}         & 0.1747         & \multicolumn{1}{c|}{}               &                & \multicolumn{1}{c|}{24.3570}        & -0.8005        & \multicolumn{1}{c|}{}               &                & \multicolumn{1}{c|}{-94.7270}       & 26.7850        & \multicolumn{1}{c|}{5.0555}         & 0.7944         \\ \hline
\multicolumn{1}{|c|}{4.5296}         & 0.3976         & \multicolumn{1}{c|}{}               &                & \multicolumn{1}{c|}{47.6550}        & 7.1555         & \multicolumn{1}{c|}{}               &                & \multicolumn{1}{c|}{21.4110}        & 20.6590        & \multicolumn{1}{c|}{-3.1753}        & 21.5150        \\ \hline
\multicolumn{1}{|c|}{0.5081}         & 23.6580        & \multicolumn{1}{c|}{}               &                & \multicolumn{1}{c|}{20.0320}        & 5.1180         & \multicolumn{1}{c|}{}               &                & \multicolumn{1}{c|}{}               &                & \multicolumn{1}{c|}{19.2140}        & 22.2210        \\ \hline
\multicolumn{1}{|c|}{}               &                & \multicolumn{1}{c|}{}               &                & \multicolumn{1}{c|}{18.1560}        & 21.7750        & \multicolumn{1}{c|}{}               &                & \multicolumn{1}{c|}{}               &                & \multicolumn{1}{c|}{}               &                \\ \hline
\end{tabular}
\end{center}
\caption{CBO results for storage calibration($\mathcal{T}=30$)} \label{st_30}
\end{table}

\pagebreak
\begin{figure}[!htbp]
\begin{center}
\includegraphics[width=0.49\linewidth]{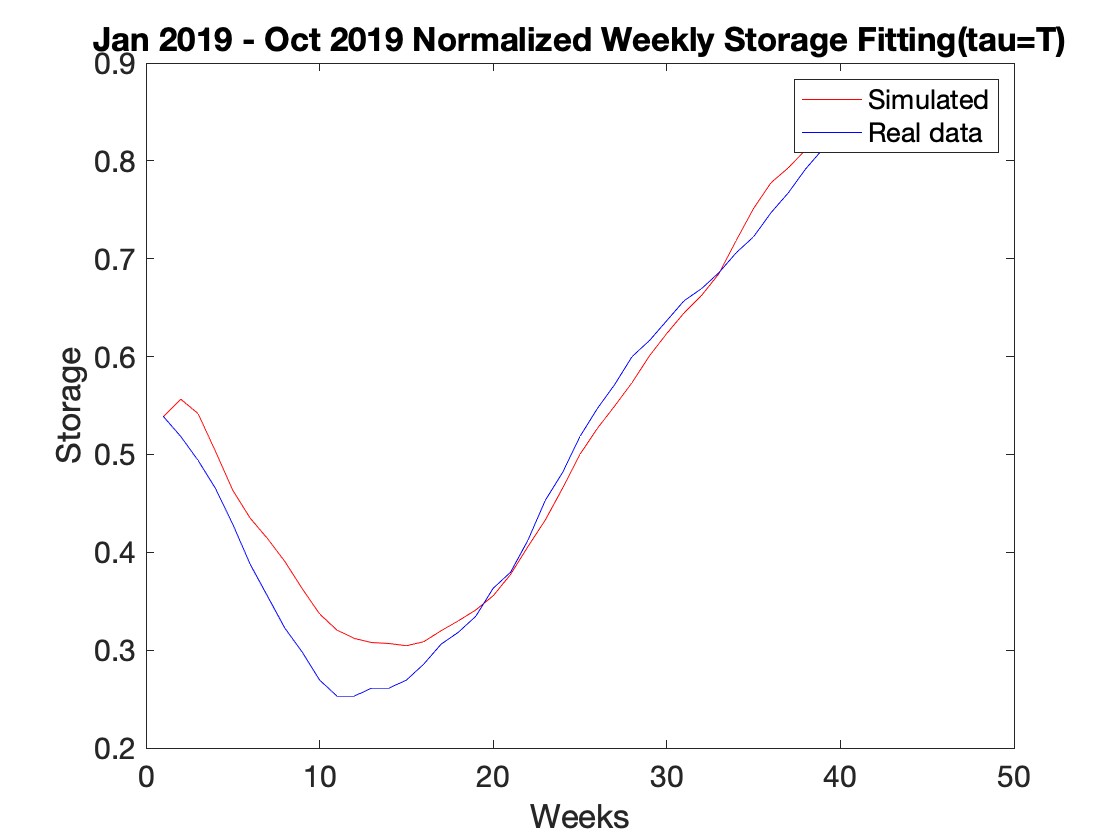}
\includegraphics[width=0.49\linewidth]{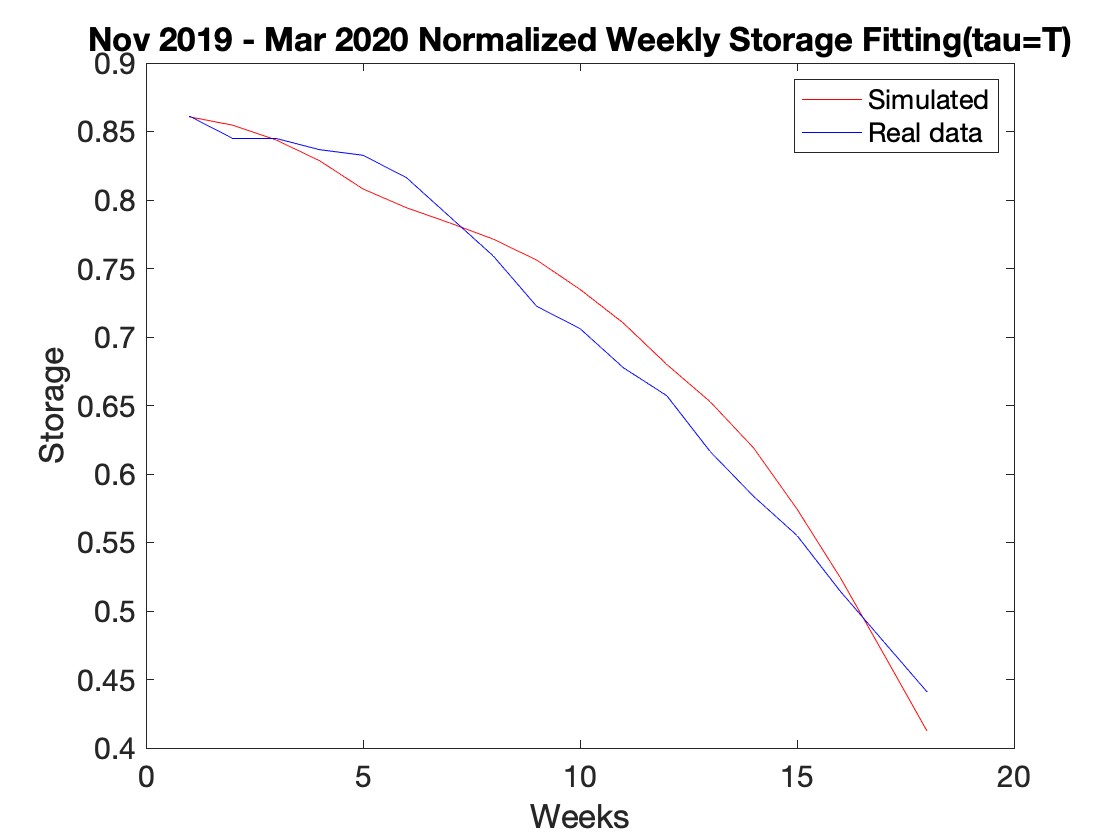}
\includegraphics[width=0.49\linewidth]{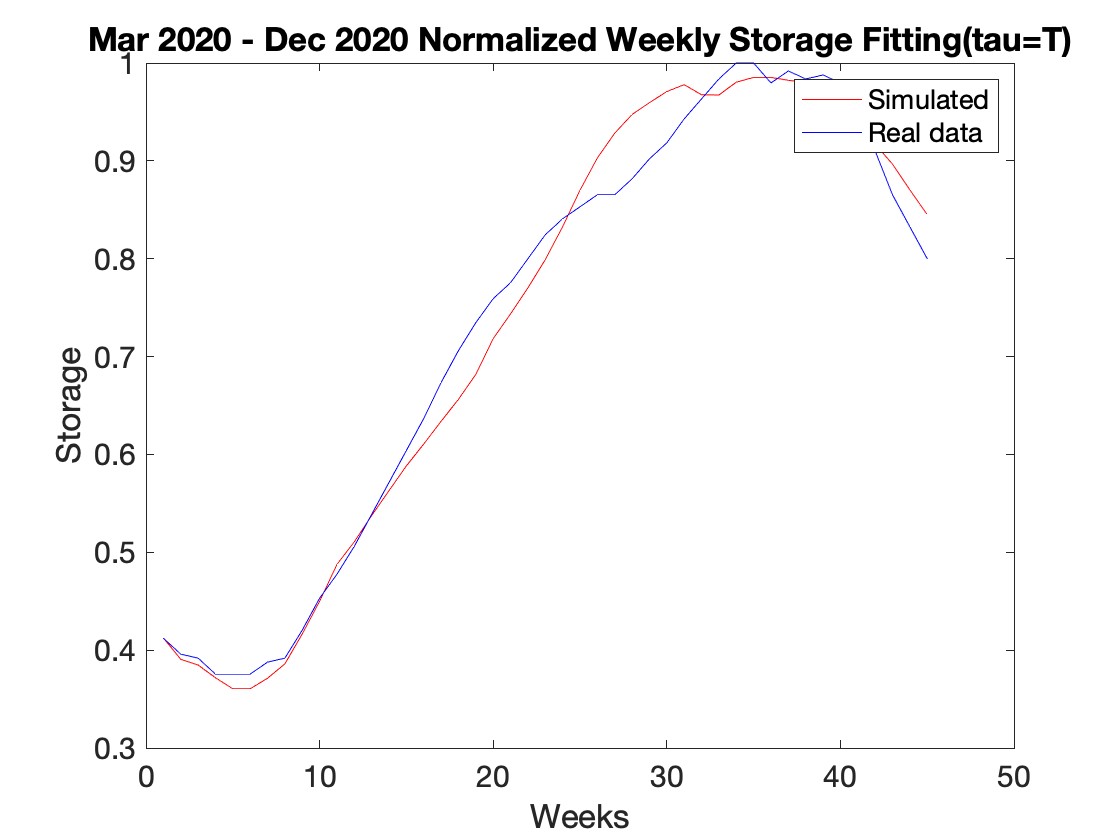}
\includegraphics[width=0.49\linewidth]{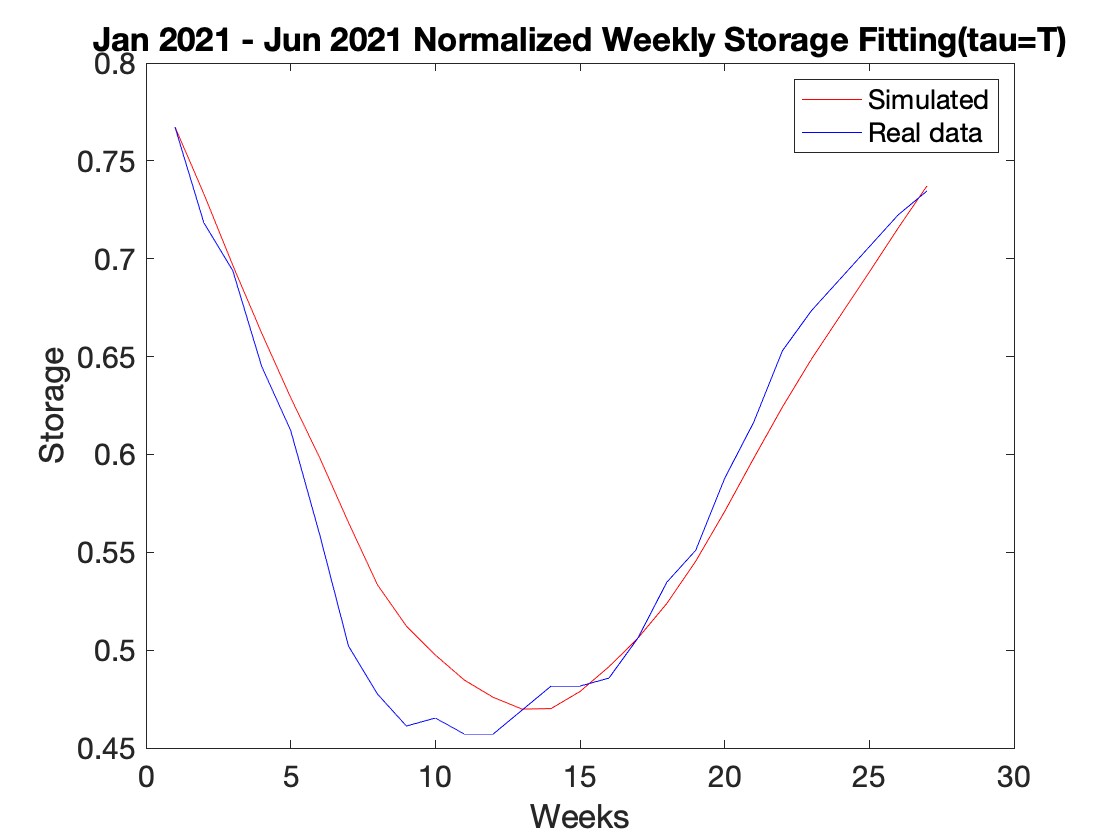}
\includegraphics[width=0.49\linewidth]{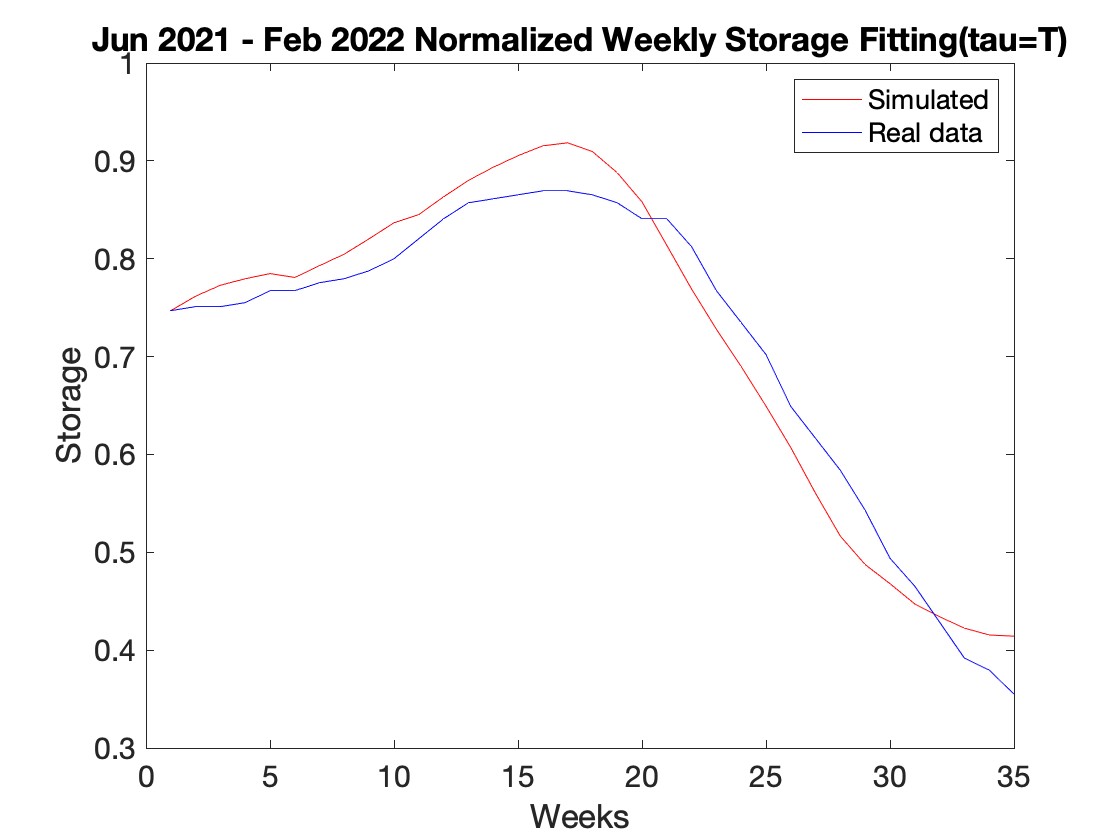}
\includegraphics[width=0.49\linewidth]{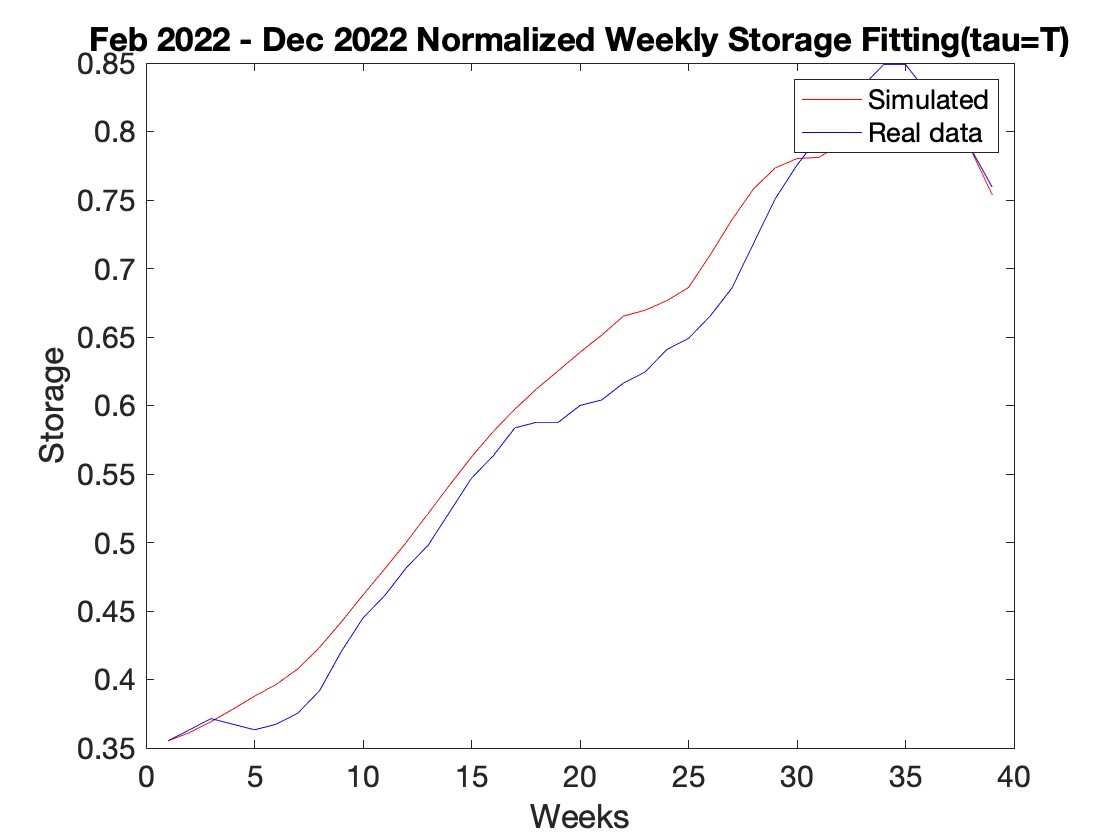}
\end{center}
\caption{Piecewise fit of normalized storage with $\mathcal{T}=T$}
\label{X_fit}
\end{figure}

\pagebreak

\begin{figure}[!htbp]
\begin{center}
\includegraphics[width=0.49\linewidth]{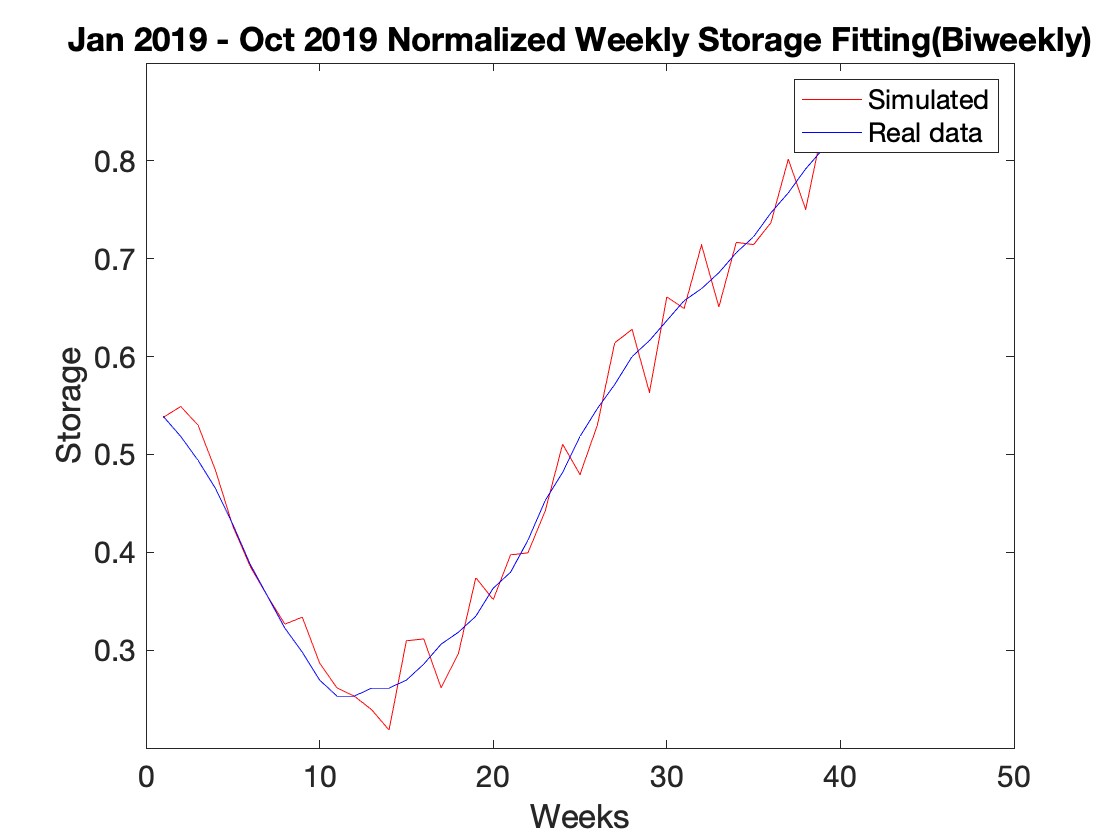}
\includegraphics[width=0.49\linewidth]{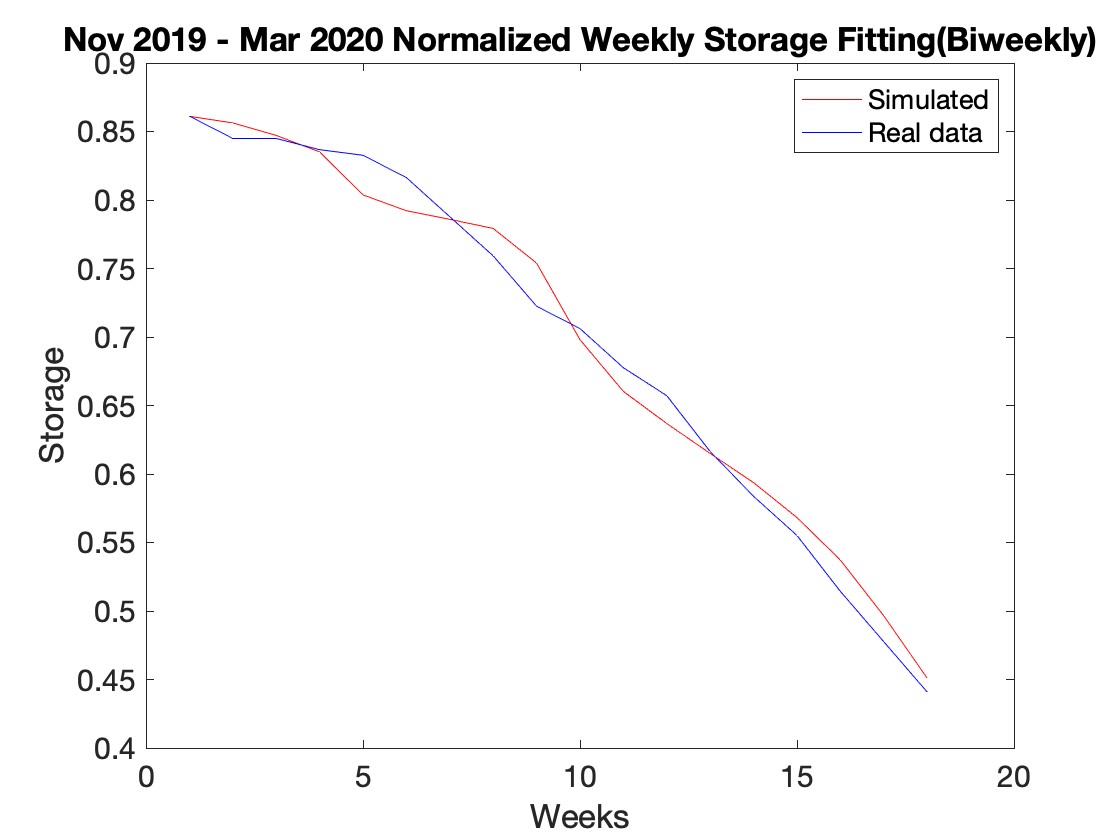}
\includegraphics[width=0.49\linewidth]{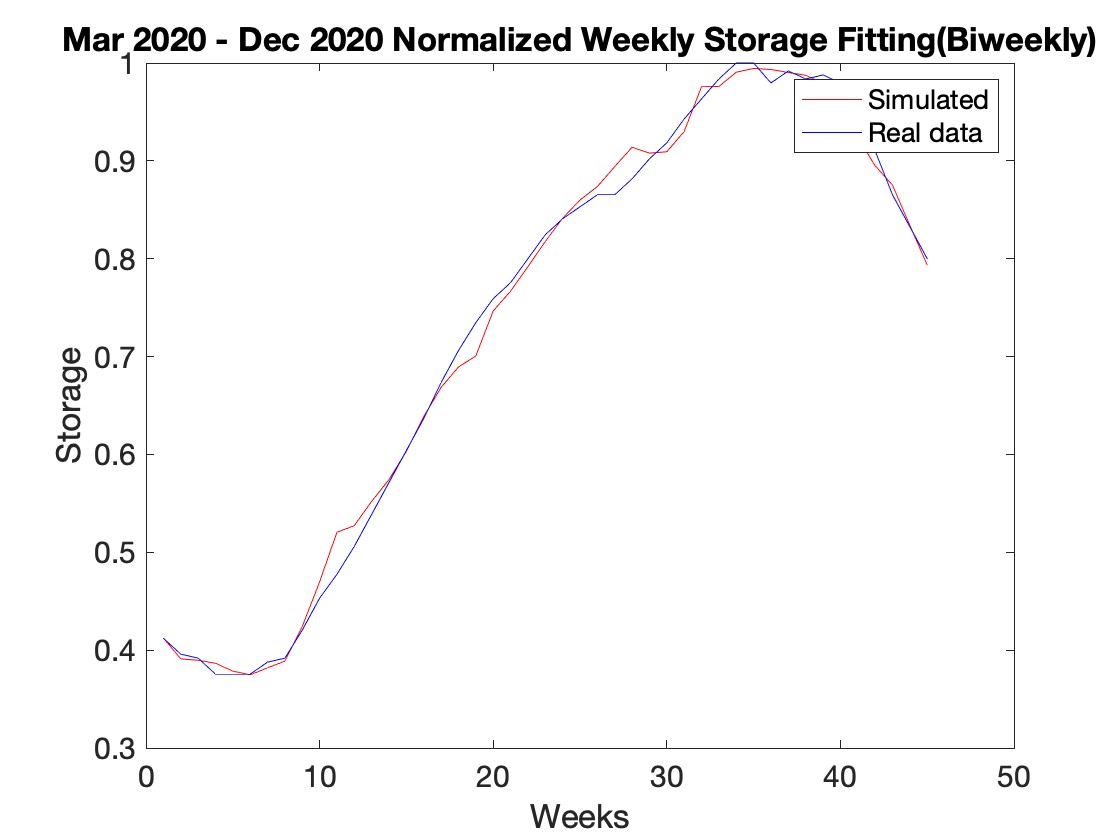}
\includegraphics[width=0.49\linewidth]{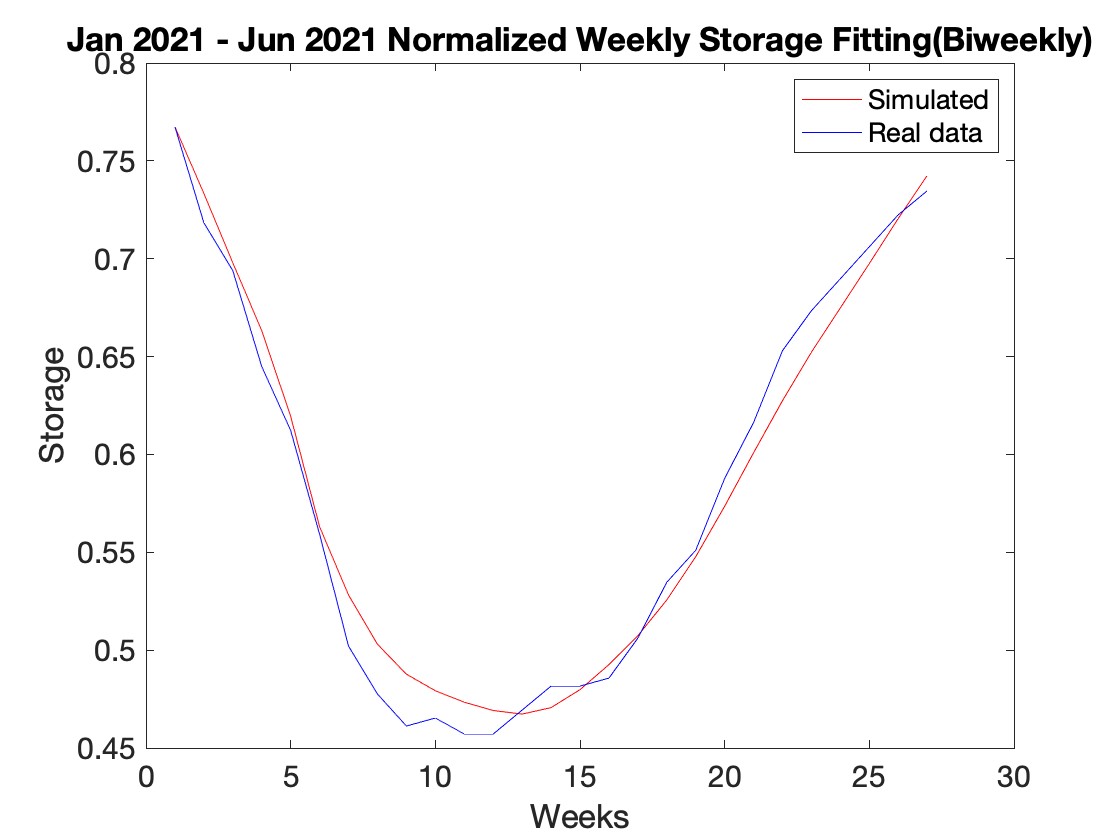}
\includegraphics[width=0.49\linewidth]{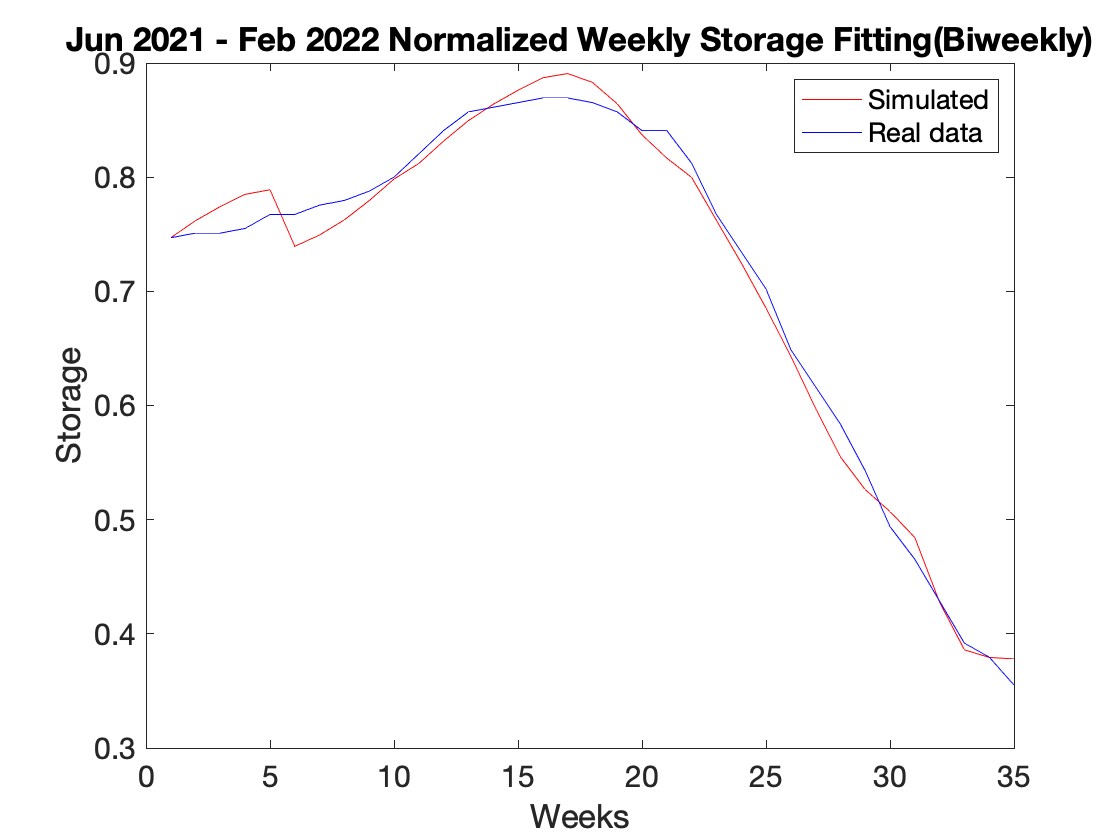}
\includegraphics[width=0.49\linewidth]{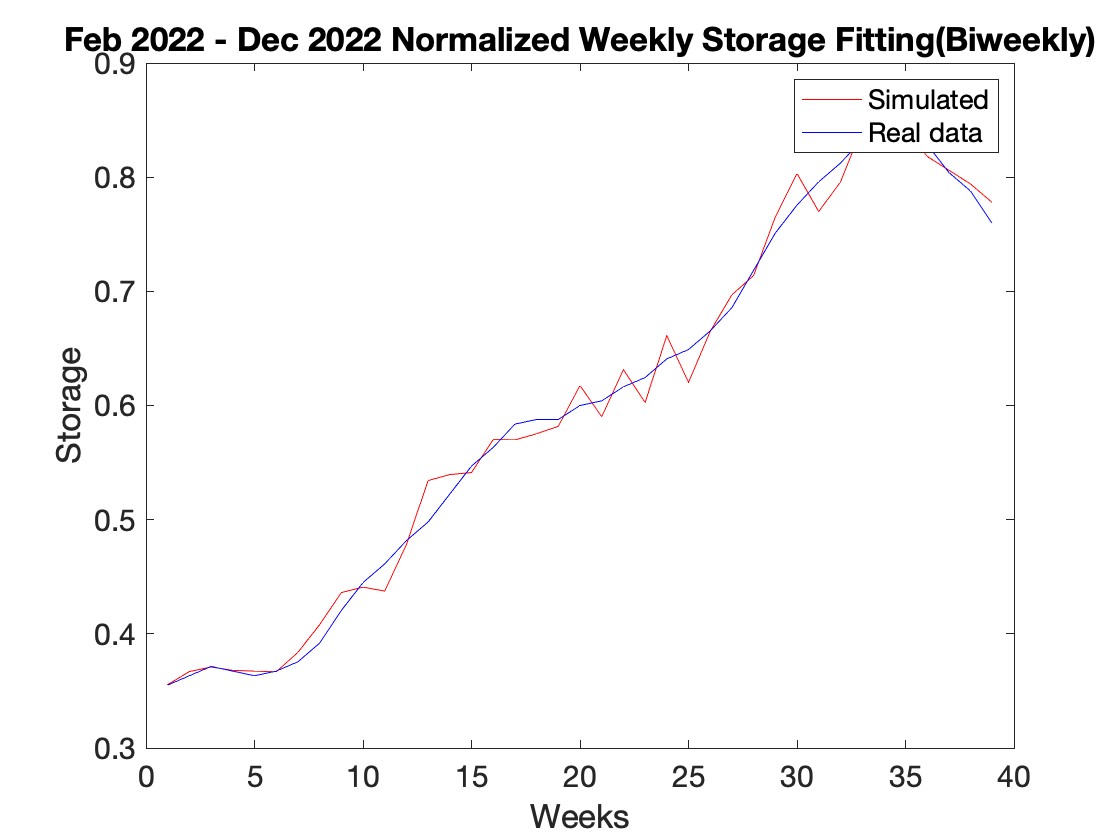}
\end{center}
\caption{Piecewise fit of normalized storage with $\mathcal{T}=14$}
\label{X_fit_14}
\end{figure}

\pagebreak

\begin{figure}[!htbp]
\begin{center}
\includegraphics[width=0.49\linewidth]{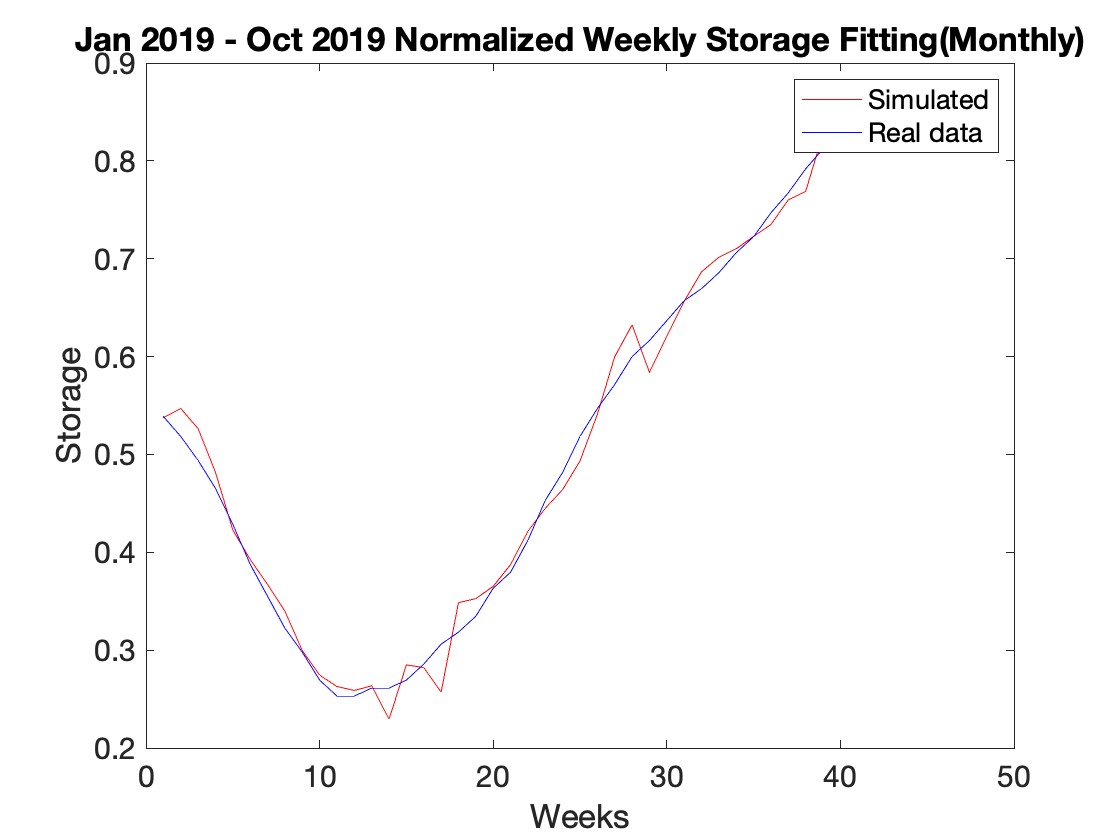}
\includegraphics[width=0.49\linewidth]{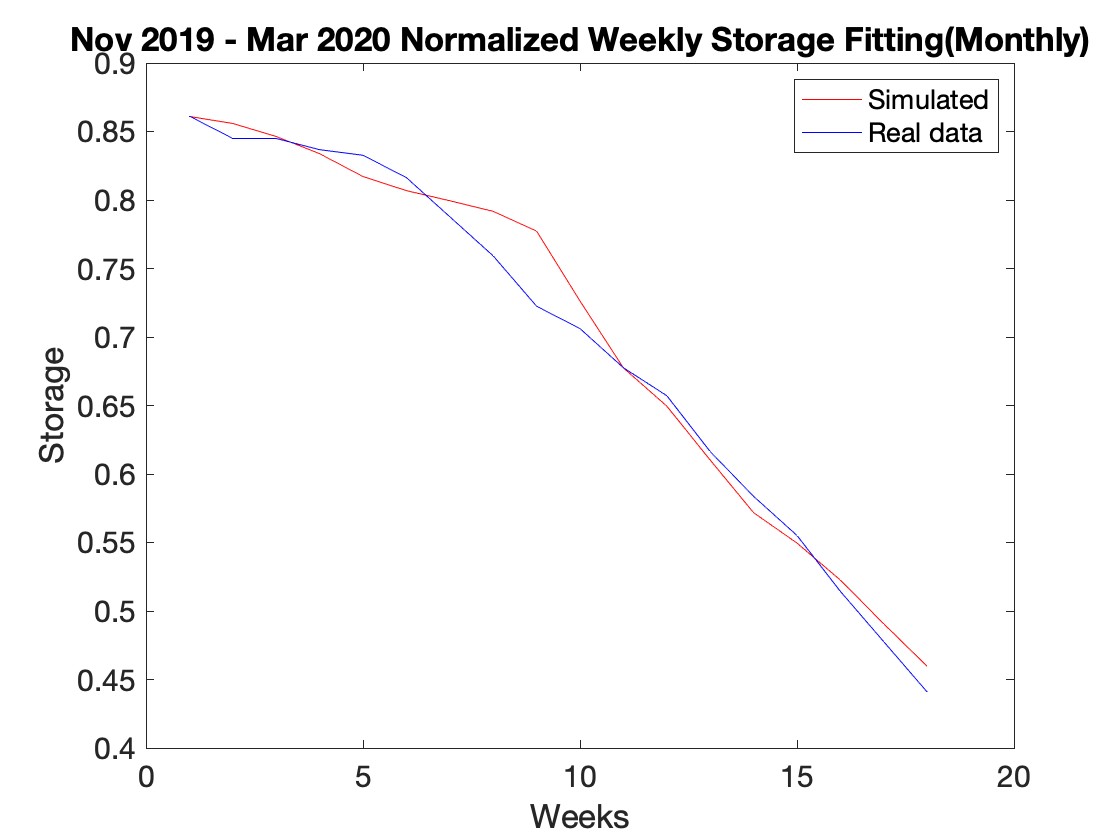}
\includegraphics[width=0.49\linewidth]{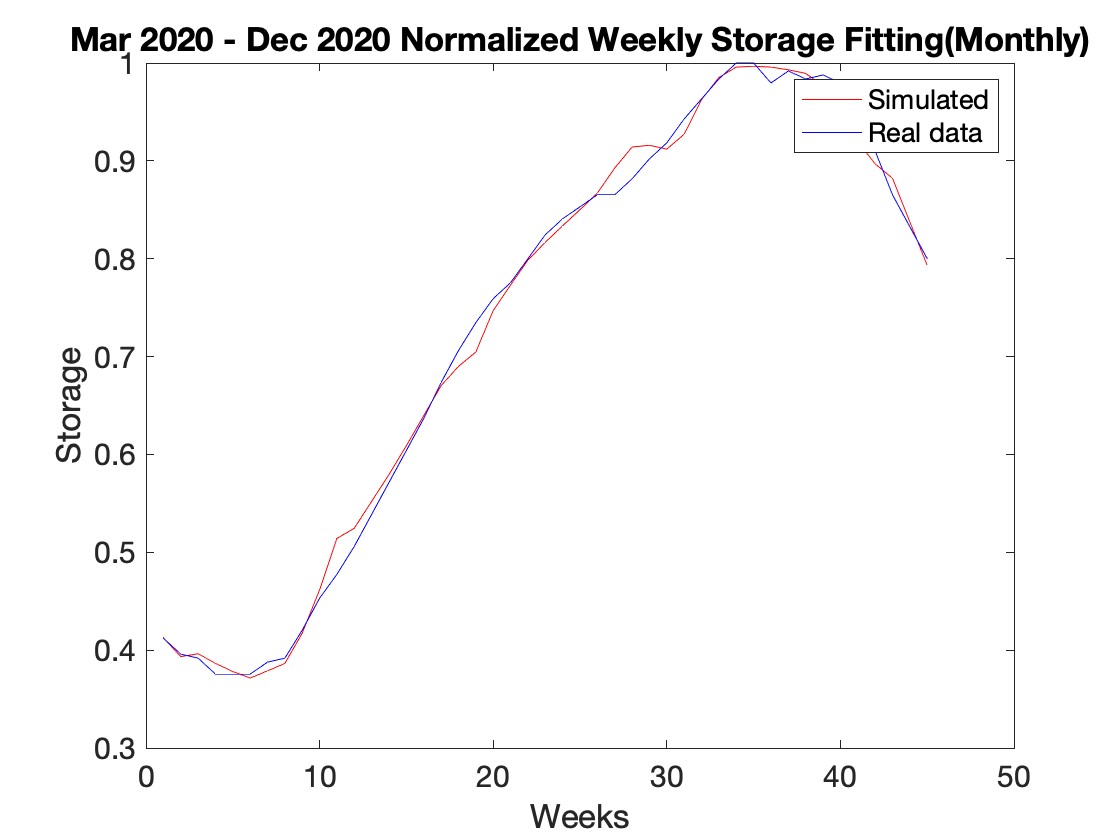}
\includegraphics[width=0.49\linewidth]{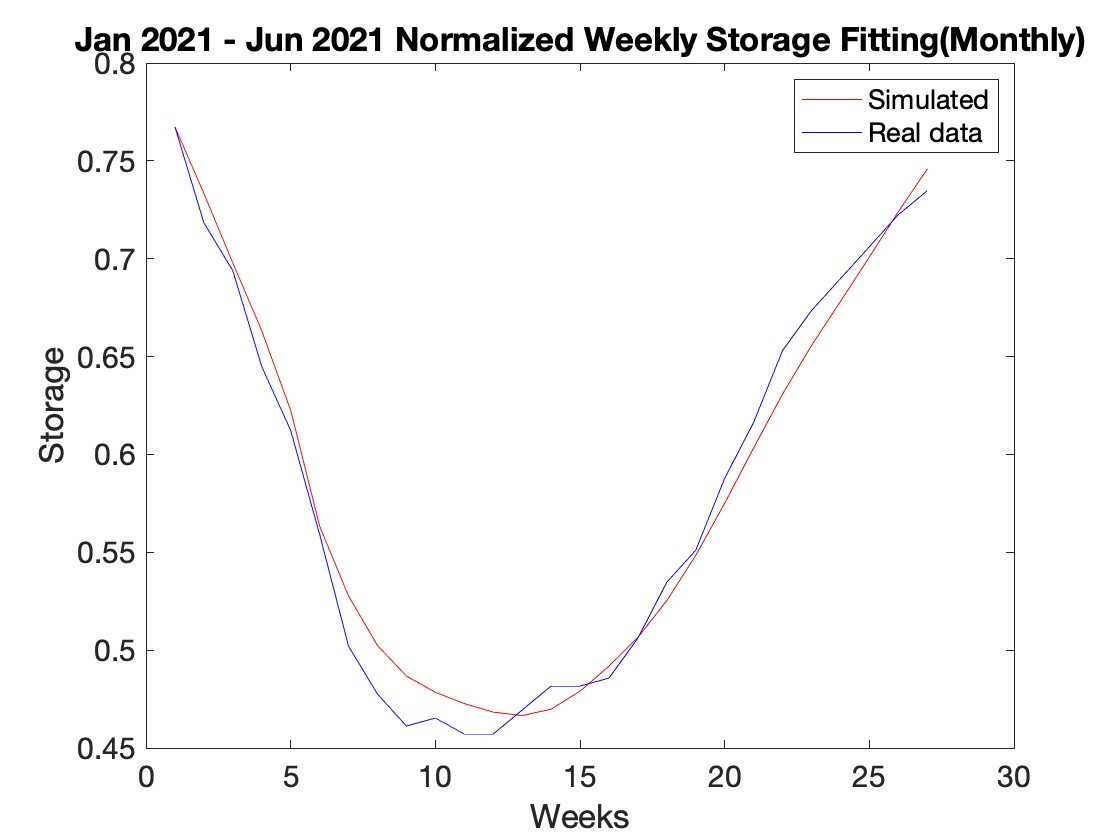}
\includegraphics[width=0.49\linewidth]{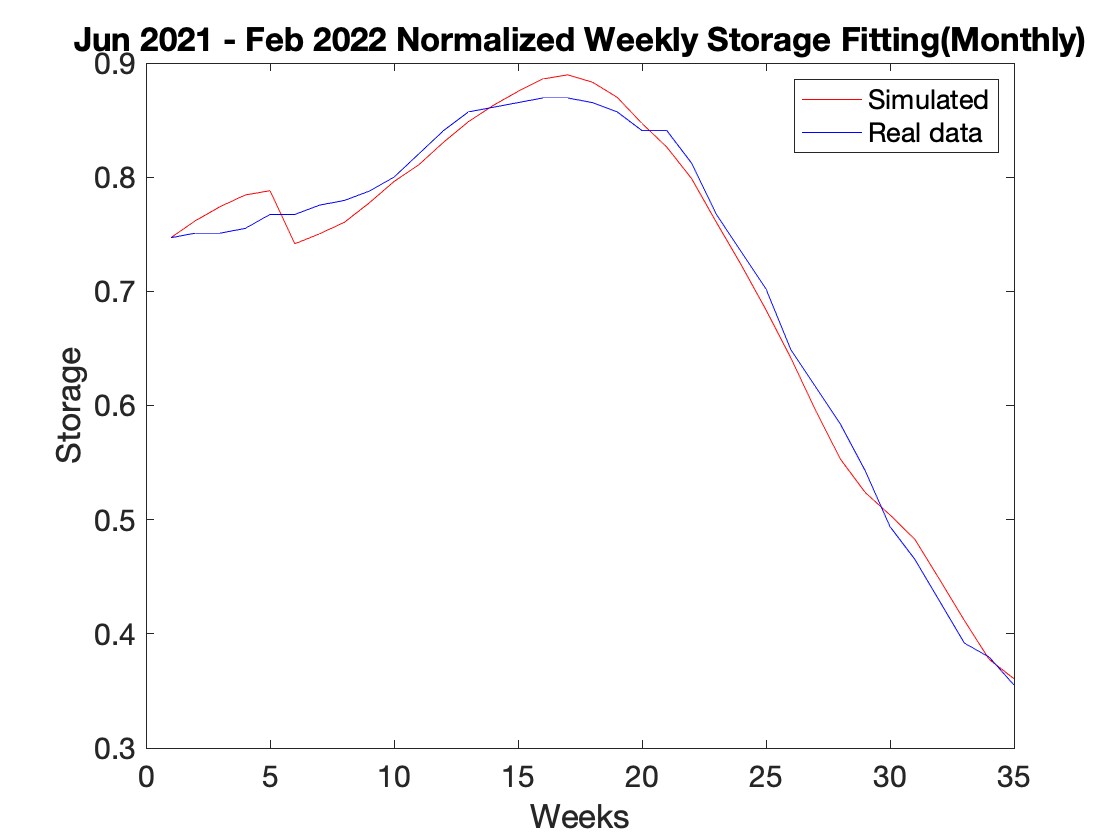}
\includegraphics[width=0.49\linewidth]{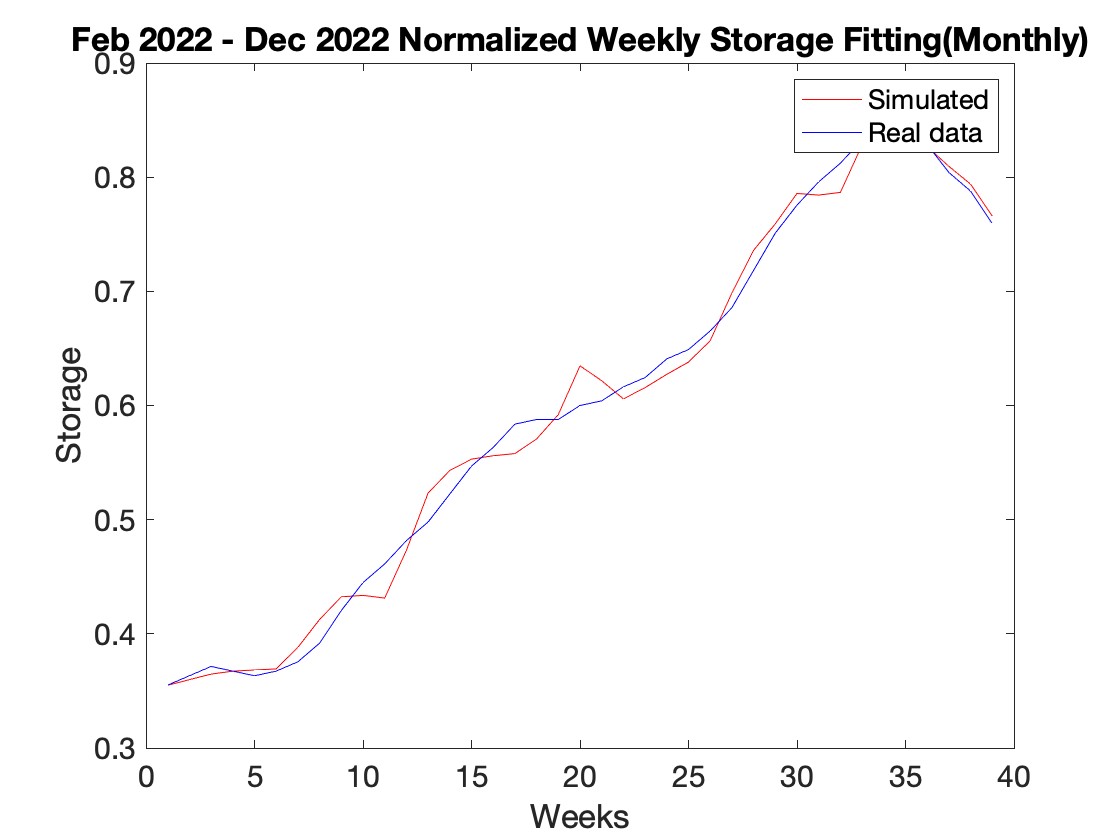}
\end{center}
\caption{Piecewise fit of normalized storage with $\mathcal{T}=30$}
\label{X_fit_30}
\end{figure}
 


\pagebreak

\bibliographystyle{siam}
\bibliography{ref_Yang}

\end{document}